\documentclass[a4paper]{article}
\usepackage{amsfonts,amssymb,amsmath,amsthm,cite}
\usepackage{ucs} 
\usepackage[utf8x]{inputenc}
\usepackage{graphicx}
\usepackage[english]{babel}
\usepackage{slashed}
\usepackage{textcomp}
\usepackage{bm}
\usepackage{titlesec}
\usepackage{stackengine}

\usepackage{etoolbox}
\patchcmd{\thebibliography}{\section*}{\section}{}{}
\usepackage{adjustbox}
\usepackage{nccmath}
\usepackage{mathtools}
\usepackage{colortbl}
\mathtoolsset{showonlyrefs}
\usepackage{hyperref}
\usepackage{calligra} 
\hypersetup{
	colorlinks,
	citecolor=red,
	filecolor=black,
	linkcolor=red,
	urlcolor=black
}
\evensidemargin=\oddsidemargin
\newtheorem{theorem}{Theorem}

\begin{document}
\vspace{20mm}
\begin{center}
	\large{\textbf{Renormalization aspects of the Yang--Mills theory with a cutoff}}
\end{center}
\vspace{2mm}
\begin{center}
	\large{\textbf{A. V. Ivanov${}^{\dagger}$~~~N. V. Kharuk${}^{2^5}$}}
\end{center}
\begin{center}
	St. Petersburg Department of Steklov Mathematical Institute of Russian Academy of Sciences,\\ 
	27 Fontanka, St. Petersburg 191023, Russia
\end{center}

\begin{center}
	${}^{\dagger}$E-mail: regul1@mail.ru\,\,\,
	${}^{2^5}$E-mail: natakharuk@mail.ru
\end{center}

\vspace{10mm}
\begin{flushright}
\large{\textbf{\textit{On the 10th anniversary of "Scenario"\footnote{See works \cite{29-1-1,29-1-0} in the list of references.}}}}
\end{flushright}

\vspace{10mm}

\textbf{Abstract.} The paper discusses renormalization aspects of the quantum four-dimensional Yang–Mills theory with a cutoff regularization in the coordinate representation. The background field method is used to formulate a generating functional, and the regularization is introduced through quasi-local probabilistic averaging. Two main types of regularization are proposed: strong deformation, which consists in averaging fluctuation fields, and weak deformation, which is a covariant generalization of the first case with respect to gauge transformations of the background field. We study singular contributions for the first two quantum corrections in this paper and compare them in detail with the case of dimensional regularization. The consistency of the action and the equation of motion after introducing the regularization and making a renormalization procedure is analyzed. New counter-vertices are studied, in particular their locality properties and dependence on the regularization parameter.

\vspace{2mm}

\newpage		
	
\tableofcontents

\section{Introduction}
\label{ym:sec:int}
The concept of "divergence" is inextricably linked to most quantum field models studied using perturbative methods \cite{3,9}, that is, using an expansion over some small parameter. This concept arose during attempts to construct quantum electrodynamics around 80 years ago \cite{ya-1,10} and at the moment, without losing its relevance, it can be found\footnote{We are talking about singular functionals and issues related to their definitions.} in many problems in theoretical and mathematical physics. For example, one can mention the classical theory of generalized functions \cite{Gelfand-1964,Vladimirov-2002}, problems of functorial quantum field theory \cite{sk-14,sk-16,sk-5,sk-7}, the theory of representations of groups and Lie algebras \cite{ya-2,ya-3}, special functions \cite{ya-4,ya-5}, questions of constructing  functional integrals \cite{sk-1,sk-2,34-6,34-c-m}, as well as integrable models \cite{ya-6}.

According to the standard theory, see \cite{6,7,105}, divergences must be regularized, that is, the classical action of the model must be deformed by introducing an auxiliary parameter\footnote{For the sake of certainty, we notate the regularization parameter with the symbol $\Lambda$, since this designation is standard when using a cutoff. The limit of regularization removal is reached by passing $\Lambda\to +\infty$.} $\Lambda$ in such a way that the divergences turn into singular functions relative to $\Lambda$. Such functions are commonly referred to as singularities. It turns out that the perturbative expansions for some theories contain singularities that obey a set of additional relations that allow them to be reduced, or subtracted, by multiplying the parameters and fields of the theory by special constants, see for example \cite{33-rev7}. Such theories are called renormalizable, and the process of finding coefficients for the constants is called the process of multiplicative renormalization.

A class of appropriate regularizations used in practice is not very wide, since it is limited not only by the desire to preserve certain internal symmetries of the theory, but also by the ability to conduct multi-loop\footnote{Correction terms of a high order.} calculations. Among the most popular options, the following approaches can be noted: dimensional regularization \cite{19,555}, higher covariant derivative regularization \cite{Bakeyev-Slavnov,29-st,AA-1,AA-2}, implicit \cite{chi-0,chi-1,chi-2}, Feynman regularization \cite{FF-1,Bog-R} and Pauli--Villars one \cite{Pauli-Villars}, regularization with a cutoff \cite{w6,w7,ww7,w8,Khar-2020,sk-b-19,Sh}. All of these methods have both advantages and some disadvantages. For example, the dimensional regularization is the main tool for multi-loop calculations, although it is built through dimensional deformation, which leads to a number of unanswered questions. In turn, the higher derivative regularization has found its application in supersymmetric theories. However, higher-order operators are less well understood, and for such situations well-posedness of the formulation is not always obvious, not to mention the spectral properties. Further, although some cutoff regularization is transparent in terms of construction and meaning, it can violate important internal symmetries in models, for example, the gauge one.

Unfortunately, the possibility of carrying out the renormalization procedure may depend not only on the choice of a specific model, but also on the type of regularization. As an example, we can give a two-dimensional principal chiral field model, see \cite{sig1,Ivanov-Akac,AIK-25,i-2626,i-2626-1}. It is renormalizable in the standard sense in the case of dimensional regularization, however, when using a cutoff, it faces the need to introduce auxiliary counter-vertices. This is primarily due to the fact that the $\delta$-functional is being replaced by a smooth function, which no longer allows subtracting all subsingularities using the standard $\mathcal{R}$-operation, see similar discussions on this topic in \cite{33-rev8,33-rev9}. This state of affairs leads either to an extension (generalization) of the classical action, as, for example, was done in the theory with the Yukawa interaction, see \cite{33-rev4,33-rev5,33-rev6} or problem 10.2 in \cite{10}, or to an extension\footnote{That is, to weaken the renormalization conditions. For example, this may consist in allowing the addition of a specific class of counter-vertices.} of concepts of renormalizability. Other popular examples of renormalizable theories within the framework of using dimensional regularization are the four-dimensional Yang--Mills quantum model \cite{1,5}, as well as various scalar models\footnote{Here the lower index indicates the dimension of space, and the upper index indicates the maximum degree of the field in the interaction.} $\{\phi_3^6,\phi_4^4,\phi_6^3\}$, see \cite{29-3,29-4}.

This paper is devoted to the study of a special cutoff regularization in the coordinate representation and its covariant generalization using the example of the four-dimensional Yang--Mills theory, see also \cite{8,4}. This approach has been incorporated in the work of \cite{34} and has since been significantly developed and improved. In addition to a number of mathematical features, such as the spectral decomposition \cite{Ivanov-2022}, relationships with an averaging operator \cite{Iv-2024,sk-b-20,ya-10} and consistency regarding the process of gluing statistical sums \cite{sksk}, new properties were studied related to an application in specific models. For example, four-loop corrections have been calculated for $\phi_3^6$, see \cite{Kh-2024,Kh-25}, three-loop corrections for $\phi_4^4$ and $\phi_5^3$, see \cite{Iv-2024-1,Iv-Kh-2024}, two-loop\footnote{As well as some parts from the third correction, see \cite{Ivanov-Kharuk-2023}.} for the four-dimensional Yang--Mills theory, see \cite{Ivanov-Kharuk-2020,Ivanov-Kharuk-20222}, and the three-loop structure of singularities for the two-dimensional principal chiral field model has been studied, see \cite{AIK-25}. It is worth noting that unlike the sigma model, where additional nonlocal singularities led to the need to extend the concept of renormalization, scalar models fully fit into the standard paradigm. This state of affairs led to a desire not only to study in more detail the dependence of the quantum action of the four-dimensional Yang--Mills theory on the cutoff in the coordinate representation, but also to generalize the regularization, making it covariant with respect to the gauge transformations of the background field.

Here it is worth making an important remark about the history of the development of a more general class of regularizations (deformations). The fact is that the introduction of the averaging operator in the main order can be reformulated either by applying some operator function $F(\cdot)$, in the argument of which there is a Laplace operator, or by multiplying the spectrum by a special regularizing function $\rho(\cdot)$. At the same time, the family of such functions has a rather specific form, as shown in \cite{Ivanov-2022}. Nevertheless, the class of functions can be significantly expanded. In this case, quasi-locality will be lost, and the theory will become nonlocal. Examples include "sharp" cutoffs in the momentum representation, which arose when studying the functional renormalization group \cite{nkf-1,nkf-2,nkf-3}, as well as various types of exponential functions, see for example \cite{ym-zz-1,ym-zz-2}. Note that the approach using nonlocal theories is used in the study of various aspects of "quantum" gravity \cite{ym-zz-6,ym-zz-7} and stochastic quantum mechanics \cite{ym-zz-5}. At the same time, its systematic study was laid down quite a long time ago, see for example \cite{ym-zz-3,ym-zz-4}. Nevertheless, quasi-local deformations, which are the key object of this work, have not been studied in the context of renormalization before.

About ten years ago, two papers were published, see \cite{29-1-1,29-1-0}, devoted to the renormalization scenario of the four-dimensional quantum Yang--Mills theory. They considered a "combinatorics" of logarithmic singularities\footnote{Here $\Lambda\gg1$ is a dimensional regularization parameter, and $\mu>0$ is an auxiliary finite fixed parameter. The value of $L$ is an analog of $1/\varepsilon$ in the case of dimensional regularization.} $L=\ln(\Lambda/\mu)$ in all quantum corrections and the corresponding solution of the Gell-Mann--Low equation, see for reference \cite{ya-11}. Since there was no suitable cutoff regularization at that time, a variant for the case with dimensional regularization was used as a basis for the distribution of singularities, see \cite{ya-12}, Section 6 in \cite{105}, as well as some issues of summation in \cite{o-3,o-4}. It is important to note that earlier attempts were made to calculate two-loop quantum corrections in the case of cutoff, see \cite{w5,w10}, as well as issues related to the violation of gauge symmetry\footnote{Including the Slavnov--Taylor (Ward--Takahashi) identities, see \cite{g1,g2,g3,g4}.}, see \cite{w6,w7,w8,w9}, and options for its restoration \cite{w3,w4}. Unfortunately, the task of generalizing to the case of a large number of loops seemed insurmountable, just as their connection with background fields and the quantum equation of motion was not clear. It should also be noted that the approach to solving recurrence relations used in \cite{29-1-1,ya-12} is a very useful tool and is found in a number of recent works \cite{o-5,o-6,o-7} devoted to the study of the renormalization group.

This work is a continuation of two articles \cite{Ivanov-Kharuk-2020,Ivanov-Kharuk-20222}, which, in turn, were considered by the authors as an addition to the "scenario" with a number of important examples on the introduction of cutoff regularizations and a set of explicit calculations. They analyze the two-loop approximation for quantum action and the one-loop approximation for the quantum equation of motion for a family of cutoff regularizations in the framework of using the background field method, see \cite{102,103,24,25,26}. A detailed comparison is also performed with known similar calculations for dimensional regularization of \cite{12}. Nevertheless, these two papers did not cover important issues related to the interaction of the renormalization processes of the action and equation, covariance with respect to gauge transformations of the background field, the structure of singularities and counter-vertices, as well as the choice of ansatz for Green's functions in multi-loop calculations. These are the issues that this article is dedicated to. The following points can be listed as the main results.\\

\noindent$\triangleright$ A variant of the "strong" deformation of the Yang--Mills theory is proposed, which preserves the functional relationship between the effective action and the quantum equation of motion. Within the framework of this approach, the following tasks were solved.
\begin{enumerate}
	\item The process of renormalizing the action near the diagonal\footnote{When the field responsible for fixing the gauge condition matches the background one.} is described in the framework of Faddeev's approach\footnote{That is, when the background field solving the quantum equation of motion coincides with the field from the gauge condition. For reference, see \cite{ya-21}, the arXiv version for \cite{23}, or Section \ref{ym:sec:pr:gen3-3} below.}.
	\item The first two coefficients of the $\beta$-function are calculated.
	\item The power-law singularities in the first two corrections are calculated.
	\item The singularities in the first loop are calculated for the first variation of the quantum action with respect to the "gauge condition" on the diagonal.
	\item Counter-vertices are found in the first two quantum corrections within the framework of an extended renormalization process allowing the appearance of terms with power singularities.
	\item The splitting process\footnote{In this case, it means that some parts of the classical action within the framework of the strong deformation may acquire an additional renormalization constant.} for the classic action is shown at the level of the second correction.
\end{enumerate}

\noindent$\triangleright$ A variant of the "weak" deformation of the Yang--Mills theory is proposed, which is covariant with respect to the gauge transformations of the background field. Within the framework of this approach, the following tasks were solved.
\begin{enumerate}
	\item The process of renormalization of quantum action in the framework of Faddeev's approach and the dependence on an auxiliary field included in the deformation operator are described.
	\item The first two coefficients of the $\beta$-function are calculated.
	\item The absence of local terms with power singularities is shown.
	\item A single counter-vertex has been found in the framework of the extended renormalization process.
	\item At the level of the first correction, it is shown that the singular parts do not depend on the field included in the deformation operator. The singularities for the first variation on the diagonal are also calculated.
	\item A comparison with the dimensional regularization is carried out: master-integrals are determined and it is shown that in the case of cutoff, the integrals themselves are deformed, while in the case of a change in dimension, only the coefficients with which these integrals appear are deformed.
	\item A comparison with the strong deformation is carried out.
\end{enumerate}
\noindent$\triangleright$ A variant of extending the classical action by adding a mass parameter is considered. The renormalization coefficients are calculated for the "weak" case in two loops and for the "strong" case in the first correction.\\

\noindent$\triangleright$ At the level of the second correction, the option of introducing quasi-local vertices, as well as their effect on the coefficients of the $\beta$-function, has been studied.\\ 

\noindent$\triangleright$ A special approach is proposed, which consists in using a "weak" deformation of the quantum action on the diagonal without studying the remaining variations. The issue of multi-loop calculations and the construction of an action outside the diagonal is discussed.\\

The work is organized as follows. \textbf{Section \ref{ym:sec:pr:gen}} contains the main definitions related to classical and quantum actions, perturbative decompositions, definitions of vertices and operators, as well as the choice of a gauge condition. In \textbf{Section \ref{ym:sec:pr:gen2}}, the diagrammatic representation for the quantum action is analyzed, elementary blocks are defined, as well as connection properties for diagrams. In \textbf{Section \ref{ym:sec:pr:gen3}}, the strong deformation of the Yang--Mills theory is discussed. Special attention is paid not only to the process of introducing the regularization, but also to the connection between the quantum action and the quantum equation of motion. In the same context, Faddeev's approach to the choice of the background field is considered, as well as the process of renormalization, taking into account the preservation of the relationship between the action and equation. The problem statement and results are described in separate subsections. \textbf{Section \ref{ym:sec:re}} has a similar structure, but is devoted to the weak deformation, which can be considered as a covariant generalization of the strong case. At the same time, \textbf{Section \ref{ym:sec:re-6}} is devoted to discussing a special approach that is attractive from the point of view of multi-loop calculations. \textbf{Sections \ref{ym:sec:one}}, \textbf{\ref{ym:sec:two}}, and \textbf{\ref{ym:sec:two1}} contain one- and two-loop calculations for the presented results. Then, in \textbf{Section \ref{ym:sec:mas}}, the issue of extending the classical action by adding a mass term, as well as its renormalization depending on the type of deformation, is discussed. Next, in \textbf{Section \ref{ym:sec:quas}}, the pros and cons of the appearance of quasi-local vertices are discussed, as well as their relationship to the obtained renormalization coefficients. \textbf{Section \ref{ym:sec:zakl}} contains a summary of the results, comments, open questions, and acknowledgements.

\section{General definitions}
\label{ym:sec:pr:gen}
In this paper, we study the Euclidean version of the Yang--Mills theory \cite{3} in the flat four-dimensional space $\mathbb{R}^4$. The elements of such a space are notated by the letters $\{x,y,z\}$, and their individual components are highlighted by the Greek indices\footnote{The paper uses the Einstein convention, which consists in automatic summation over repeated indices.}. For example, $x^\mu$ or $x_\mu$, where $\mu$ of $\{1,2,3,4\}$. It is clear that in this case the metric tensor is reproduced by the Kronecker symbol $\delta^{\mu\nu}$, therefore $x^\mu=x_\mu$, and the location of the indices will not always be tracked. 

Let $G$ denote a compact semisimple Lie group, and the symbol $\mathfrak{g}$ represent the corresponding Lie algebra, see \cite{2} for reference. In this case, the symbols $t^a$, where $a\in\{1,\ldots,\dim\frak{g}\}$, denote the generators of the algebra $\mathfrak{g}$. Without loss of generality, we assume that they satisfy the relations
\begin{equation}\label{ya-a-1}
[t^a,t^b]=f^{abc}t^c,\,\,\,
\mathrm{tr}(t^at^b)=-2\delta^{ab}.
\end{equation}
Here, square brackets denote a commutator, $\mathrm{tr}(\cdot)$ is a trace operation, and the coefficients $f^{abc}$ are completely antisymmetric real structure constants that satisfy\footnote{In fact, the structure constants can be chosen as new generators of the algebra $\mathfrak{g}$. This representation is called adjoint.} the normalization condition and the Jacobi identity in the form
\begin{equation}\label{ya-a-2}
f^{abc}f^{abe}=c_2\delta^{ce},\,\,\,
f^{acb}f^{bed}f^{dga}=-\frac{c_2}{2}f^{ceg}.
\end{equation}
The constant $c_2$ is an eigenvalue of the Casimir operator and depends on the choice of the group $G$. For example, for the group $\mathrm{SU}(n)$ the value of $c_2=n$.

A smooth Yang--Mills field is a set $\{A_\mu^{\phantom{1}}=A_\mu^at^a\}_{\mu=1}^4$ of Lie algebra elements $\mathfrak{g}$, coefficients $A_\mu^a$ of which belong to $C^{\infty}(\mathbb{R}^4,\mathbb{R})$. In this case, the classical action for the Euclidean version of the Yang--Mills theory is determined by the equation \footnote{Here it is assumed that the fields decrease at infinity quickly enough, so that the integral converges.}
\begin{equation}\label{ya-a-3}
S_{\mathrm{cl}}[A]=\frac{1}{4g^2}
\int_{\mathbb{R}^4}\mathrm{d}^4x\,
F_{\mu\nu}^a[A]F_{\mu\nu}^a[A],
\end{equation}
where the curvature (field strength) tensor, which for each fixed indices $\mu$ and $\nu$ forms an element $F_{\mu\nu}^{\phantom{1}}[A]=F_{\mu\nu}^a[A]t^a$ of the Lie algebra $\mathfrak{g}$, defined in local coordinates by the relation
\begin{equation}\label{ya-a-4}
F_{\mu\nu}^a[A](x)=
\partial_{x^\mu}^{\phantom{1}}A_\nu^a(x)-
\partial_{x^\nu}^{\phantom{1}}A_\mu^a(x)+
f^{abc}A_\mu^b(x)A_\nu^c(x).
\end{equation}

The main method of transition to perturbative decomposition in this work is the background field method\footnote{See also \cite{ym-d-1,ym-d-2,I-R} for an example of using the method in the case of 0D models and multidimensional integrals.} \cite{102,103,24,25,26}, which in general terms consists in decomposing into a background field $B_\mu^a$ and a fluctuation field $a_\mu^a$ in the form $A_\mu^a=B_\mu^a+ga_\mu^a$. This substitution leads to the decomposition of the classical action into the sum of a finite number of terms
\begin{align}\label{ya-a-5}
S_{\mathrm{cl}}[B+ga]=\frac{W_{-1}}{4g^2}&+
\frac{1}{g}\Gamma_1[a]+\frac{1}{2}
\int_{\mathbb{R}^4}\mathrm{d}^4x\,
a_\mu^aM_{1\mu\nu}^{\,\,ab}a_\nu^b\\&+
g\Gamma_3[a]+\frac{g^2}{4}\Gamma_4[a]-
\frac{1}{2}\int_{\mathbb{R}^4}\mathrm{d}^4x\,
\big(D_\mu^{ae}a_\mu^e\big)
\big(D_\nu^{ab}a_\nu^b\big),
\end{align}
where, for convenience, the notations $W_{-1}\equiv W_{-1}[B]=4g^2S_{\mathrm{cl}}[B]$ and $F_{\mu\nu}^a\equiv F_{\mu\nu}^a[B]$ were used, the covariant derivative $D_\mu^{ab}$ is defined, which in local coordinates takes the form
\begin{equation}\label{ya-a-6}
D_\mu^{ab}(x)=\partial_{x^\mu}^{\phantom{1}}\delta^{ab}+f^{acb}B_\mu^c(x),
\end{equation}
as well as the following auxiliary functionals
\begin{align}\label{ya-a-7}
\Gamma_1[a]&=-\int_{\mathbb{R}^4}\mathrm{d}^4x\,
a_\nu^aD_\mu^{ab}F_{\mu\nu}^b,	
	\\\label{ya-a-8}
\Gamma_3[a]&=\int_{\mathbb{R}^4}\mathrm{d}^4x\,
\big(D_\mu^{ae}a_\nu^e\big)f^{abc}a_\mu^b
a_\nu^c,\\\label{ya-a-9}
\Gamma_4[a]&=\int_{\mathbb{R}^4}\mathrm{d}^4x\,
f^{abc}a_\mu^ba_\nu^cf^{aed}a_\mu^e
a_\nu^d,
\end{align}
and the operators
\begin{equation}\label{ya-a-10}
M_0^{ab}=-D_\mu^{ac}D_\mu^{cb},\,\,\,
M_{1\mu\nu}^{\,\,ab}=M_0^{ab}\delta_{\mu\nu}^{\phantom{1}}-2f^{acb}F_{\mu\nu}^c
\end{equation}
have been used. It can be noted that the subscript in the $\Gamma$-functionals corresponds to the degree of the fluctuation field. Besides, $\Gamma_1$ and $\Gamma_3$ also depend on the background field $B_\mu$. 

According to the general theory \cite{3}, when moving to the quantum case, the classical Yang--Mills action should be supplemented by the $S_{\mathrm{gf}}$ term, which fixes a gauge condition, as well as the Faddeev--Popov term $S_{\mathrm{g}}$, corresponding to the ghost fields $c=c^at^a$ and $\bar{c}=\bar{c}^at^a$, see \cite{27}. To do this, we first need to determine the type of the gauge condition. Let us choose it as follows
\begin{equation}\label{ya-a-31}
\big(\partial_{x^\mu}\delta^{ac}+f^{adc}e_\mu^d(x)\big)a_\mu^c(x)=
\mathfrak{D}_\mu^a(x)a_\mu^c(x)
=0,
\end{equation}
where $e_\mu^{\phantom{1}}=e_\mu^at^a$ is an element of $\mathfrak{g}$ with smooth coefficients. Its explicit appearance is not important at the moment and will be fixed in the next sections. Then, using the representation from \cite{23}, the additions to the classical action take the form\footnote{It is important to note that in the general case, the functional $S_{\mathrm{gf}}$ must contain the $\xi$ parameter. The option $\xi=1$ is fixed here in order to work exclusively with the standard Laplace operator ($-\partial_\mu\partial_\mu+\ldots$). }
\begin{align}\label{ya-a-11}
S_{\mathrm{gf}}[a,e]&=\frac{1}{2}\int_{\mathbb{R}^4}\mathrm{d}^4x\,\big(\mathfrak{D}_\mu^{ae}a_\mu^e\big)
\big(\mathfrak{D}_\nu^{ab}a_\nu^b\big),\\\label{ya-a-12}
S_{\mathrm{g}}[B,a,c,\bar{c},e]&=\int_{\mathbb{R}^4}\mathrm{d}^4x\,\bar{c}^a\mathfrak{M}_0^{ab}c^b
+g\Omega_3,
\end{align}
where auxiliary functionals and operators have the form
\begin{equation}\label{ya-a-13}
\Omega_3[a,c,\bar{c},e]=\int_{\mathbb{R}^4}\mathrm{d}^4x\,
\big(\mathfrak{D}_\mu^{ab}\bar{c}^b\big)f^{aed}a_\mu^ec^d,
\end{equation}
\begin{equation}\label{ya-a-32}
\mathfrak{M}_0^{ab}(x)=-\mathfrak{D}_\mu^{ae}(x)D_\mu^{eb}(x),
\end{equation}
\begin{equation}\label{ya-a-33}
\mathfrak{M}_{1\mu\nu}^{\,\,ab}(x)=M_{1\mu\nu}^{\,\,ab}(x)+
D_\mu^{ae}(x)D_\nu^{eb}(x)-
\mathfrak{D}_\mu^{ae}(x)\mathfrak{D}_\nu^{eb}(x).
\end{equation}

Considering all the constructions presented above, the quantum action $W[B,e]$ for the Euclidean version of the Yang--Mills theory with the gauge condition\footnote{Clearly, the theory should not depend on the gauge condition. In particular, it depends on the choice of $e_\mu$. Nevertheless, the presence of this field will be noted in the quantum action, since after the introduction of regularization, invariance may be violated. The fact of violation of this invariance is fixed by the appearance of additional singularities. The methods of its restoration, in turn, are not discussed in this paper.} \eqref{ya-a-31} can be written symbolically using the following functional integral
\begin{equation}\label{ya-a-14}
\exp\Big(-W[B,e]\Big)=\int_{\mathcal{H}}\mathcal{D}a\mathcal{D}^\prime \bar{c}\mathcal{D}^\prime c\,\exp\Big(-S_{\mathrm{tot}}[B,e]\Big),
\end{equation}
where
\begin{equation}\label{ya-z-25}
S_{\mathrm{tot}}[B,e]\equiv
S_{\mathrm{tot}}[B,e;a,c,\bar{c},g]=S_{\mathrm{cl}}[B+ga]+S_{\mathrm{gf}}[a,e]+S_{\mathrm{g}}[B,a,c,\bar{c},e].
\end{equation}
For convenience, the integration variables will usually be omitted in the future. In formula \eqref{ya-a-14}, the symbol $\mathcal{H}$ conditionally denotes a "domain" of integration, functions from which have a specified behavior for large argument values. It is known, see \cite{3,23,Ivanov-Kharuk-2020}, that such an object can be defined as a perturbative decomposition with respect to the coupling constant by performing Gaussian integrations of polynomials. For these purposes and further formulations, inverse operators, kernels for which are Green's functions, for the mentioned Laplace operators are needed. Therefore, at the end of this section, we introduce these functions
\begin{equation}\label{ya-a-34}
M_0^{ac}(x)G_0^{cb}(x,y)=\delta^{ab}\delta(x-y),\,\,\,
\mathfrak{M}_0^{ac}(x)\mathfrak{G}_0^{cb}(x,y)=\delta^{ab}\delta(x-y),
\end{equation}
\begin{equation}\label{ya-a-35}
M_{1\mu\sigma}^{\,\,ab}(x)G_{1\sigma\nu}^{\,\,cb}(x,y)=\delta^{ab}\delta_{\mu\nu}\delta(x-y),\,\,\,
\mathfrak{M}_{1\mu\sigma}^{\,\,ab}(x)\mathfrak{G}_{1\sigma\nu}^{\,\,cb}(x,y)=\delta^{ab}\delta_{\mu\nu}\delta(x-y).
\end{equation}
It is clear that the Green's functions are uniquely fixed by choosing suitable boundary conditions, which must be found for physical reasons. Next, we assume that the task is well-posed. Note also that for $e_\mu=B_\mu$ we get $\mathfrak{M}_0=M_0$ and $\mathfrak{M}_1=M_1$, as well as $\mathfrak{G}_0=G_0$ and $\mathfrak{G}_1=G_1$.

\section{Quantum action}
\label{ym:sec:pr:gen2}
Formula \eqref{ya-a-14} is purely symbolic, since the currently available methods of functional integration, see for example \cite{34-6,sk-2}, do not cover standard quantum field models. Nevertheless, a neat mathematical formulation is still possible. For this, the presented formula must be understood as a formal series with respect to the coupling constant $g^2$. It is easy to obtain it by decomposing exponentials into a series and then calculating multidimensional Gaussian integrals, see \cite{3}.

As practice shows, it is convenient to write (encrypt) combinations of Green's functions in a compact form using diagrammatic techniques. To do this, we first define the basic elements: vertices and connecting lines. In this case, we use the notation approach proposed in \cite{13} and used in \cite{Ivanov-Kharuk-2020,Ivanov-Kharuk-20222}.
\begin{enumerate}
	\item We need to introduce four types of vertices: unary, two triple, and one quadruple. Let us define them according to the following comparison
	\begin{equation}\label{ya-a-16}
		\Gamma_1\sim{\centering\adjincludegraphics[width = 1 cm, valign=c]{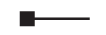}},\,\,\,
		\Gamma_3\sim{\centering\adjincludegraphics[width = 1.5 cm, valign=c]{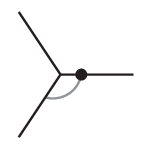}},\,\,\,
		\Omega_3\sim{\centering\adjincludegraphics[width = 1.5 cm, valign=c]{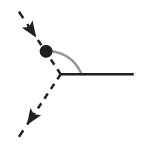}},\,\,\,
		\Gamma_4\sim{\centering\adjincludegraphics[width = 1.9 cm, valign=c]{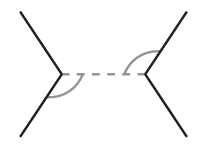}}.
	\end{equation}
Note that the images have an internal structure for easy decryption. For example, if we compare the vertex $\Gamma_3$ with formula \eqref{ya-a-8}, we can see that the round dot indicates the covariant derivative, the gray line indicates the convolution of the indices, and the order of the outer lines reproduces the order of the indices $a\to b\to c$ for the structure constant $f^{abc}$. We can decrypt $\Gamma_4$ in the same way. In the case of the vertex $\Omega_3$, arrows are additionally used to indicate the ghost fields.
	\item Connecting lines are related to the Green's functions that arise in the theory. In this case, there are two functions: the function $\mathfrak{G}_0^{bc}(x,y)$ for the operator $\mathfrak{M}_0^{ab}(x)$ and the function $\mathfrak{G}_{1\sigma\nu}^{\,\,bc}(x,y)$ for the operator $\mathfrak{M}_{1\mu\sigma}^{\,\,ab}(x)$. Their regularized versions are mapped to the lines
	\begin{equation}\label{ya-a-15}
	\mathfrak{G}_0\big|_{\mathrm{reg.}}\sim{\centering\adjincludegraphics[width = 1.5 cm, valign=c]{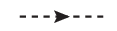}},\,\,\,
	\mathfrak{G}_1\big|_{\mathrm{reg.}}\sim{\centering\adjincludegraphics[width = 1.5 cm, valign=c]{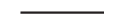}}.
	\end{equation}
It is important to keep in mind that the ends of the lines contain arguments and indices of the corresponding Green's functions. For example, the right end of the solid line corresponds to the set $\{y,\nu,c\}$. They are omitted in the drawing for convenience.
\end{enumerate}
Note that at this step, a restriction was actually introduced on the type of regularization, which ultimately boils down to the deformation of the corresponding Green's functions. This is quite enough, since the approach considered in the paper falls within the specified framework. For convenience, we assume that $\Lambda$ is a regularization parameter that can be "removed"\footnote{That is, to switch back to divergent values.} by the limit transition $\Lambda\to+\infty$. 

Next, we define an operator that maps the sum of diagrams to a set of vertices with a certain number of external lines. This approach makes it possible to compactly write down the answers for Gaussian integrals, which, as is known, ultimately boil down to the application of Wick's theorem on pairings, see \cite{Bog-R}. Let $j,i,k\in\mathbb{N}\cup\{0\}$ and $\Gamma$ be a set of a finite number of vertices defined above, with a total number of solid lines equal to $i$ and a total number of dashed lines equal to $k$. The symbol $\mathbb{H}_j^{\mathrm{c(sc)}}$ denotes an operator that performs a series of transformations with a set of $\Gamma$. If $j>i$, or at least one number from the set $\{i-j,k\}$ is odd, then the answer is zero. Otherwise, the operator performs the following three procedures:
\begin{enumerate}
	\item connects all the dashed lines ($k$ pieces) in all possible ways, while in each diagram each dashed loop is multiplied by $-1$;
	\item connects each randomly selected combination of $i-j$ external lines in all possible ways;
	\item saves only the connected (strongly connected\footnote{It is also called the one-particle irreducible (or $\mathrm{1PI}$) part.}) part. This is reflected by the superscript $\mathrm{c}(\mathrm{sc})$.
\end{enumerate}

Taking into account all of the above, the regularized quantum action of the Yang--Mills theory is defined as a formal series in the coupling constant of the form
\begin{align}\label{ya-a-17}
W_{\mathrm{reg}}^{\mathrm{c}}[B,e]=\frac{1}{4g^2}W_{-1}&+\bigg(\ln\det(\mathfrak{G}_0|_{\mathrm{reg.}})-\frac{1}{2}\ln\det(\mathfrak{G}_1|_{\mathrm{reg.}})+\kappa_0^{\mathrm{c}}\bigg)\\\label{ya-a-17-1}
&-\mathbb{H}_0^{\mathrm{c}}\bigg[\exp\bigg(-\frac{1}{g}\Gamma_1-g\Gamma_3-\frac{g^2}{4}\Gamma_4-g\Omega_3\bigg)\bigg]
+\sum_{k=1}^{+\infty}g^{2k}\kappa_k^{\mathrm{c}}.
\end{align}
Here, the values of $\kappa_i^{\mathrm{c}}$ in each order in the coupling constant remove singular terms that do not depend on the background field. This shift is done for convenience, as regularization may not cover too "strong"\footnote{The fact is that the singularity is multiplied by the integral of the density, which depends on the background field. If there is no field, then there is an integral of the constant by $\mathbb{R}^4$. This fact is a consequence of the presence of infrared divergences, which are not discussed in this paper.} divergences that are easier to "remove by hand". Note that the description of physical phenomena does not depend on shifting the action by a constant. Next, the superscript in $W_{\mathrm{reg}}^\mathrm{c}$ and $\kappa_i^{\mathrm{c}}$ is associated with the index in the operator $\mathbb{H}_0^{\mathrm{c}}$. For example, the notation $W_{\mathrm{reg}}^\mathrm{sc}$, which will appear in the next section, means the operator $\mathbb{H}_0^{\mathrm{sc}}$ instead of $\mathbb{H}_0^{\mathrm{c}}$.

Let us discuss special gauge transformations. Let $h(\cdot)\in C^{\infty}(\mathbb{R}^4,G)$ be a smooth function. In the adjoint representation, the elements of such a matrix are notated by $h^{ab}(x)$. Further, we also imply that the fields $B_\mu^a$ and $e_\mu^a$ change under gauge transformations as follows
\begin{align}\label{ya-p-1}
f^{acb}B_\mu^c(x)&\to f^{acb}B_\mu^{h,c}(x)=h^{-1,ad}(x)f^{dce}B_\mu^c(x)h^{eb}(x)+
h^{-1,ae}(x)\partial_{x^\mu}h^{eb}(x),
	\\\label{ya-p-2}
f^{acb}e_\mu^c(x)&\to f^{acb}e_\mu^{h,c}(x)=h^{-1,ad}(x)f^{dce}e_\mu^c(x)h^{eb}(x)+
h^{-1,ae}(x)\partial_{x^\mu}h^{eb}(x),
\end{align} 
while the fluctuation field from \eqref{ya-a-5} according to the rule
\begin{equation}\label{ya-p-4}
f^{acb}a_\mu^c(x)\to f^{acb}a_\mu^{h,c}(x)=h^{-1,ad}(x)f^{dce}a_\mu^c(x)h^{eb}(x).
\end{equation}
Additionally, we assume that the domain of integration of $\mathcal{H}$ from \eqref{ya-a-14} is invariant with respect to such transformations. Formally, in the absence of regularization, such invariance exists, as evidenced by the definition in terms of a perturbative series. Then it can be argued that at the formal level, in the absence of regularization, the quantum action is invariant with respect to the transformations \eqref{ya-p-1} and \eqref{ya-p-2}, which can be written as
\begin{equation}\label{ya-p-3}
	W^{\mathrm{c}}[B,e]=W^{\mathrm{c}}[B^h,e^h].
\end{equation}
As the calculations below show, the regularization can violate this invariance. This fact affects the renormalization process and the coefficients of the corresponding constants.

\section{Strong deformation}
\label{ym:sec:pr:gen3}
\subsection{Motivation}
\label{ym:sec:pr:gen3-1}
First, let us decipher the name of the section. The word "deformation" refers to the process of changing some parameters of the classical action, which eventually leads to regularization\footnote{We are talking about ultraviolet divergences that depend on the background field.} of the quantum action. The word "strong" reflects the degree of violation of internal symmetries. In this case, it symbolizes the loss of invariance under the gauge transformations of the background field $B_\mu^a$ and the field $e_\mu^a$, see formula \eqref{ya-p-3}. Nevertheless, this is done for a reason and is a price to pay for maintaining some useful functional equalities.

First of all, we need to figure out exactly what the standard regularized\footnote{The variant with the removed regularization is further indicated by the absence of $|_{\mathrm{reg.}}$. Such a series contains divergences.} quantum\footnote{In the lowest order of the coupling constant, the quantum equation of motion reproduces the classical equation. The remaining additions are called quantum corrections.} equation of motion looks like
\begin{equation}\label{ya-z-1}
Q_\mu^a|_{\mathrm{reg.}}[B_q,e](x)=0,
\end{equation}
the left part of which can be formally written out as the sum of all strongly connected diagrams with one external line, that\footnote{The factor $-g^{-1}$ is chosen here for the convenience of writing the relation with the quantum action, see \eqref{ya-a-22} or \eqref{ya-a-23}.} is
\begin{equation}\label{ya-a-18}
	\int_{\mathbb{R}^4}\mathrm{d}^4x\,a_\mu^aQ_{\mu}^a|_{\mathrm{reg.}}=
	-g^{-1}\mathbb{H}_1^{\mathrm{sc}}\bigg[\exp\bigg(-\frac{1}{g}\Gamma_1-g\Gamma_3-\frac{g^2}{4}\Gamma_4-g\Omega_3\bigg)\bigg].
\end{equation}
In \eqref{ya-z-1} and further, the symbol $B_{q,\mu}^a$ denotes the solution of the regularized quantum equation of motion. It is a special case of selecting the background field $B_\mu^a$. If the regularized quantum equation of motion is fulfilled, then only strongly connected contributions are preserved in the quantum action \eqref{ya-a-17-1}, that is, the index $\mathrm{c}\to\mathrm{sc}$ is replaced in the operator $\mathbb{H}$. Thus, it can be argued that
\begin{equation*}
\mbox{if}\,\,\, Q_\mu^a|_{\mathrm{reg.}}[B_q,e](x)=0,\,\,\,\mbox{then}\,\,\,
W_{\mathrm{reg}}^{\mathrm{c}}[B_q,e]=W_{\mathrm{reg}}^{\mathrm{sc}}[B_q,e].
\end{equation*}
This equality is well known. Indeed, it follows from the fact that, in the general case, the difference of actions can be represented as
\begin{align}\label{ya-a-46}
	W^{\mathrm{c}}_{\mathrm{reg}}[\,\cdot\,,e]-W^{\mathrm{sc}}_{\mathrm{reg}}[\,\cdot\,,e]=
	\sum_{k=2}^{+\infty}\int_{\mathbb{R}^{4\times k}}
	\mathrm{d}^4x_1\ldots\mathrm{d}^4x_k\,&
	w_{\mu_1\ldots\mu_k}^{a_1\ldots a_k}(x_1,\ldots,x_k)\times\\\nonumber
	\times &Q_{\mu_1}^{a_1}|_{\mathrm{reg.}}[\,\cdot\,,e](x_1)\times\ldots\times Q_{\mu_k}^{a_k}|_{\mathrm{reg.}}[\,\cdot\,,e](x_k),
\end{align}
each term of which is proportional to the functional $Q_\mu^a|_{\mathrm{reg.}}[\,\cdot\,,e](x)$, with summation starting with two. This is described in more detail, for example, in \cite{Vas-98}. Further, we note that at the formal level, before the introduction of regularization, there is a relation\footnote{It is more correct to perceive it as the difference between the left and right sides, when the non-integrable densities cancel each other.} between the quantum action and the density from the quantum equation of motion
\begin{equation}\label{ya-a-22}
\frac{\delta}{\delta B_\mu^a(x)}W^{\mathrm{sc}}[B,e]=Q_\mu^a[B,e](x).
\end{equation}
The violation of such a connection when regularization is introduced leads to the fact that the further renormalization process ceases to be consistent, in the sense that renormalization of the action will not guarantee the renormalization of the equation, and vice versa. Accordingly, it will be necessary to adjust the process in some way at each step (in each order in the coupling constant). Unfortunately, the possibility of such an adjustment is not clear. Thus, it is important to require that the functional relationship \eqref{ya-a-22} be preserved during the regularization process. Referring to formula \eqref{ya-a-17}, mathematically the equality can be written as follows
\begin{equation}\label{ya-a-23}
\int_{\mathbb{R}^4}\mathrm{d}^4x\,a_\mu^a\frac{\delta}{\delta B_\mu^a(x)}W_{\mathrm{reg}}^{\mathrm{sc}}[B,e]=
-g^{-1}\mathbb{H}_1^{\mathrm{sc}}\bigg[\exp\bigg(-\frac{1}{g}\Gamma_1
-g\Gamma_3-\frac{g^2}{4}\Gamma_4-g\Omega_3\bigg)\bigg].
\end{equation}
It turns out that this relation can be maintained. To do this, it should be noted that the diagrams have a block structure (they consist of vertices and lines), therefore, maintaining similar relationships for each individual element guarantees the presence of relation \eqref{ya-a-23}.

\subsection{Regularization}
\label{ym:sec:pr:gen3-2}
Before formulating specific deformation rules, we first write out the relations for unregularized Green's functions and vertices, which are necessary to prove formula \eqref{ya-a-22}. The first relation has the form
\begin{equation}\label{ya-a-24}
\int_{\mathbb{R}^4}\mathrm{d}^4x\,\phi_\mu^a\frac{\delta\Gamma_3[a]}{\delta B_\mu^a(x)}=\frac{1}{4}
\int_{\mathbb{R}^4}\mathrm{d}^4x\,\phi_\mu^a\frac{\delta\Gamma_4[a]}{\delta a_\mu^a(x)}
\end{equation}
and it allows the following interpretation. The variation of the vertex $\Gamma_3$ in the background field is equivalent to highlighting one external line of the vertex $\Gamma_4$. This means that by varying $\Gamma_3$ in a quantum action, it is possible to obtain the entire set of diagrams from the density of the quantum equation, which is associated with marking one external line of the vertex $\Gamma_4$ during the application of the operator $\mathbb{H}_1$. The remaining relations can be written out in the same way, dividing into three groups. Indeed, let the functional differentiation operator be defined as
\begin{equation}\label{ya-a-28}
D[\psi,\phi]=\int_{\mathbb{R}^4}\mathrm{d}^4x\,\phi_\mu^a\frac{\delta}{\delta\psi_\mu^a(x)},
\end{equation}
then the first set of relations includes formula \eqref{ya-a-24} and three additional equalities, which are obtained by direct differentiation of the functionals:
\begin{equation}\label{ya-a-29}
D[B,\phi]W_{-1}=
4D[a,\phi]\Gamma_1[a],
\end{equation}
\begin{equation}\label{ya-a-26}
D[B,\phi]\Gamma_4[a]=0,\,\,\,D[B,\phi]\Omega_3[a,c,\bar{c},e]=0.
\end{equation}
The next set consists of two equalities. The first of which is related to the differentiation of the Green's function for the ghost fields
\begin{equation}\label{ya-a-27}
D[B,\phi]\mathfrak{G}_0^{ab}(x,y)=-
\int_{\mathbb{R}^4}\mathrm{d}^4z\,
D_\mu^{cd}(z)\mathfrak{G}_0^{ad}(x,z)
f^{ceg}\phi_\mu^e(z)
\mathfrak{G}_0^{gb}(z,y).
\end{equation}
It corresponds to the vertex $\Omega_3$ and allows the following interpretation: the inner dashed line with the arrow is cut and the vertex $\Omega_3$ is connected to the cut point, taking into account the rules of diagramming and an additional minus sign. Note that there is only one connection option. The second equality is related to the differentiation of the solid line. It consists of six parts
\begin{align}
D[B,\phi]\mathfrak{G}_{1\sigma\rho}^{\,\,ab}(x,y)=&+
\int_{\mathbb{R}^4}\mathrm{d}^4z\,
\mathfrak{G}_{1\sigma\mu}^{\,\,ad}(x,z)
f^{deg}\phi_\nu^e(z)D_\nu^{gc}(z)
\mathfrak{G}_{1\mu\rho}^{\,\,cb}(z,y)\\
&-
\int_{\mathbb{R}^4}\mathrm{d}^4z\,
D_\nu^{cd}(z)\mathfrak{G}_{1\sigma\mu}^{\,\,ad}(x,z)
f^{ceg}\phi_\nu^e(z)
\mathfrak{G}_{1\mu\rho}^{\,\,gb}(z,y)\\
&+
\int_{\mathbb{R}^4}\mathrm{d}^4z\,
\mathfrak{G}_{1\sigma\mu}^{\,\,ad}(x,z)
f^{deg}\phi_\mu^e(z)D_\nu^{gc}(z)
\mathfrak{G}_{1\nu\rho}^{\,\,cb}(z,y)\\
&-
\int_{\mathbb{R}^4}\mathrm{d}^4z\,
D_\mu^{cd}(z)\mathfrak{G}_{1\sigma\mu}^{\,\,ad}(x,z)
f^{ceg}\phi_\nu^e(z)
\mathfrak{G}_{1\nu\rho}^{\,\,gb}(z,y)\\
&-
\int_{\mathbb{R}^4}\mathrm{d}^4z\,
\mathfrak{G}_{1\sigma\mu}^{\,\,ad}(x,z)
f^{deg}\phi_\nu^e(z)D_\mu^{gc}(z)
\mathfrak{G}_{1\nu\rho}^{\,\,cb}(z,y)\times2\\
&+
\int_{\mathbb{R}^4}\mathrm{d}^4z\,
D_\nu^{cd}(z)\mathfrak{G}_{1\sigma\mu}^{\,\,ad}(x,z)
f^{ceg}\phi_\mu^e(z)
\mathfrak{G}_{1\nu\rho}^{\,\,gb}(z,y)\times2
\end{align}
and is interpreted as an operation of cutting a solid line with further connection of the vertex $\Gamma_3$ in all possible ways and additional multiplication by $-1$. The last third set of relations follows from the differentiation of determinants. To do this, it should be noted that the determinant of an operator in the framework of this work is understood in the sense of a perturbative decomposition in degrees of potentials for a functional integral. Their explicit formulas are presented below in \eqref{ya-a-48} and \eqref{ya-a-185}. We get two equalities
\begin{equation}\label{ya-a-30}
D[B,\phi]\ln\det(\mathfrak{G}_0)=-
\int_{\mathbb{R}^4}\mathrm{d}^4z\,
D_\mu^{cd}(z)\mathfrak{G}_0^{ad}(x,z)\big|_{x=z}
f^{cea}\phi_\mu^e(z)=\mathbb{H}_1^{\mathrm{c}}(\Omega_3)\big|_{a=\phi}
\end{equation}
and
\begin{align}
D[B,\phi]\ln\det(\mathfrak{G}_1)=&+2
	\int_{\mathbb{R}^4}\mathrm{d}^4z\,
	f^{beg}\phi_\nu^e(z)D_\nu^{gc}(z)
	\mathfrak{G}_{1\mu\mu}^{\,\,cb}(z,y)\big|_{y=z}\\
	&+2
	\int_{\mathbb{R}^4}\mathrm{d}^4z\,
	f^{beg}\phi_\mu^e(z)D_\nu^{gc}(z)
	\mathfrak{G}_{1\nu\mu}^{\,\,cb}(z,y)\big|_{y=z}\\
	&-4
	\int_{\mathbb{R}^4}\mathrm{d}^4z\,
	f^{beg}\phi_\nu^e(z)D_\mu^{gc}(z)
	\mathfrak{G}_{1\nu\mu}^{\,\,cb}(z,y)\big|_{y=z}=-2\mathbb{H}_1^{\mathrm{c}}(\Gamma_3)\big|_{a=\phi}.
\end{align}
In diagrammatic language, both relations are equivalent to cutting a single-loop correction and connecting a triple vertex in all possible ways. Additionally, we note that the last equalities should be considered as the difference between the left and right sides, so that singular components in the first orders are reduced. All the above-mentioned relations are sufficient to show the validity of formula \eqref{ya-a-22}.

So, the main requirement for a possible regularization is that relation \eqref{ya-a-22} is also fulfilled for the regularized case. This can be achieved if all the functional relations presented in this section are correct for regularized objects as well. It turns out that a suitable option exists. For convenience, we formulate the rules in two equivalent ways.\\

\textbf{The first way.} Let us choose a piecewise continuous\footnote{A more wide class of kernels is presented in \cite{ya-10,Iv-2024}.} function $\omega(\cdot)$ on the half-axis $\mathbb{R}_+$ as an averaging kernel in such a way that the relations $\mathrm{supp}(\omega)\subset[0,1/2]$ and $\omega\geqslant0$ are fulfilled. Then the introduction of regularization is called the deformation of the Green's function for the free Laplace operator of the following form
\begin{equation}\label{ya-a-36}
R_0^{\phantom{1}}(x)=\frac{1}{4\pi^2|x|^2}\to R_0^\Lambda(x)=\frac{1}{4\pi^2}
\int_{\mathbb{R}^4}\mathrm{d}^4y\int_{\mathbb{R}^4}\mathrm{d}^4z\,
\frac{\omega(|y|)\omega(|z|)}{|x+y/\Lambda+z/\Lambda|^2}.
\end{equation}
Such an operator is also called quasi-local probability averaging, and the auxiliary parameter $\Lambda\gg1$ is a regularizing one. Additionally, we assume that the smoothness properties of the kernel guarantee that $R_0^\Lambda(\cdot)$ belongs to $C^2(\mathbb{R}^4)$. It was previously shown, see \cite{Ivanov-2022}, that under such a deformation, the free Green's function admits the following representation
\begin{equation}\label{ya-a-37}
R_0^\Lambda(x)=\frac{\Lambda^2\mathbf{f}\big(|x|^2\Lambda^2\big)}{4\pi^2}+\frac{1}{4\pi^2}
\begin{cases}
	\,\,\,\Lambda^2, &\mbox{for}\,\,\,|x|\leqslant1/\Lambda;\\
|x|^{-2}, &\mbox{for}\,\,\,|x|>1/\Lambda,
\end{cases}
\end{equation}
where the support of the function $\mathbf{f}(\cdot)\in C(\mathbb{R}_+)$ is contained in $[0,1]$.

Further, returning to the explicit form of the quantum action, we note that in the process of such regularization, the vertices are not deformed, while the corresponding perturbative expansions for regularized Green's functions are obtained by replacing $R_0^{\phantom{1}}(\cdot)\to R_0^{\Lambda}(\cdot)$ and are written out explicitly as follows
\begin{fleqn}
\begin{align}\label{ya-a-38}
\mathfrak{G}_{1\sigma\rho}^{\,\,ab}(x,y)=\delta^{ab}\delta_{\sigma\rho}R_0^{\Lambda}(x-y)+
\sum_{k=1}^{+\infty}\int_{\mathbb{R}^{4\times k}}\mathrm{d}^4z_1\ldots&\mathrm{d}^4z_k\,
R_0^{\Lambda}(x-z_1)\\\nonumber\times&\big(-\mathfrak{M}_{1\sigma\mu_1}^{\,\,ac_1}(z_1)-\delta^{ac_1}\delta_{\sigma\mu_1}\partial_{z_1^\nu}\partial_{z_1^\nu}\big)R_0^{\Lambda}(z_1-z_2)\times\ldots\\\nonumber\times&
\big(-\mathfrak{M}_{1\mu_k\rho}^{\,\,c_kb}(z_k)-\delta^{c_kb}\delta_{\mu_k\rho}\partial_{z_k^\nu}\partial_{z_k^\nu}\big)R_0^{\Lambda}(z_k-y),
\end{align}
\end{fleqn}
\begin{fleqn}
\begin{align}\label{ya-a-39}
\,\mathfrak{G}_{0}^{ab}(x,y)=\delta^{ab}R_0^{\Lambda}(x-y)+
\sum_{k=1}^{+\infty}\int_{\mathbb{R}^{4\times k}}\mathrm{d}^4z_1\ldots&\mathrm{d}^4z_k\,
R_0^{\Lambda}(x-z_1)\\\nonumber\times&\big(-\mathfrak{M}_{0}^{ac_1}(z_1)-\delta^{ac_1}\partial_{z_1^\nu}\partial_{z_1^\nu}\big)R_0^{\Lambda}(z_1-z_2)\times\ldots\\\nonumber\times&
\big(-\mathfrak{M}_{0}^{c_kb}(z_k)-\delta^{c_kb}\partial_{z_k^\nu}\partial_{z_k^\nu}\big)R_0^{\Lambda}(z_k-y).
\end{align}
\end{fleqn}
It can be seen from the construction that the deformation does not affect the fields, therefore, all the functional relations outlined at the beginning of the section remain valid during the transition to the regularized (deformed) objects.\\

\textbf{The second way.} As is known, the regularization of the quantum action in a perturbative presentation should be introduced by deforming the classical action, see for example the arguments from Section 3.3 in \cite{sksk}. Therefore, we reformulate the approach to regularization described above in an alternative way. It can be shown that the deformation of \eqref{ya-a-36} corresponds to the accomplishment of three procedures with the classical action, see also \eqref{ya-z-25},
\begin{equation}\label{ya-a-40}
	S_{\mathrm{cl}}[B+ga]+S_{\mathrm{gf}}[a,e]+S_{\mathrm{g}}[B,a,c,\bar{c},e].
\end{equation}
\begin{enumerate}
	\item Average all the fluctuation fields $a_\mu^b(x)$ according to the formula
	\begin{equation}\label{ya-a-41}
		a_\mu^b(x)\to a_\mu^{\Lambda,b}(x)=
		\int_{\mathbb{R}^4}\mathrm{d}^4y\,
		\omega(|y|)a_\mu^b(x+y/\Lambda).
	\end{equation}
    \item Similarly, average all ghost fields: $c^a\to c^{\Lambda,a}$ and $\bar{c}^a\to\bar{c}^{\Lambda,a}$.
	\item Subtract quadratic forms with the free Laplace operator for averaged fields and add forms without averaging instead. In other words, add the summand to the existing classical action with averaged fields
	\begin{equation}\label{ya-a-42}
	+\frac{1}{2}
	\int_{\mathbb{R}^4}\mathrm{d}^4x\,
	\Big(a_\mu^aA_0a_\mu^a-a_\mu^{\Lambda,a}A_0a_\mu^{\Lambda,a}\Big)+
	\int_{\mathbb{R}^4}\mathrm{d}^4x\,\Big(\bar{c}^aA_0c^a-\bar{c}^{\Lambda,a}A_0c^{\Lambda,a}\Big),
	\end{equation}
where $A_0(x)=-\partial_{x^\mu}\partial_{x_\mu}$.
\end{enumerate}
Thus, the regularization consists of averaging all ghost and fluctuation fields, except those in quadratic forms with the free Laplace operator $A_0$. This approach is equivalent to deforming Green's functions from the first method. This fact follows from the application of Wick's theorem on pairings. Let us give a proof using the example of fluctuation fields. For ghost fields, it is performed similarly. It is known that in the absence of regularization, pairs of fields are replaced by the corresponding Green's functions, that is
\begin{equation*}
a_\mu^a(x)a_\nu^b(y)\longrightarrow\delta^{ab}\delta_{\mu\nu}R_0^{\phantom{1}}(x-y).
\end{equation*}
In turn, after the fields are deformed, two averaging operators appear, which, after applying Wick's theorem, act on the Green's function. Thus, we get
\begin{equation*}
a_\mu^{\Lambda,a}(x)a_\nu^{\Lambda,b}(y)\longrightarrow\delta^{ab}\delta_{\mu\nu}R_0^\Lambda(x-y).
\end{equation*}
This implies the equivalence of the two methods.

\subsection{Gauge fixing}
\label{ym:sec:pr:gen3-3}
Returning to Section \ref{ym:sec:pr:gen2}, we note that each term of the perturbative decomposition for the quantum action $W^{\mathrm{c}}_{\mathrm{reg}}[B,e]$ from \eqref{ya-a-17} is a functional that depends\footnote{Indeed, the term $S_{\mathrm{gf}}$ from \eqref{ya-a-11} does not contain the background field, and therefore $B_\mu^a$ cannot be removed by simply shifting the variable, as is usually done in scalar models, see for example \cite{Iv-2024-1}.} on the background field $B_\mu^a$ and on the field from the gauge condition, which in the proposed formulation is dictated solely by the choice of the field $e_\mu^a$. Let us do some thought experiments for the regularized case. To do this, we take advantage of the fact that the background field can be chosen as $B_\mu^a\to B_{q,\mu}^a$, so that it solves the regularized quantum equation of motion $Q_\mu^a|_{\mathrm{reg.}}[B_q,e](x)=0$, see \eqref{ya-a-18}. In this case, the background field $B_q$ will at the same time be a function dependent\footnote{Regularization can violate invariance, so the presence of dependence is quite expected.} on the field $e_\mu^a$. Based on this, we have $B_q=B_q[e]$. The question arises: "Is it possible to choose the field $e_\mu^a=\hat{e}_\mu^a$ from the gauge condition in such a way that the equality $B_q[\hat{e}]=\hat{e}$ is satisfied?"

It is impossible to give a mathematically accurate answer to this question, since practically nothing is known about the methods of finding a solution to the quantum equation of motion, about its uniqueness, or about the properties of smoothness. Nevertheless, there is a hypothetically possible procedure for constructing such a solution. It was proposed by Faddeev and can be found in a short version\footnote{See the latest version on arXiv for \cite{23}.} in \cite{ya-21,23}. Let us describe it in stages.
\begin{enumerate}
	\item Consider the density functional $Q_\mu^a|_{\mathrm{reg.}}[B,e](x)$. It depends on the fields $B$ and $e$.
	\item Substitute $e=B$ and solve the equation $Q_\mu^a|_{\mathrm{reg.}}[B,B](x)=0$.
	\item Get some solution $B_{f}$.
	\item Start anew the procedure for considering the quantum action $W^{\mathrm{c}}_{\mathrm{reg}}[B,e]$ for $e=B_f$.
	\item By construction, $B_f$ is a solution to the quantum equation of motion, that is
	\begin{equation*}
\frac{\delta W^{\mathrm{sc}}_{\mathrm{reg}}[B,B_f]}{\delta B_\mu^a(x)}\Bigg|_{B=B_f}=0.
	\end{equation*}
\end{enumerate}
Thus, the strongly connected quantum action depends on $B_f$ and is equal to $W^{\mathrm{sc}}_{\mathrm{reg}}[B_f,B_f]$. However, new issues arise, in particular, related to the possibility of carrying out the renormalization procedure. Indeed, previously, the quantum action and the equation were related by relation \eqref{ya-a-23}, which made it possible to work with only one object, the second one was renormalized automatically. In the new version, the situation has changed because
\begin{align}\label{ya-z-16}
\frac{\delta W^{\mathrm{sc}}_{\mathrm{reg}}[B,B]}{\delta B_{\mu}^a(x)}&=
\frac{\delta W^{\mathrm{sc}}_{\mathrm{reg}}[B,e]}{\delta B_\mu^a(x)}\Bigg|_{e=B}+
\frac{\delta W^{\mathrm{sc}}_{\mathrm{reg}}[B,e]}{\delta e_\mu^a(x)}\Bigg|_{e=B}\\\nonumber
&=Q_\mu^a|_{\mathrm{reg.}}[B,B](x)+
\frac{\delta W^{\mathrm{sc}}_{\mathrm{reg}}[B,e]}{\delta e_\mu^a(x)}\Bigg|_{e=B}\neq Q_\mu^a|_{\mathrm{reg.}}[B,B](x).
\end{align}
Moreover, it is precisely inequality that takes place. Direct calculations in the first two loops show that the second term has its own divergences, moreover, they do not depend on the choice of regularization. This fact symbolizes that the renormalization of the quantum action on the diagonal, that is, when $e=B$, no longer guarantees the renormalization of the quantum equation of motion. Moreover, from the renormalized $W^{\mathrm{sc}}_{\mathrm{reg}}[B,B]$ it is impossible to find the renormalized functional $Q_\mu^a|_{\mathrm{reg.}}[B,B](x)$, and this means that it will not be possible to calculate a renormalized analog for solution $B_f$.

The problem that has arisen is due to the fact that the information when working on the diagonal is much more scarce. Indeed, drawing a direct analogy with classical analysis, we note that it is impossible to reconstruct the derivatives of a function, knowing only its value at a selected point. Thus, to decipher the quantum equation of motion, it is necessary not only to renormalize $W^{\mathrm{sc}}_{\mathrm{reg}}[B,B]$, but also to additionally store some information about the variation with respect to the second argument. At the same time, working with strongly connected diagrams with a higher number of external lines leads to the need to work with the next variations with respect to the second argument.

This paper provides an analysis of the first two quantum corrections for the quantum action and the first correction for its first functional derivative. In connection with this formulation of the problem, it is proposed to work with an object of the form
\begin{equation}\label{ya-a-164}
W^{\mathrm{sc}}_{\mathrm{reg}}[B,B+\varepsilon],
\end{equation}
where the second argument $e_\mu^a$ is a small perturbation near the value of $B_\mu^a$. Thus, we discuss not only the renormalization in two loops for the action $W^{\mathrm{sc}}_{\mathrm{reg}}[B,B]$ on the diagonal, but also its relation to the first variation.
\begin{equation}\label{ya-a-165}
\frac{\delta W^{\mathrm{sc}}_{\mathrm{reg}}[B,B+\varepsilon]}{\delta \varepsilon_{\mu}^a(x)}
\Bigg|_{\varepsilon=0}.
\end{equation}
Note that the choice of $e_\mu^a=B_\mu^a+\varepsilon_\mu^a$ is also determined by the desire to work with Green's functions constructed using covariant operators. Indeed, the operators \eqref{ya-a-10} appear in the main order of $\varepsilon_\mu^a$, rather than \eqref{ya-a-32} and \eqref{ya-a-33}, which greatly simplifies calculations.\\

\noindent\textbf{Remark.} Next, for convenience, we use the notations
\begin{equation*}
W^{\mathrm{sc}}_{\mathrm{reg}}[B]\equiv W^{\mathrm{sc}}_{\mathrm{reg}}[B,B]\,\,\,\mbox{and}\,\,\, 
Q_\mu^a|_{\mathrm{reg.}}[B](x)\equiv Q_\mu^a|_{\mathrm{reg.}}[B,B](x).
\end{equation*}

\subsection{Renormalization of action}
\label{ym:sec:pr:gen3-4}
According to the general theory, renormalization of the strongly connected quantum regularized action from \eqref{ya-a-17} consists in redefining (shifting) the coupling constant $g\to g_\Lambda$, scaling the fields, and possibly multiplying individual parts of the classical action by renormalization constants. Thus, the appearance of additional terms of a new type in the classical action \eqref{ya-a-5} from the point of view of the generally accepted approach, which is inextricably linked to the use of dimensional regularization, indicates that the theory is not renormalizable. An example of such a theory is a special case of a two-dimensional non-linear sigma model (principal chiral field model) with a cutoff regularization, see \cite{AIK-25}. In this paper\footnote{Section \ref{ym:sec:mas} discusses the possibility of working with the classical renormalization process by introducing a mass term.}, renormalizability will have the classical formulation, supplemented by the ability to introduce terms with power singularities, which are not present in dimensional regularization by construction. Finally, extending the classical formulation, we assume that during the renormalization procedure, the classical action can be changed\footnote{In the case of dimensional regularization, not only the coupling constant $g$ is transformed, but also the parameter $\xi$ from the term \eqref{ya-a-11} responsible for gauge-fixing term. In our formulation, $\xi=1$. Nevertheless, the corresponding contributions with logarithmic singularities arise and are included in the second term of formula \eqref{ya-a-173}. It is shown below that such contributions are consistent with the available known results.} as follows
\begin{equation}\label{ya-a-173}
S_{\mathrm{tot}}[B,e]\longrightarrow S_{\mathrm{tot}}[B,e]\big|_{g\to g_\Lambda}+\hat{W}[B,e]
.
\end{equation}
In this case, the additive can be presented in the form of three parts
\begin{equation}\label{ya-a-174}
\hat{W}[B,e;a,c,\bar{c},g_\Lambda]=\bigg(\sum_{k=1}^{+\infty}\frac{g_\Lambda^{2k-2}}{4}\hat{W}_k[B,e]\bigg)
+g_\Lambda\hat{\Gamma}_1[B,e,a,g_\Lambda]
+\bigg(\sum_{k=2}^{+\infty}g^{k}_\Lambda\hat{\Gamma}_k[B,e;a,c,\bar{c}\,]\bigg)
,
\end{equation}
where the right hand side satisfies a set of the following assumptions.
\begin{enumerate}
	\item The first term is independent of the integration variables $a_\mu^a$, $c^a$, and $\bar{c}^a$. At the same time, important\footnote{This condition can be waived. Then the renormalization of the constant $g\to g_\Lambda\to g_{\mathrm{ren}}$ can be rewritten as the sum of auxiliary vertices.} is that $\hat{W}_k[B,B]$ is not proportional to $W_{-1}[B]$ for all index values.
	\item The second part is proportional to the first degree of the fluctuation. In this case, $\hat{\Gamma}_1$ can be decomposed in a series with respect to $g_\Lambda$, starting from the zero degree.
	\item In the third part, each coefficient is a finite set of vertices with two, three, or four external lines. Moreover, excluding the term with power singularity $\Lambda^2$ and proportional to $a_\mu^aa_\mu^a$, all vertices are parts of the classical action $S_{\mathrm{tot}}$. 
\end{enumerate}
To summarize, it can be argued that the renormalization process consists in the transformation of the form 
\begin{equation}\label{ya-a-170}
W_{\mathrm{reg}}^{\mathrm{sc}}[B,e]\longrightarrow W_{\mathrm{ren}}^{\mathrm{sc}}[B,e].
\end{equation}
In turn, in general, the ansatz looks like this
\begin{align}\label{ya-a-169}
W_{\mathrm{ren}}^{\mathrm{sc}}[B,e]&=\frac{1}{4g^2_\Lambda}
\bigg(W_{-1}[B]+\sum_{k=1}^{+\infty}g^{2k}_\Lambda\hat{W}_k[B,e]\bigg)\\\nonumber
&\phantom{=}+\bigg(\ln\det(\mathfrak{G}_0|_{\mathrm{reg.}})-\frac{1}{2}\ln\det(\mathfrak{G}_1|_{\mathrm{reg.}})+\kappa_0^\prime\bigg)\\\nonumber
&\phantom{=}-\mathbb{H}_0^{\mathrm{sc}}\bigg[\exp\bigg(-g_\Lambda\Gamma_3-\frac{g^2_\Lambda}{4}\Gamma_4-g_\Lambda\Omega_3-\sum_{k=2}^{+\infty}g^{k}_\Lambda\hat{\Gamma}_k\bigg)\bigg]
+\sum_{k=1}^{+\infty}g^{2k}\kappa_k^\prime.
\end{align}
Let us explain the notation. The strokes for the constants $\kappa_i^\prime$ indicate their redefinition, since densities independent of the background field can change during the renormalization. The new coupling constant $g_\Lambda$ is a function of the regularization parameter $\Lambda$ and can be decomposed through the finite and independent of $\Lambda$ renormalized constant $g_{\mathrm{ren}}$ as follows
\begin{equation}\label{ya-a-171}
\frac{1}{g^2_\Lambda}=\frac{1}{g_{\mathrm{ren}}^2}+\sum_{k=0}^{+\infty}a_k^{\phantom{1}}(\Lambda)g_{\mathrm{ren}}^{2k},
\end{equation}
where the coefficients $a_k^{\phantom{1}}(\Lambda)$ are polynomials of $\ln(\Lambda/\sigma)$ with degrees that are not higher than $k$ for $k>0$ and not higher than $1$ for $k=0$. The important condition is that the additives \eqref{ya-a-174} must be introduced in a coordinated manner, that is, without disrupting the connection between the quantum action and the quantum equation of motion. The fact that such a condition is met must be checked at each step.

After performing the renormalization process \eqref{ya-a-170}, taking into account the fact that the relationship between action and equation has been preserved, it is necessary to write out the renormalized quantum equation of motion and then find the renormalized quantum background field $B_{q,\mathrm{ren}}$. In the end, the usual action can be written out as
\begin{equation}\label{ya-a-172}
W_{\mathrm{ren}}^{\mathrm{sc}}[B,e]\longrightarrow
W_{\mathrm{ren}}^{\mathrm{sc}}[B_{q,\mathrm{ren}},e]=W_{\mathrm{ren}}^{\mathrm{c}}[B_{q,\mathrm{ren}},e].
\end{equation}

\subsection{Renormalization near diagonal}
\label{ym:sec:pr:gen3-5}
In Section \ref{ym:sec:pr:gen3-3}, it was noted that in Faddeev's approach, when studying the action of \eqref{ya-a-17}, some information is lost, since it is impossible to unambiguously restore the global properties of a function by its value at a point. Therefore, in order to preserve the connection between the action and the equation, it is necessary to study in parallel the first variation of \eqref{ya-a-165} over the field responsible for the gauge condition. 

Consider the renormalization problem near the diagonal and choose $e_\mu^a=B_\mu^a+\varepsilon_\mu^a$ as a small change in the background field of $B_\mu^a$ with fixed specified boundary conditions. Then the functionals depending on $e_\mu^a$ can be represented as finite series with respect to a small fluctuation of $\varepsilon_\mu^a$ as follows:
\begin{fleqn}
\begin{equation*}
\mbox{\textbf{Vertex:}}\,\,\,\,\,\,\Omega_3=\Omega_3\big|_{e=B}+\int_{\mathbb{R}^4}\mathrm{d}^4x\,
\varepsilon_\mu^{ab}\bar{c}^bf^{aed}a_\mu^ec^d;
\end{equation*}
\end{fleqn}
\begin{fleqn}
\begin{align*}
\mbox{\textbf{Operators:}}\,\,\,\,\,\,\,\,\mathfrak{M}_0^{ab}(x)&=M_0^{ab}(x)-\varepsilon_\mu^{ae}(x)D_\mu^{eb}(x),\\
\mathfrak{M}_{1\mu\nu}^{\,\,ab}(x)&=M_{1\mu\nu}^{\,\,ab}(x)-
\varepsilon_\mu^{ae}(x)D_\nu^{eb}(x)-D_\mu^{ae}(x)\varepsilon_\nu^{eb}(x)-
\varepsilon_\mu^{ae}(x)\varepsilon_\nu^{eb}(x);
\end{align*}
\end{fleqn}
\begin{fleqn}
	\begin{align*}
		\mbox{\textbf{Renormalization functionals:}}\,\,\,\,\,\,\hat{W}_k^{\phantom{1}}[B,e]&=\hat{W}_k^0[B]+\hat{W}_k^1[B,\varepsilon]+\ldots,
		\\
		\hat{\Gamma}_k^{\phantom{1}}[B,e]&=\hat{\Gamma}_k^0[B]+\hat{\Gamma}_k^1[B,\varepsilon]+\ldots,
	\end{align*}
\end{fleqn}
where $\hat{\Gamma}_k^i$ and $\hat{W}_k^i$ are proportional to the $i$-th power of $\varepsilon$. Also we introduce three additional auxiliary vertices
\begin{align}\label{ya-z-22}
V_1[B,\varepsilon]&=\int_{\mathbb{R}^4}\mathrm{d}^4x\,\bar{c}^a\varepsilon_\mu^{ae}D_\mu^{eb}c^b,
\\\label{ya-z-23}
V_2[B,\varepsilon]&=\int_{\mathbb{R}^4}\mathrm{d}^4x\,
\big(a_\mu^e\varepsilon_\mu^{ea}\big)
\big(D_\nu^{ab}a_\nu^b\big),
\\\label{ya-z-24}
V_3[\varepsilon]&=\int_{\mathbb{R}^4}\mathrm{d}^4x\,\varepsilon_\mu^{ab}\bar{c}^bf^{aed}a_\mu^ec^d.
\end{align}
\textbf{Remark.} Next, the symbol $\Omega_3$ will denote the vertex at $e_\mu^a=B_\mu^a$, and the connecting lines \eqref{ya-a-15} will be matched with the regularized Green's functions $G_0$ and $G_1$ taking into account Section \ref{ym:sec:pr:gen3-2}.\\

\noindent Note that the strongly connected renormalized action \eqref{ya-a-169} can be represented as a series by powers of the small field $\varepsilon_\mu^a$ as follows
\begin{equation*}
W_{\mathrm{ren}}^{\mathrm{sc}}[B,B+\varepsilon]=
W_{\mathrm{ren}}^{\mathrm{sc}}[B,B]+
\dot{W}_{\mathrm{ren}}^{\mathrm{sc}}[B,\varepsilon]+\mathcal{O}(\varepsilon^2),
\end{equation*}
where the second coefficient with the dot denotes the derivative of the action near the diagonal with respect to the second argument and, thus, is linear\footnote{Formula \eqref{ya-a-175} actually defines the auxiliary field $J_\mu^a$.} perturbation
\begin{equation}\label{ya-a-175}
\dot{W}_{\mathrm{ren}}^{\mathrm{sc}}[B,\varepsilon]=\Big(\partial_s^{\phantom{1}}W_{\mathrm{ren}}^{\mathrm{sc}}[B,B+s\varepsilon]\Big)
\Big|_{s=0}=
\int_{\mathbb{R}^4}\mathrm{d}^4x\,\varepsilon_\mu^a(x)J_\mu^a(x).
\end{equation}
At the same time, the functionals themselves, taking into account ansatz \eqref{ya-a-169}, allow explicit representations
\begin{fleqn}
\begin{align}\label{ya-a-176}
	W_{\mathrm{ren}}^{\mathrm{sc}}[B,B]&=\frac{1}{4g^2_\Lambda}
	\bigg(W_{-1}^{\phantom{1}}+\sum_{k=1}^{+\infty}g^{2k}_\Lambda\hat{W}_k^0\bigg)
+\bigg(\ln\det(G_0|_{\mathrm{reg.}})-\frac{1}{2}\ln\det(G_1|_{\mathrm{reg.}})+\kappa_{0}^\prime\bigg)\\\nonumber
	&\,\,\,\,\,\,\,~~~~~~~~~~~~~~~~~~~~~~~~~~~~~
	-\mathbb{H}_0^{\mathrm{sc}}\bigg[\exp\bigg(-g_\Lambda\Gamma_3-\frac{g^2_\Lambda}{4}\Gamma_4-g_\Lambda\Omega_3-\sum_{k=2}^{+\infty}g^{k}_\Lambda\hat{\Gamma}_k^0\bigg)\bigg]
	+\sum_{k=1}^{+\infty}g^{2k}\kappa_{k}^\prime,
\\\label{ya-a-177}
\dot{W}_{\mathrm{ren}}^{\mathrm{sc}}[B,\varepsilon\,]&=\frac{1}{4}
	\sum_{k=1}^{+\infty}g^{2k-2}_\Lambda\hat{W}_k^1
-\mathbb{H}_0^{\mathrm{sc}}\bigg[
	\bigg(V_1+V_2-V_3-\sum_{k=2}^{+\infty}g^{k}_\Lambda\hat{\Gamma}_k^1\bigg)\\\nonumber&
	\,\,\,\,\,\,\,~~~~~~~~~~~~~~~~~~~~~~~~~~~~~
	\times
	\exp\bigg(-g_\Lambda\Gamma_3-\frac{g^2_\Lambda}{4}\Gamma_4-g_\Lambda\Omega_3-\sum_{k=2}^{+\infty}g^{k}_\Lambda\hat{\Gamma}_k^0\bigg)\bigg].
\end{align}
\end{fleqn}
Thus, it can be seen from the construction that the two actions are not equivalent to each other, therefore it is necessary to renormalize for both. It should be noted that the renormalization of the first one \eqref{ya-a-176} can be performed independently. In this case, the study of the "$k$-th loop", subject to the availability of results for previous corrections, leads to answers for the following renormalization values:
\begin{align}\nonumber
	\mbox{$1$-st correction for}\,\,W_{\mathrm{ren}}^{\mathrm{sc}}\,\,&\longrightarrow\,\,
	a_{0},\,\,\hat{W}_1^0;
	\\\nonumber
\mbox{$2$-nd correction for}\,\,W_{\mathrm{ren}}^{\mathrm{sc}}\,\,&\longrightarrow\,\,
a_{1},\,\,\hat{W}_2^0,\,\,\hat{\Gamma}_{2}^0;
\\\label{ya-a-182}
\mbox{$k$-th correction for}\,\,W_{\mathrm{ren}}^{\mathrm{sc}}\,\,&\longrightarrow\,\,
a_{k-1},\,\,\hat{W}_k^0,\,\,\hat{\Gamma}_{2k-2}^0,\,\,\hat{\Gamma}_{2k-3}^0,\,\,
\mbox{if}\,\,k>2.
\end{align}
In turn, the renormalization of the second action depends on the auxiliary vertices and coefficients found when working with the first one. The procedure turns out to be recurrent again:
\begin{align}\nonumber
	\mbox{$1$-st corrections for}\,\,W_{\mathrm{ren}}^{\mathrm{sc}}\,\,\mbox{and}\,\,\dot{W}_{\mathrm{ren}}^{\mathrm{sc}}\,\,&\longrightarrow\,\,
\hat{W}_1^1;
\\\nonumber
	\mbox{$2$-nd corrections for}\,\,W_{\mathrm{ren}}^{\mathrm{sc}}\,\,\mbox{and}\,\,\dot{W}_{\mathrm{ren}}^{\mathrm{sc}}\,\,&\longrightarrow\,\,
\hat{W}_2^1,\,\,\hat{\Gamma}_{2}^1;
\\\label{ya-a-183}
	\mbox{$k$-th corrections for}\,\,W_{\mathrm{ren}}^{\mathrm{sc}}\,\,\mbox{and}\,\,\dot{W}_{\mathrm{ren}}^{\mathrm{sc}}\,\,&\longrightarrow\,\,
	\hat{W}_k^1,\,\,\hat{\Gamma}_{2k-2}^1,\,\,\hat{\Gamma}_{2k-3}^1,\,\,
	\mbox{if}\,\,k>2.
\end{align}
Recall that the second object from \eqref{ya-a-177} is auxiliary and is used exclusively to reconstruct the quantum equation of motion. Indeed, taking into account \eqref{ya-a-175}, the following formula is valid
\begin{align}\label{ya-a-180}
Q_\mu^a\Big|_{\mathrm{ren.}}[B](x)&=
\frac{\delta W^{\mathrm{sc}}_{\mathrm{ren}}[B,B]}{\delta B_{\mu}^a(x)}-
\frac{\delta W^{\mathrm{sc}}_{\mathrm{ren}}[B,B+\varepsilon]}{\delta \varepsilon_\mu^a(x)}\Bigg|_{\varepsilon=0}\\\nonumber
&=\frac{\delta W^{\mathrm{sc}}_{\mathrm{ren}}[B,B]}{\delta B_{\mu}^a(x)}-J_\mu^a[B](x).
\end{align}
However, the resulting auxiliary field $J_\mu^a$ is cumbersome. Let us highlight the elements that can be used to restore it, if necessary. It is clear that, unlike \eqref{ya-a-176}, formula \eqref{ya-a-177} additionally contains new parts of the classical action $\hat{W}_k^1$ and vertices $\hat{\Gamma}_k^1$. Therefore, knowing them, we can restore the second action entirely. Let us define one more auxiliary field
\begin{equation}\label{ya-a-181}
\int_{\mathbb{R}^4}\mathrm{d}^4x\,\varepsilon_\mu^a(x)j_\mu^a[B](x)=
\frac{1}{4}\sum_{k=1}^{+\infty}g^{2k-2}\hat{W}_k^1[B,\varepsilon]
+\sum_{k=2}^{+\infty}g^{k}\hat{\Gamma}_k^1[B,\varepsilon].
\end{equation}
However, the part with $\hat{\Gamma}_1^1$ is missing from the definition, since it is actually missing\footnote{Because the incoming diagrams must be strongly connected.} in the quantum equation of motion on the diagonal and in the auxiliary field $J_\mu^a$. Thus, if we define the field $j_\mu^a$, then using \eqref{ya-a-177}, we get $J_\mu^a$.

\subsection{What do we calculate?}
\label{ym:sec:pr:gen3-6}
So, from a computational point of view, the main purpose of Section \ref{ym:sec:pr:gen3} devoted to strong deformation is to demonstrate the process of determining the coefficients and vertices of renormalization in the first two "loops" for the action on the diagonal, as well as for the first correction in the case of the first functional derivative. Using the schemes from \eqref{ya-a-182} and \eqref{ya-a-183}, we note that it is necessary to perform the following tasks:
\begin{align*}
\mbox{1-st correction for}\,\,W_{\mathrm{ren}}^{\mathrm{sc}}&\,\,\longrightarrow\,\,\mbox{answers for}\,\,
a_0,\,\,\hat{W}_1^0;
\\
\mbox{1-st correction for}\,\,\dot{W}_{\mathrm{ren}}^{\mathrm{sc}}&\,\,\longrightarrow\,\,\mbox{answer for}\,\,\hat{W}_1^1;
\\
\mbox{2-nd correction for}\,\,W_{\mathrm{ren}}^{\mathrm{sc}}&\,\,\longrightarrow\,\,\mbox{answer for}\,\,
\hat{\Gamma}_2^0\,\,+\,\,\mbox{structure for}\,\,a_1,\,\,\hat{W}_2^0.
\end{align*}
Using expansions \eqref{ya-a-176} and \eqref{ya-a-177} obtained in the previous section, we write out the first orders with respect to the renormalized coupling constant for the quantum action on the diagonal and for its first derivative with respect to the second argument
\begin{fleqn}
\begin{align}\label{ya-z-17}
	W_{\mathrm{ren}}^{\mathrm{sc}}[B,B]&=\frac{1}{4g^2_{\mathrm{ren}}}W_{-1}[B]+W_0^{\Lambda}[B]+
	g^2_{\mathrm{ren}}W_1^{\Lambda}[B]+\mathcal{O}\big(g^4_{\mathrm{ren}}\big),\\\label{ya-z-18}
~~~~~~~~~~~~~~	\dot{W}_{\mathrm{ren}}^{\mathrm{sc}}[B,\varepsilon\,]&=\dot{W}_0^{\Lambda}[B,\varepsilon]+\mathcal{O}\big(g^2_{\mathrm{ren}}\big),
\end{align}
\end{fleqn}
where the coefficients for arbitrary $B_\mu^a$ and $\varepsilon_\mu^a$ are determined by the following equalities
\begin{fleqn}
\begin{align}\label{ya-a-49}
	W_0^{\Lambda}&=\ln\det(G_0|_{\mathrm{reg.}})-\frac{1}{2}\ln\det(G_1|_{\mathrm{reg.}})
	+\frac{1}{4}\Big(a_0W_{-1}^{\phantom{1}}+\hat{W}_1^0\Big)+\kappa_0^\prime,
	\\\label{ya-a-163}
~~~~~~~~~~~~~~	W_1^{\Lambda}&=-\frac{1}{2}\mathbb{H}_0^{\mathrm{sc}}\big(\Gamma_3^2\big)
	+\frac{1}{4}\mathbb{H}_0^{\mathrm{sc}}\big(\Gamma_4^{\phantom{1}}\big)-\frac{1}{2}\mathbb{H}_0^{\mathrm{sc}}\big(\Omega_3^2\big)+\mathbb{H}_0^{\mathrm{sc}}\big(\hat{\Gamma}_2^0\big)+
	\frac{1}{4}\Big(a_1W_{-1}^{\phantom{1}}+\hat{W}_2^0\Big)+\kappa_1^\prime,
\\\label{ya-a-178}
	\dot{W}_0^{\Lambda}&=\frac{1}{4}\hat{W}_1^1-\mathbb{H}_0^{\mathrm{sc}}
	\big(V_1+V_2\big).
\end{align}
\end{fleqn}
In this case, the diagrammatic representation of the main coefficients can be written in the following standard way
\begin{align}\label{ya-a-19}
	\mathbb{H}_0^{\mathrm{sc}}\big(\Gamma_3^2\big)&=
	{\centering\adjincludegraphics[width = 1.2 cm, valign=c]{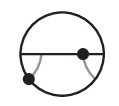}}-
	{\centering\adjincludegraphics[width = 1.2 cm, valign=c]{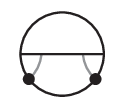}}+
	{\centering\adjincludegraphics[width = 1.2 cm, valign=c]{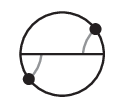}}-2
	{\centering\adjincludegraphics[width = 1.2 cm, valign=c]{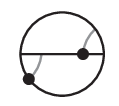}}+
	{\centering\adjincludegraphics[width = 1.2 cm, valign=c]{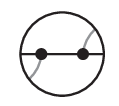}},
	\\\label{ya-a-20}
	\mathbb{H}_0^{\mathrm{sc}}\big(\Omega_3^2\big)&=
	-{\centering\adjincludegraphics[width = 1.2 cm, valign=c]{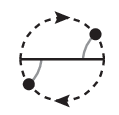}},
	\\\label{ya-a-21}
	\mathbb{H}_0^{\mathrm{sc}}\big(\Gamma_4\big)&=
	{\centering\adjincludegraphics[width = 1.2 cm, valign=c]{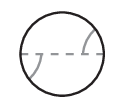}}-
	{\centering\adjincludegraphics[width = 1.2 cm, valign=c]{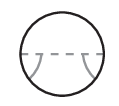}}+
	{\centering\adjincludegraphics[width = 2 cm, valign=c]{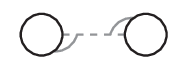}}.
\end{align}
Next, we use the assumption that the ansatz \eqref{ya-a-176} and \eqref{ya-a-177} do not contain ultraviolet singularities after a suitable choice of renormalization coefficients. The feasibility of such an assumption is quite expected, since the Yang--Mills theory is renormalizable in the case of dimensional regularization, as evidenced by the value of the divergence index. It is assumed that in the case of the cutoff, it is renormalizable in the generalized\footnote{See discussions in Section \ref{ym:sec:pr:gen3-4}.} formulation involving power-law divergences. Then it can be argued that there are no ultraviolet singularities in each individual order with respect to the coupling constant $g_{\mathrm{ren}}$. Therefore, each order leads to recurrence equations, which, after introducing a special sign for the equality of the singular parts, can be written as follows
\begin{align}\label{ya-a-184}
W_0^{\Lambda}&\stackrel{\mathrm{s.p.}}{=}0,\,\,\,W_1^{\Lambda}\stackrel{\mathrm{s.p.}}{=}0
\,\,\,\mbox{etc.},
\\\label{ya-a-184-1}
\dot{W}_0^{\Lambda}&\stackrel{\mathrm{s.p.}}{=}0\,\,\,\mbox{etc.}
\end{align}
The study of the explicit form of renormalization constants and vertices included in the selected relations is the main task associated with the strong deformation.

\subsection{Results}
\label{ym:sec:pr:gen3-7}
\begin{theorem}\label{ya-t1}
Let us assume that all the assumptions formulated above in Section \ref{ym:sec:pr:gen3} are fulfilled. Let also $L=\ln(\Lambda/\sigma)$, where $\Lambda$ is a regularization parameter and $\sigma>0$ is an auxiliary fixed parameter to make combinations dimensionless. We additionally define the functionals
\begin{equation*}
S_2[B]=\int_{\mathbb{R}^4}\mathrm{d}^4x\,B_\mu^a(x)B_\mu^a(x),\,\,\,
S_{\mathrm{f}}[a,e]=\int_{\mathbb{R}^4}\mathrm{d}^4x\,\Big(D_\mu^{ab}(x)a_\mu^b(x)\Big)\Big(\mathfrak{D}_\nu^{ac}(x)a_\nu^c(x)\Big),
\end{equation*}
as well as six functionals $W_{-1}^i[B]$, where $i\in\{1,\ldots,6\}$, which are parts of the classical action and are determined by integrals over $\mathbb{R}^4$ from densities of the following form
\begin{equation*}
\begin{tabular}{ccc}
$\big(\partial_{y^\mu}B_\mu^a(y)\big)\big(\partial_{y^\nu}B_\nu^a(y)\big)$,&
$\big(\partial_{y^\mu}B_\mu^a(y)\big)f^{abc}B_{\nu}^b(y)B_\nu^c(y)$,&
$f^{ade}B_{\mu}^d(y)B_\mu^e(y)f^{abc}B_{\nu}^b(y)B_\nu^c(y)$,\\
$\big(\partial_{y^\mu}B_\nu^a(y)\big)\big(\partial_{y^\mu}B_\nu^a(y)\big)$,&
$\big(\partial_{y^\mu}B_\nu^a(y)\big)f^{abc}B_{\mu}^b(y)B_\nu^c(y)$,&
$f^{ade}B_{\mu}^d(y)B_\nu^e(y)f^{abc}B_{\mu}^b(y)B_\nu^c(y)$.
\end{tabular}
\end{equation*}
Next, we define several\footnote{Here $\mathrm{B}_{1}$ is a closed ball of unit radius centered at the origin.} auxiliary numbers\footnote{They are functionals that depend on the regularizing function $\mathbf{f}(\cdot)$ and were introduced when calculating asymptotic expansions with respect to the parameter $\Lambda$ in Section \ref{ym:sec:two-vs}.}, depending on the regularized free fundamental solution $R_0^1(x)$, see \eqref{ya-a-37},
\begin{equation*}
	\frac{\rho_0}{4\pi^2}=R_0^1(0),\,\,\,
	\frac{\rho_3}{4\pi^2}=\int_{\mathrm{B}_{1}}\mathrm{d}^4x\,R_0^1(x)A_0^{\phantom{1}}(x)R_0^1(x),\,\,\,
	\frac{2\rho_4}{\pi^2}=\int_{\mathrm{B}_{1}}\mathrm{d}^4x\,R_0^1(x)|x|^2A_0^{\phantom{1}}(x)R_0^1(x),
\end{equation*}
\begin{equation*}
	\frac{\rho_2}{8\pi^2}=
	\int_{\mathrm{B}_{1}}\mathrm{d}^4x\,A_0^{\phantom{1}}(x)R_0^1(x)
	\int_{\mathbb{R}^4}\mathrm{d}^4y\,\Big(R_0^1(x-y)R_0^1(y)-R_0^1(y)R_0^1(y)\Big),
\end{equation*}
\begin{equation*}
	\frac{\rho_6}{4\pi^2}=\int_{\mathrm{B}_{1}}\mathrm{d}^4x\,A_0^{\phantom{1}}(x)R_0^1(x)
	\int_{\mathbb{R}^4}\mathrm{d}^4y\,R_0^1(x-y)A_0^{\phantom{1}}(y)R_0^1(y),
\end{equation*}
as well as
\begin{align*}
	\frac{\rho_7}{16\pi^4}=\int_{\mathbb{R}^4}\mathrm{d}^4x\,\bigg(
	\Big(&-2\partial_{x_\mu}\tilde{\theta}_\mu(x)-\tilde{\tau}(x)-R_0^1(x)\Big)\Big(R_0^1(x)\Big)^2\\&
	+2R_0^1(x)\Big(A_0^{\phantom{1}}(x)R_0^1(x)\Big)
	\Big(\theta(x)-A_0^{\phantom{1}}(x)\tau(x)\Big)\Big|_{\Lambda=1,\,\sigma\to+0}\\&+
	R_0^1(x)\tilde{\theta}_\mu(x)\Big(\partial_{x_\mu}R_0^1(x)-A_0^{\phantom{1}}(x)\tilde{\theta}_\mu(x)\Big)
	-\frac{1}{2}\tilde{\theta}_\mu(x)\tilde{\theta}_\mu(x)A_0^{\phantom{1}}(x)R_0^1(x)
	\bigg)
	,
\end{align*}
where
\begin{align*}
\tilde{\theta}_\mu(x)&=\int_{\mathbb{R}^4}
\mathrm{d}^4z\,R_0^1(x-z)\partial_{z^\mu}R_0^1(z),\\
\tilde{\tau}(x)&=
\int_{\mathbb{R}^4}\mathrm{d}^4z\int_{\mathbb{R}^4}\mathrm{d}^4y\,R_0^{1}(x-z)A_0^{\phantom{1}}(z)R_0^{1}(z-y)A_0^{\phantom{1}}(y)R_0^{1}(y).
\end{align*}
Also, let the numbers $\tilde{a}$, $\{\tilde{a}_i,\tilde{b}_i\}_{i=0,1}$, and $\{\tilde{d}_i\}_{i=1}^6$ denote arbitrary finite elements from $\mathbb{R}$. Then the presented equalities \eqref{ya-a-184} and \eqref{ya-a-184-1} lead to the following coefficients and auxiliary vertices

\begin{fleqn}
\begin{equation*}
1)\,\,W_0^{\Lambda}\stackrel{\mathrm{s.p.}}{=}0\,\,\Longrightarrow\,\,
a_0(\Lambda)=\frac{11c_2L}{24\pi^2}+\tilde{a}_0,\,\,\,\hat{W}_1^0[B]=\frac{c_2\Lambda^2(\rho_3-2\rho_0)}{2\pi^2}S_2[B].
\end{equation*}
\end{fleqn}
\begin{fleqn}
	\begin{equation*}
		2)\,\,\dot{W}_0^{\Lambda}\stackrel{\mathrm{s.p.}}{=}0\,\,\Longrightarrow\,\,
		\hat{W}_1^1[B,\varepsilon]=\Gamma_1[\varepsilon]\bigg(\frac{c_2L}{2\pi^2}+\tilde{b}_0\bigg).
	\end{equation*}
\end{fleqn}
\begin{fleqn}
\begin{align*}
3)\,\,W_1^{\Lambda}\stackrel{\mathrm{s.p.}}{=}0\,\,\Longrightarrow\,\,
\hat{\Gamma}_2^0[B]&=\frac{\Lambda^2c_2(2\rho_3-3\rho_0)}{8\pi^2}
S_2[a]-\bigg(\frac{5Lc_2}{48\pi^2}+\tilde{b}_1\bigg)S_{\mathrm{f}}[a,B].\\
a_1(\Lambda)&=\frac{c_2^2L}{(4\pi)^4}\Big(26+2^417\rho_2/3-2^7\rho_4+\tilde{a}\Big)+\tilde{a}_1,
\\
\hat{W}_2^0[B]&=\frac{c_2^2\Lambda^2}{2\pi^4}\bigg(\frac{5L}{24}\big(\rho_0+\rho_6-2\rho_3\big)-\rho_2\rho_3-\rho_7\bigg)S_2[B]
\\
&-\frac{\tilde{a}c_2^2L}{(4\pi)^4}W_{-1}[B]+\Big(\mbox{linear combination of}\,\,\big(L+\tilde{d}_i\big)\times W_{-1}^i[B]\Big).
\end{align*}
\end{fleqn}
\end{theorem}
\vspace{3mm}
\noindent\textbf{Comment 1.} The first coefficient $a_0$ has the standard value and is consistent with the results when using other\footnote{The logarithmic singularity in the first correction does not depend on the choice of regularization and the scheme of subtraction, therefore, in all cases it has the same form.} regularizations, see for example \cite{16,17}. The second coefficient $a_1$ has a different value from the one obtained earlier in dimensional regularization, see \cite{104,104-11} or in the framework of the background field method \cite{11,12}, and in the case of implicit regularization, see \cite{chi-2}. It depends on the regularizing function $\mathbf{f}(\cdot)$, see \eqref{ya-a-37}, and, moreover, can be shifted by selecting the free parameters $\tilde{a}$ and $\tilde{a}_1$. In this case, the parameter $\tilde{a}$ is a consequence of the ambiguity of fixing the auxiliary functional of the renormalization $\hat{W}_2^0$, which appeared due to the lack of invariance with respect to the gauge transformations of the background field and, as a result, the splitting of the classical action. The strict additional conditions for fixing the parameter are not clear at the moment and may be chosen out of convenience. For example, if we define
\begin{equation*}
\tilde{a}=-59/9-2^55\rho_2/3,
\end{equation*}
then the coefficient $a_1$ will coincide with the case of covariant regularization from Theorem \ref{ya-t2}, and the functional $\hat{W}_2^0$ in this case will symbolize a value of deviation caused by an additional deformation of the covariant case.

\vspace{3mm}
\noindent\textbf{Comment 2.} The functional $\hat{W}_1^1[B,\varepsilon]$ from the second point has a standard value and coincides with the results for other regularizations. It is easy to verify this by calculating the variations, taking into account \eqref{ya-a-180}, for the renormalization functionals
\begin{equation*}
\frac{1}{4}\bigg(\frac{\delta}{\delta B_\mu^a(x)}-\frac{\delta}{\delta \varepsilon_\mu^a(x)}\bigg)
\Big(a_0W_{-1}^{\phantom{1}}[B]+\hat{W}_1^1[B,\varepsilon]\Big)\bigg|_{\varepsilon=0}=D_\nu^{ab}(x)F_{\nu\mu}^b(x)
\bigg(-\frac{c_2L}{3\pi^2}-\tilde{a}_0\bigg).
\end{equation*}
Indeed, the result exactly compensates for the singular part of the sum of $\mathbb{H}_1^{\mathrm{c}}(\Omega_3)$ and $\mathbb{H}_1^{\mathrm{c}}(\Gamma_3)$, the value of which was previously verified in the formula (80) of work \cite{Ivanov-Kharuk-20222} during the study of the quantum equation of motion. 

\vspace{3mm}
\noindent\textbf{Comment 3.} The counter-vertex $S_{\mathrm{f}}[a,B]$ with a logarithmic coefficient $L$ has a standard form and is completely consistent with the known results, see formula (3.42) in \cite{12}. An interesting fact is that it is introduced out of the need to reduce "nonlocal" singularities in the quantum action. For example, in \cite{ya-20} this vertex was determined differently, from the renormalization of the quantum equation of motion. Thus, at the level of the first two corrections, it can be seen that the quantum action on the diagonal is closed from the point of view of determining the renormalization constants. 

\vspace{3mm}
\noindent\textbf{Comment 4.} The counter-term $S_2[B]$ and the counter-term $S_2[a]$ are included together with power-law singularities $\Lambda^2$. Such terms are standard in the case of cutoff and symbolize the loss of invariance with respect to the gauge transformations of the background field. The coefficients depend on the regularizing function $\mathbf{f}(\cdot)$.

\vspace{3mm}
\noindent\textbf{Comment 5.} Note that in the theorem, all coefficients with logarithmic singularities are provided with additional free constants, which can be fixed based on the choice of the subtraction scheme. At the same time, terms with power singularities do not have such additions. This choice is not important and is used by the authors solely for the convenience of calculations in the next corrections. If necessary, we can add free constants.

\vspace{3mm}
\noindent\textbf{Proof.} Following the formulation, we divide the solution of the equations into three parts.\\
	
\noindent\underline{The first part.} The equality of the singular part to zero for the first quantum renormalized correction \eqref{ya-a-49} boils down to searching for singularities for "$\ln\det$" and further defining the coefficient $a_0$ and the auxiliary functional $\hat{W}_1^0$ so that the result is finite. In Section \ref{ym:sec:one}, the determinants are investigated using explicit perturbative formulas \eqref{ya-a-48} and \eqref{ya-a-185}. The results consist of two parts:
\begin{align*}
\ln\det(G_0|_{\mathrm{reg.}})\,\,\longrightarrow\,\,\mbox{see formula}\,\,\eqref{ya-a-64};
\\
\ln\det(G_1|_{\mathrm{reg.}})\,\,\longrightarrow\,\,\mbox{see formula}\,\,\eqref{ya-a-69}.
\end{align*}
Substituting the results into \eqref{ya-a-49} and equating the singular part to zero, we obtain an equation of the form
\begin{equation*}
-\frac{c_2\Lambda^2(\rho_3-2\rho_0)}{8\pi^2}S_2[B]-\frac{11c_2L}{96\pi^2}W_{-1}[B]
+\frac{1}{4}\Big(a_0W_{-1}^{\phantom{1}}[B]+\hat{W}_1^0[B]\Big)\stackrel{\mathrm{s.p.}}{=}0,
\end{equation*}
where $L=\ln(\Lambda/\sigma)$. By solving this linear equation, we obtain the answers indicated in the formulation.\\

\noindent\underline{The second part.} The equality of the singular part to zero for the first variation in the field responsible for the gauge condition is reduced to the variation of the determinants, which, in turn, are reduced to the calculation of one-loop contributions $\mathbb{H}_0^{\mathrm{sc}}
\big(V_1\big)$ and $\mathbb{H}_0^{\mathrm{sc}}
\big(V_2\big)$. They are studied in Section \ref{ym:sec:one}, and the results are presented in formulas:
\begin{align*}
-\mathbb{H}_0^{\mathrm{sc}}
\big(V_1\big)=\int_{\mathbb{R}^{4}}\mathrm{d}^4z\,e_\mu^g(z)\Bigg(\frac{\delta}{\delta e_\mu^g(z)}\ln\det(\mathfrak{G}_0|_{\mathrm{reg.}})\bigg|_{e=B}\Bigg)\,\,\longrightarrow\,\,\mbox{see formula}\,\,\eqref{ya-a-67};
	\\
2\mathbb{H}_0^{\mathrm{sc}}
\big(V_2\big)=\int_{\mathbb{R}^{4}}\mathrm{d}^4z\,e_\mu^g(z)\Bigg(\frac{\delta}{\delta e_\mu^g(z)}\ln\det(\mathfrak{G}_1|_{\mathrm{reg.}})\bigg|_{e=B}\Bigg)\,\,\longrightarrow\,\,\mbox{see formula}\,\,\eqref{ya-a-73}.
\end{align*}
Substituting the obtained relations into \eqref{ya-a-178} and equating the singular part to zero, we obtain an equation of the form
\begin{equation*}
\frac{1}{4}\hat{W}_1^1[B,\varepsilon]-\frac{c_2L}{8\pi^2}\Gamma_1[\varepsilon]
\stackrel{\mathrm{s.p.}}{=}0
\end{equation*}
from which follows the stated answer.\\

\noindent\underline{The third part.} Consider the second correction for the action on the diagonal. In this case, it is convenient to study the linear combination of standard diagrams separately:
\begin{equation*}
-\frac{1}{2}\mathbb{H}_0^{\mathrm{sc}}\big(\Gamma_3^2\big)
+\frac{1}{4}\mathbb{H}_0^{\mathrm{sc}}\big(\Gamma_4^{\phantom{1}}\big)-\frac{1}{2}\mathbb{H}_0^{\mathrm{sc}}\big(\Omega_3^2\big)
\,\,\longrightarrow\,\,\mbox{see formula}\,\,\eqref{ya-a-166}.
\end{equation*}
Additionally, we substitute the relations of \eqref{ya-a-158} and \eqref{ya-a-159} for nonlocal parts into the resulting ratio, using the definitions for the counter-vertices
\begin{equation}\label{ya-z-14}
\mathbb{H}_0^{\mathrm{sc}}\big(S_2[\,\cdot\,]\big)=\int_{\mathbb{R}^4}\mathrm{d}^4x\,G_{1\rho\rho}^{\,\,aa}(x,x),\,\,\,
\mathbb{H}_0^{\mathrm{sc}}\big(S_{\mathrm{f}}[\,\cdot\,,B]\big)=-\int_{\mathbb{R}^4}\mathrm{d}^4x\,\Big(D_\rho^{bc}(x)D_\sigma^{ca}(x)G_{1\sigma\rho}^{\,\,ab}(x,y)\Big)\Big|_{y=x}.
\end{equation}
Then, after substituting and cancelling some terms, we obtain the equation of the form
\begin{multline}\label{ya-a-186}
	-\frac{\Lambda^2c_2(2\rho_3-3\rho_0)}{8\pi^2}
	\mathbb{H}_0^{\mathrm{sc}}\big(S_2[\,\cdot\,]\big)+\frac{5Lc_2}{48\pi^2}
	\mathbb{H}_0^{\mathrm{sc}}\big(S_{\mathrm{f}}[\,\cdot\,,B]\big)+\mathbb{H}_0^{\mathrm{sc}}\big(\hat{\Gamma}_2^0\big)\\
	-\frac{c_2^2W_{-1}L}{32(4\pi^2)^2}\Big(13+68\rho_2/3-64\rho_4\Big)
	+\frac{c_2^2W_{-1}^{+}L}{32(4\pi^2)^2}\Big(2\rho_2-1/2\Big)\\
	+2\mathrm{J}_{\ominus}+\frac{c_2\rho_3\Lambda^2\mathrm{J}_{\odot}}{\pi^2}-\frac{5Lc_2}{48\pi^2}
	\mathrm{J}_{\otimes}+
	\frac{1}{4}\Big(a_1W_{-1}^{\phantom{1}}+\hat{W}_2^0\Big)+\kappa_1^\prime\stackrel{\mathrm{s.p.}}{=}0.
\end{multline}
If the singular part is equal to zero, the answer for the vertex $\hat{\Gamma}_2^0[B,a]$ immediately follows, so the first line can be excluded from the ratio. Further, the second line has a suitable form as the sum of the classical action and one of its special parts, see \eqref{ya-c-1}. Therefore, next we need to use calculations for the remaining diagrams:
\begin{align*}
	\mathrm{J}_{\otimes}\,\,&\longrightarrow\,\,\mbox{see formula}\,\,\eqref{ya-a-188};
	\\
	\mathrm{J}_{\odot}\,\,&\longrightarrow\,\,\mbox{see formula}\,\,\eqref{ya-a-190};
	\\
	\mathrm{J}_{\ominus}\,\,&\longrightarrow\,\,\mbox{see formula}\,\,\eqref{ya-a-191}.
\end{align*}
Summing up and substituting into \eqref{ya-a-186}, we get the stated answers for $a_1(\Lambda)$ and $\hat{W}_2^0[B]$. The theorem has been proved.

\section{Weak deformation}
\label{ym:sec:re}

\subsection{General remarks}
\label{ym:sec:re-1}

In Section \ref{ym:sec:pr:gen3}, a variant of strong deformation was presented, which consisted of probabilistic averaging of fluctuation fields, see formula \eqref{ya-a-41}. One of the problems associated with this approach was the loss of invariance under the gauge transformations of the background field $B_\mu^a$. It is precisely because of this fact that at the level of the second quantum correction, see Theorem \ref{ya-t1} in Section \ref{ym:sec:pr:gen3-7}, the singular contribution proportional to the classical action $W_{-1}$ was split into several parts $W_{-1}^i$. This indicates a significant complication of the renormalization process, as it essentially restricts the possibilities of using "simple"\footnote{This refers to the exclusion of terms from the expansions for the Green's function, which obviously will not lead to singular components. For examples, see Section \ref{ym:sec:two1-1}.} decompositions for the Green's functions, effectively reducing them to zero.

In this section, an alternative approach to the introduction of regularization is proposed, which consists in deforming the fluctuation fields by acting on them with a special gauge invariant operator depending on $B_\mu^a$. At the same time, in the main order, when decomposed into "small" fields, such an operator reproduces the probabilistic averaging from Section \ref{ym:sec:pr:gen3-1}. Such a deformation will be called weak for two reasons. First of all, it does not lead\footnote{Hypothetically, the deformation can be chosen in such a way that the first loop is also regularized. That is, so that the increase in eigenvalues is "compensated" by a decrease in the regularized spectral density. However, the authors did not manage to study this fact well enough.} to the regularization of the first loop, that is "$\ln\det$". For these purposes, one should use, for example, the regularization of the heat kernel, see Section 6.2 in \cite{30-1-1} and section $\mathrm{B}$ in \cite{ya-22}, or the analytical continuation of the $\zeta$-function \cite{ya-23}. A similar situation occurs in other regularizations, see \cite{ya-24}. Secondly, such regularization additionally deforms the relationship between the quantum action and the density of the quantum equation of motion, which leads to additional difficulties in recalculation.

Nevertheless, the weak deformation approach is very attractive from the point of view of multi-loop calculations. This is primarily due to the fact that the elementary blocks that make up the Feynman diagrams are invariant with respect to the gauge transformations of the background field $B_\mu^a$. The following important properties are an immediate consequence of this fact.
\begin{enumerate}
	\item Logarithmic local contributions are proportional only to $W_{-1}[B]$.
	\item There are no power-law local contributions.
\end{enumerate}
Such constraints significantly simplify the calculation process and allow us to replace the decomposition for Green's functions near the diagonal with simpler functions. This process is described in detail in Section \ref{ym:sec:two1} devoted to the study of "two loops".

\subsection{Regularization}
\label{ym:sec:re-2}

Let us discuss the process of introducing regularization. To do this, we use the formulations of the two methods for the case of strong deformation from Section \ref{ym:sec:pr:gen3-2} and adjust them. This is possible because, according to the construction, the "strong" case should be reproduced in the main order for the case of weak (small) fields.\\

\noindent\textbf{The first way.} Let us formulate the rules for the deformation of Green's functions. To do this, as in Section \ref{ym:sec:pr:gen3-2}, we choose a piecewise continuous function $\omega(\cdot)$ as the averaging kernel with the same properties\footnote{A generalization of the properties can be found in \cite{ya-10,Iv-2024}.}. Next, note that the double averaging of the free Green's function, described by the transition \eqref{ya-a-36} using the integral operator and reducing to the averaging of the oscillating exponent, can be rewritten in operator form as follows\footnote{The right hand side of \eqref{ya-a-36-1} must be understood in the sense of operators. That is, this is the Taylor series by degrees of the operator. Note that only even powers arise in the expansion.}
\begin{equation}\label{ya-a-36-1}
\int_{\mathbb{R}^4}\mathrm{d}^4y\int_{\mathbb{R}^4}\mathrm{d}^4z\,
\omega(|y|)\omega(|z|)\,e^{ik\cdot(x+y/\Lambda+z/\Lambda)}=
\hat{\Omega}^2(r/\Lambda)\Big|_{r^2=A_0(x)}e^{ik\cdot x},
\end{equation}
where $r=|x|$, $A_0(x)=-\partial_{x_\mu}\partial_{x^\mu}$, which in hyperspherical coordinates is equal to $-r^{-3}\partial_rr^3\partial_r$ for the case of spherically symmetric functions, as well as an auxiliary function
\begin{equation}\label{ya-z-9-1}
	\hat{\Omega}(r)=2S_3\int_0^1\mathrm{d}t\,t^2\omega(t)
	J_1(tr)/r.
\end{equation}
Note that formula (20) from \cite{Ivanov-2022} was used for the last transition. Next, returning to the Green's functions, we define deformations as follows
\begin{equation}\label{ya-z-15}
G_{0}^{\Lambda,\,ab}(x,y)=\bigg(\hat{\Omega}^2(r/\Lambda)\Big|_{r^2=M_0(x)}\bigg)^{ac}G_{0}^{cb}(x,y),
\end{equation}
as well as
\begin{align}\label{ya-a-38-3}
	G_{1\sigma\rho}^{\Lambda\,ab}(x,y)=\delta_{\sigma\rho}G_{0}^{\Lambda,\,ab}(x,y)+
	\sum_{k=1}^{+\infty}&2^k\int_{\mathbb{R}^{4\times k}}\mathrm{d}^4z_1\ldots\mathrm{d}^4z_k\,
	G_{0}^{\Lambda,\,ad_1}(x,z_1)\times\\\label{ya-a-38-4}\times&F_{\sigma\mu_1}^{d_1c_1}(z_1)G_{0}^{\Lambda,\,c_1d_2}(z_1,z_2)\times\ldots\times
	F_{\mu_k\rho}^{d_kc_k}(z_k)G_{0}^{\Lambda,\,c_kb}(z_k,y),
\end{align}
It follows from the construction that the functions change covariantly with respect to the gauge transformation \eqref{ya-p-1}, since they are constructed using the gauge invariant operator $M_0^{ab}(x)$, and their expansion near the diagonal for the special case of fields was previously studied in \cite{Ivanov-2022}, see Section 4.1.\\

\noindent\textbf{The second way.} Additionally, we point out an alternative view of the introduction of operator \eqref{ya-z-15}. It is completely similar to the one described in Section \ref{ym:sec:pr:gen3-2}, so let us pay attention to only two key changes.

\begin{enumerate}
	\item The fluctuation and ghost fields are changed not by the averaging operator \eqref{ya-a-41}, but by the operator function \eqref{ya-z-9-1}. For example, deformed fluctuations have the form
\begin{equation}\label{ya-re-1}
a_\mu^{\Lambda,a}(x)=\bigg(\hat{\Omega}(r/\Lambda)\Big|_{r^2=M_0(x)}\bigg)^{ab}a_\mu^b(x).
\end{equation}
The fields $c^a,\,\bar{c}^a$ are modified in the same way.
	\item The substitution from the first point applies to all fields except those in quadratic form with the operator $M_0$, that is, in 
\begin{equation*}
	\frac{1}{2}
	\int_{\mathbb{R}^4}\mathrm{d}^4x\,
	a_\mu^a(x)M_0^{ab}(x)a_\mu^b(x)+
	\int_{\mathbb{R}^4}\mathrm{d}^4x\,\bar{c}^a(x)M_0^{ab}(x)c^b(x).
\end{equation*}
\end{enumerate}
By decomposing the action into a perturbative series, we can make sure that the last rules are equivalent to the transition to the deformed Green's functions \eqref{ya-z-15} and \eqref{ya-a-33}.\\

\noindent\textbf{About the determinant.} As was already noted in Section \ref{ym:sec:re-2}, the proposed approach with a covariant deforming operator does not guarantee the regularization of the first correction with "$\ln\det$". For these purposes, additional regularization should be used, which can be introduced either within the framework of the heat kernel method, for example, using an analytical continuation or a cutoff, or by simple division and multiplication (deformation of the measure of the functional integral). Since it is the singular parts that are studied in this paper, and the singularity in the first correction does not depend on regularization, then for further study it is enough to have only the relation 
\begin{equation}\label{ya-z-21}
\ln\det\big(G_{0}/G_{0}|_{B=0}\big)\big|_{\mathrm{reg.}}\stackrel{\mathrm{s.p.}}{=}
-\frac{c_2L}{96\pi^2}W_{-1}[B].
\end{equation}
As expected, the contribution does not contain power-law parts $(\sim\Lambda)$, since the value is invariant with respect to the gauge transformations of the background field.\\

\noindent\textbf{Expansion.} Next, we note that it is not necessary to use the background field $B_\mu^a$ to introduce the last regularization, but another $r_\mu^a$ can be selected, which changes under the gauge transformations in the same way as the background field. Of course, when calculating singular parts, it is the specific choice of $r_\mu^a=B_\mu^a$ that is important for the quantum action. However, when trying to calculate the variation in the field from the deforming operator, it becomes necessary to distinguish the fields. For these purposes, for example, it is convenient to choose $r_\mu^a=B_\mu^a+\varepsilon_\mu^a$, then
\begin{equation*}
\frac{\delta}{\delta r_\mu^a(x)}\bigg|_{r=B}=\frac{\delta}{\delta \varepsilon_\mu^a(x)}\bigg|_{\varepsilon=0}.
\end{equation*}
In this regard, we expand the notation for the regularized quantum action by adding a third argument, which is responsible for the field in the regularizing operator. Thus, the strongly connected part of \eqref{ya-a-17}, regularized by the above procedure with $B_\mu^a\to r_\mu^a$, is further notated by the symbol
\begin{equation}\label{ya-z-20}
W_{\mathrm{reg}}^{\mathrm{sc}}[B,e,r],\,\,\,\mbox{as well as}\,\,\,
W_{\mathrm{reg}}^{\mathrm{sc}}[B]\equiv W_{\mathrm{reg}}^{\mathrm{sc}}[B,B,B].
\end{equation}

\subsection{Situation around diagonal}
\label{ym:sec:re-5}
\noindent\textbf{Problematic.} In the case of weak deformation, the variation of the quantum action on the diagonal with respect to the background field, see \eqref{ya-z-16}, becomes more complicated, since the deforming operator additionally depends on the auxiliary field $r_\mu^a(x)$. Thus, the first variation of the action on the diagonal, which in this case is characterized by the condition $r_\mu^a=e_\mu^a=B_\mu^a$, described by the formula
\begin{equation}\label{ya-zzz-1}
	\frac{\delta W^{\mathrm{sc}}_{\mathrm{reg}}[B]}{\delta B_{\mu}^a(x)}=
	\frac{\delta W^{\mathrm{sc}}_{\mathrm{reg}}[B,e,r]}{\delta B_\mu^a(x)}\Bigg|_{r=e=B}+
	\frac{\delta W^{\mathrm{sc}}_{\mathrm{reg}}[B,e,B]}{\delta e_\mu^a(x)}\Bigg|_{e=B}+
	\frac{\delta W^{\mathrm{sc}}_{\mathrm{reg}}[B,B,r]}{\delta r_\mu^a(x)}\Bigg|_{r=B}
	,
\end{equation}
in which only the first term is equal to the density of the regularized quantum equation of motion. Once again, note that since the action depends on the regularizing operator $\hat{\Omega}$, the inequality generally holds
\begin{equation*}
\frac{\delta W^{\mathrm{sc}}_{\mathrm{reg}}[B,e,B]}{\delta B_\mu^a(x)}\Bigg|_{e=B}\neq
\frac{\delta W^{\mathrm{sc}}_{\mathrm{reg}}[B,e,r]}{\delta B_\mu^a(x)}\Bigg|_{r=e=B}\equiv Q_\mu^a|_{\mathrm{reg.}}[B](x).
\end{equation*}
Thus, in order to recalculate the renormalization values from the quantum action to the density of the equation of motion, it is necessary to track not only the variation with respect to $e_\mu^a$, the singular part of which, in the case of strong deformation, determines the auxiliary current $j_\mu^a(x)$, see \eqref{ya-a-181}, but also a variation with respect to $r_\mu^a$. In this formulation, the task becomes significantly more complicated, so it is easier to calculate not both variations, but the singularities for the quantum equation directly. That is, in addition to renormalizing the quantum action, it is also necessary to study the renormalization of the object $Q_\mu^a|_{\mathrm{reg.}}[B](x)$ in parallel.\\ 

\noindent\textbf{Renormalization.} Taking into account the last remarks, we will clarify what exactly will be understood by the renormalization of the quantum action and the equation of motion. First, consider the fields $r_\mu^a$ and $e_\mu^a$, equal\footnote{In the general case, there is no equality. This choice is convenient for further calculation.} to each other and close to the background field, that is 
\begin{equation*}
r_\mu^a(x)=e_\mu^a(x)=B_\mu^a(x)+\varepsilon_\mu^a.
\end{equation*}
Then the renormalization of the pair $W^{\mathrm{sc}}_{\mathrm{reg}}[B]$ and $Q_\mu^a|_{\mathrm{reg.}}[B]$ consists in shifting the coupling constant $g\to g_\Lambda$, which is expressed as a series in terms of a renormalized constant
\begin{equation*}
	\frac{1}{g^2_\Lambda}=\frac{1}{g_{\mathrm{ren}}^2}+\sum_{k=0}^{+\infty}a_k^{\phantom{1}}(\Lambda)g_{\mathrm{ren}}^{2k},
\end{equation*}
and adding a set of auxiliary terms to the classical action, as was done in the case of strong deformation, see \eqref{ya-a-174}. However, given the invariance with respect to the gauge transformations of the background field, the terms can be written out in a simpler form
\begin{equation*}
	\bigg(\sum_{k=2}^{+\infty}g^{k}\hat{\Gamma}_k^0[B,a]\bigg)+
	\int_{\mathbb{R}^4}\mathrm{d}^4x\,\varepsilon_\mu^a(x)j_\mu^a(x)
	+\mathcal{O}(\varepsilon^2),
\end{equation*}
where the auxiliary current $j_\mu^a(x)$, taking into account the explicit form for the single vertex $\Gamma_1[a]$ from \eqref{ya-a-7}, is determined by the equality
\begin{equation*}
	\int_{\mathbb{R}^4}\mathrm{d}^4x\,\varepsilon_\mu^a(x)j_\mu^a(x)=
	\frac{\Gamma_1[\varepsilon]}{4}\sum_{k=1}^{+\infty}g^{2k-2}w_k(\Lambda)
	+\sum_{k=2}^{+\infty}g^{k}\hat{\Gamma}_k^1[B,a].
\end{equation*}
Let us pay attention to what values are found in this case after performing the renormalization:
\begin{align*}
W^{\mathrm{sc}}_{\mathrm{reg}}[B]&\longrightarrow
\{a_k,\hat{\Gamma}_{k+2}^0\}_{k=0}^{+\infty};
\\
\{a_k,\hat{\Gamma}_{k+2}^0\}_{k=0}^{+\infty}\,\,+\,\,
Q_\mu^a|_{\mathrm{reg.}}[B](x)&\longrightarrow
\{w_{k},\hat{\Gamma}_{k+1}^1\}_{k=1}^{+\infty},
\end{align*}
where, at the last step, a set of renormalization data is obtained from considering the difference between the two studied objects
\begin{equation*}
\frac{\delta W^{\mathrm{sc}}_{\mathrm{reg}}[B,B+\varepsilon,B+\varepsilon]}{\delta \varepsilon_\mu^a(x)}\Bigg|_{\varepsilon=0}
=
\frac{\delta W^{\mathrm{sc}}_{\mathrm{reg}}[B]}{\delta B_{\mu}^a(x)}-Q_\mu^a|_{\mathrm{reg.}}[B](x).
\end{equation*}
Moreover, the sequence is similar to the one in Section \ref{ym:sec:pr:gen3-5}, see formulas \eqref{ya-a-182} and \eqref{ya-a-183}. The ansatz looks much simpler than in formula \eqref{ya-a-174}, since the regularization is introduced covariantly. It is for this reason that the functionals $\hat{W}_k^1$ are explicitly written out as $w_k\Gamma_1[\varepsilon]$. Thus, it is necessary to look for the proportionality coefficient $w_k$, rather than the whole functional and its form.

\subsection{What do we calculate?}
\label{ym:sec:re-3}
Let us introduce the expansion coefficients for the renormalized quantum action on the diagonal and its first variation on the second argument, also on the diagonal. Formally, they are the same as in formulas \eqref{ya-z-17} and \eqref{ya-z-18}, so for convenience we will use the same notation. To clarify, in the "weak" case there is also a third argument, it is chosen as $r_\mu^a=B_\mu^a$. The coefficients themselves have the same form\footnote{Note that the operator $\mathbb{H}$ connects vertices with the regularized Green's functions. In the "strong" and "weak" cases, they are different, see for comparison \eqref{ya-a-39} and \eqref{ya-z-15}.}, as in \eqref{ya-a-49}, \eqref{ya-a-163}, and \eqref{ya-a-178}.

The main task of this section is to renormalize the quantum action on the diagonal, taking into account the first two corrections, as well as to demonstrate the fact that the variation with respect to the deforming operator within the framework of the first quantum correction does not contain singular parts. Thus, it is planned to calculate the following values:
\begin{align*}
	\mbox{1-st correction for}\,\,W_{\mathrm{ren}}^{\mathrm{sc}}&\,\,\longrightarrow\,\,\mbox{answer for}\,\,
	a_0;
	\\
	\mbox{1-st correction for}\,\,\dot{W}_{\mathrm{ren}}^{\mathrm{sc}}&\,\,\longrightarrow\,\,\mbox{answer for}\,\,w_1;
	\\
	\mbox{2-nd correction for}\,\,W_{\mathrm{ren}}^{\mathrm{sc}}&\,\,\longrightarrow\,\,\mbox{answers for}\,\,
	\hat{\Gamma}_2^0\,\,\mbox{and}\,\,a_1.
\end{align*}
And also additionally prove equality in the sense of formal relations by the coupling constant
\begin{equation}\label{ya-z-19}
	\frac{\delta W^{\mathrm{sc}}_{\mathrm{reg}}[B]}{\delta B_{\mu}^a(x)}-
	Q_\mu^a|_{\mathrm{reg.}}[B](x)-
	\frac{\delta W^{\mathrm{sc}}_{\mathrm{reg}}[B,e,B]}{\delta e_\mu^a(x)}\Bigg|_{e=B}
	\stackrel{\mathrm{s.p.}}{=}\mathcal{O}(g^2).
\end{equation}

\subsection{Results}
\label{ym:sec:re-4}

\begin{theorem}\label{ya-t2}
Let us assume that all the assumptions formulated above in Section \ref{ym:sec:pr:gen3} are fulfilled. Let also $L=\ln(\Lambda/\sigma)$, where $\Lambda$ is a regularization parameter and $\sigma>0$ is an auxiliary fixed parameter to make combinations dimensionless. We additionally define the functionals
\begin{equation*}
	S_2[B]=\int_{\mathbb{R}^4}\mathrm{d}^4x\,B_\mu^a(x)B_\mu^a(x),\,\,\,
	S_{\mathrm{f}}[a,e]=\int_{\mathbb{R}^4}\mathrm{d}^4x\,\Big(D_\mu^{ab}(x)a_\mu^b(x)\Big)\Big(\mathfrak{D}_\nu^{ac}(x)a_\nu^c(x)\Big),
\end{equation*}
and several auxiliary numbers depending on the regularized free fundamental solution $R_0^1(x)$, see \eqref{ya-a-37},
\begin{equation*}
	\frac{\rho_0}{4\pi^2}=R_0^1(0),\,\,\,
	\frac{\rho_3}{4\pi^2}=\int_{\mathrm{B}_{1}}\mathrm{d}^4x\,R_0^1(x)A_0^{\phantom{1}}(x)R_0^1(x),\,\,\,
	\frac{2\rho_4}{\pi^2}=\int_{\mathrm{B}_{1}}\mathrm{d}^4x\,R_0^1(x)|x|^2A_0^{\phantom{1}}(x)R_0^1(x),
\end{equation*}
\begin{equation*}
	\frac{\rho_2}{8\pi^2}=
	\int_{\mathrm{B}_{1}}\mathrm{d}^4x\,A_0^{\phantom{1}}(x)R_0^1(x)
	\int_{\mathbb{R}^4}\mathrm{d}^4y\,\Big(R_0^1(x-y)R_0^1(y)-R_0^1(y)R_0^1(y)\Big).
\end{equation*}
Let also the numbers $\tilde{a}$ and $\{\tilde{a}_i,\tilde{b}_i\}_{i=0,1}$ denote arbitrary finite elements from $\mathbb{R}$. Then the presented equalities \eqref{ya-a-184} and \eqref{ya-a-184-1} lead to the following coefficients and auxiliary vertices
	
	\begin{fleqn}
		\begin{equation*}
			1)\,\,W_0^{\Lambda}\stackrel{\mathrm{s.p.}}{=}0\,\,\Longrightarrow\,\,
			a_0(\Lambda)=\frac{11c_2L}{24\pi^2}+\tilde{a}_0.
		\end{equation*}
	\end{fleqn}
	\begin{fleqn}
		\begin{equation*}
			2)\,\,\dot{W}_0^{\Lambda}\stackrel{\mathrm{s.p.}}{=}0\,\,\Longrightarrow\,\,
			\omega_1(\Lambda)=\frac{c_2L}{2\pi^2}+\tilde{b}_0.
		\end{equation*}
	\end{fleqn}
	\begin{fleqn}
		\begin{align*}
			3)\,\,W_1^{\Lambda}\stackrel{\mathrm{s.p.}}{=}0\,\,\Longrightarrow\,\,
			\hat{\Gamma}_2^0[B]&=\frac{\Lambda^2c_2(2\rho_3-3\rho_0)}{8\pi^2}
			S_2[a]-\bigg(\frac{5Lc_2}{48\pi^2}+\tilde{b}_1\bigg)S_{\mathrm{f}}[a,B],\\
			a_1(\Lambda)&=\frac{c_2^2L}{(4\pi)^4}\Big(175/9+2^47\rho_2/3-2^7\rho_4\Big)+\tilde{a}_1.
		\end{align*}
	\end{fleqn}
\end{theorem}
\vspace{3mm}
\noindent\textbf{Comment 1.} As in the case of the result from Theorem \ref{ya-t1}, the first coefficient $a_0$ has the standard value\footnote{See the links in the comments to Theorem \ref{ya-t1}.}. The second coefficient $a_1$ has a different value. In Section \ref{ym:sec:two1-4}, a detailed comparison is made with the case of dimensional regularization. It was shown that the answer for the coefficient is constructed from a finite set of master-integrals. They are listed in Table \ref{ya:table:2}. At the same time, Table \ref{ya:table:3} shows the results for individual diagrams and their linear combinations. It turned out that in the case of dimensional regularization, the coefficients with which the master-integrals enter are deformed, while in the case of cutoff, the integrals themselves are deformed. This observation is complemented by the fact that, in the case of cutoff, the singularity for each master-integral can be divided into the main part, which is transferred to the dimensional case, and the auxiliary correction part, which depends on the regularizing function.

\vspace{3mm}
\noindent\textbf{Comment 2.} The coefficient $w_1$ has the standard form, a more detailed comparison is made in a similar comment for Theorem \ref{ya-t1}.

\vspace{3mm}
\noindent\textbf{Comment 3.} The counter-vertex $S_{\mathrm{f}}[a,B]$ also has the standard form and is discussed in a similar comment for Theorem \ref{ya-t1}.

\vspace{3mm}
\noindent\textbf{Comment 4.} Note that the counter-term $S_2[B]$ with power singularity $\Lambda^2$ is missing, which is a consequence of the covariance of regularized objects. However, the counter-vertex $S_2[a]$ remains and has the same form as in Theorem \ref{ya-t1}. Such a term does not violate gauge invariance, since the field $a_\mu^a$ transforms according to a different law, see \eqref{ya-p-4}.

\vspace{3mm}
\noindent\textbf{Proof.} In many ways, the calculation repeats the process of verifying the statements of Theorem \ref{ya-t1}, so we give only auxiliary comments. Following the above description, we divide the explanation into three parts.\\

\noindent\underline{The first part.} The calculation of the singular part for the "first loop" is performed in two stages. First, it is necessary to take advantage of the fact that the singular contribution for the determinant of the scalar operator $M_0^{ab}$ under conditions of weak deformation is given by equation \eqref{ya-z-21}. Further, following the logic of counting for the "strong" case, it is convenient to represent the singular part for a vector operator as the sum of the part depending on the scalar operator and the contribution consisting of potentials, see \eqref{ya-a-68}. Considering the fact that, in the main order, the Green's functions for both regularizations coincide in construction, see \eqref{ya-a-39} and \eqref{ya-z-15}, then the answer is dictated by formula \eqref{ya-a-69} without a power singularity with $S_2[B]$. Considering the linear combination on the right hand side of \eqref{ya-a-49}, we get the stated answer.\\

\noindent\underline{The second part.} In this case, it is necessary to calculate the singular part for the diagrams $-\mathbb{H}_0^{\mathrm{sc}}(V_1)$ and $-\mathbb{H}_0^{\mathrm{sc}}(V_2)$, see the proof of Theorem \ref{ya-t1} and the definition \eqref{ya-a-178}. In this case, an approach with direct calculation of variations for determinants, as was done in Section \ref{ym:sec:one}, will not work, since the deforming operator also depends on the background field. Therefore, at first glance, an explicit calculation of both diagrams is necessary. However, note that the singular part for the second diagram, after using the expansion near the diagonal in the form \eqref{ya-z-11} and the definitions \eqref{ya-z-22} and \eqref{ya-z-23}, is written out as
\begin{equation*}
\mathbb{H}_0^{\mathrm{sc}}(V_2)\stackrel{\mathrm{s.p.}}{=}-\mathbb{H}_0^{\mathrm{sc}}(V_1)-c_2\theta(0)
\int_{\mathbb{R}^4}\mathrm{d}^4x\,
\varepsilon_\mu^{a}(x)D_\nu^{ab}(x)F_{\nu\mu}^{b}(x),
\end{equation*}
where notation \eqref{ya-a-76} was used for the special function $\theta(\cdot)$. Therefore, the combination of interest from \eqref{ya-a-178}, taking into account \eqref{ya-a-98}, is rewritten as follows
\begin{equation*}
-\mathbb{H}_0^{\mathrm{sc}}(V_1)-\mathbb{H}_0^{\mathrm{sc}}(V_2)
\stackrel{\mathrm{s.p.}}{=}-\frac{c_2L}{8\pi^2}\Gamma_1[\varepsilon]=-\frac{1}{4}\hat{W}_1^1[B,\varepsilon].
\end{equation*}
Further, taking into account equality $\hat{W}_1^1[B,\varepsilon]=w_1\Gamma_1[\varepsilon]$, we get the stated answer.\\

\noindent\underline{The third part.} In this case, the proof process is almost identical to the third point of the proof of Theorem \ref{ya-t1}. The main difference lies in the decomposition of the Green's functions near the diagonal, see Section \ref{ym:sec:two1-1}. This leads to the need to transform a number of small calculations, which is described in Section \ref{ym:sec:two1-2}, after which the linear combination of the main diagrams takes the form \eqref{ya-z-13}. An additional counter-diagram is given in Section \ref{ym:sec:two1-4}, while the final answers for the local part and counter-diagrams are written out in Tables \ref{ya:table:1} and \ref{ya:table:3}. The theorem has been proved.

\subsection{But what if so?}
\label{ym:sec:re-6}

From a computational point of view, the consistency of renormalization of the quantum action with renormalization of the density from the quantum equation of motion, see Sections \ref{ym:sec:pr:gen3-5} and \ref{ym:sec:re-5}, significantly complicates the entire process. Indeed, it is actually necessary to renormalize not one line, but two. And if we keep consistency with a large number of variational derivatives of the action, then the task becomes even more difficult. In this regard, there is a desire to simplify the renormalization algorithm, limiting itself to working only with the quantum action. But what happens in this case? Let us briefly describe the sequence of actions. We will consider the quantum action directly on the diagonal, that is
\begin{equation*}
W^{\mathrm{sc}}_{\mathrm{reg}}[B]\equiv W^{\mathrm{sc}}_{\mathrm{reg}}[B,B,B]
\,\,\,\mbox{Instead of}\,\,\,
W^{\mathrm{sc}}_{\mathrm{reg}}[B,e,r].
\end{equation*}
Next, we replace the standard strongly connected\footnote{This fact is verified using functional relations from Section \ref{ym:sec:pr:gen3-2}.} regularized Green's functions\footnote{It is worth noting here that we are talking about functions with cutted external lines off (or "imputed legs").}
\begin{equation*}
\Gamma_{\mu_1\ldots\mu_k}^{a_1\ldots a_k}(x_1,\ldots,x_k)=
\bigg(
\frac{\delta}{\delta B_{\mu_k}^{a_k}(x_k)}\cdot\ldots\cdot
\frac{\delta}{\delta B_{\mu_1}^{a_1}(x_1)}W^{\mathrm{sc}}_{\mathrm{reg}}[B,e,r]
\bigg)\Bigg|_{r=e=B},
\end{equation*}
where $k\in\mathbb{N}$, with deformed analogues of the form
\begin{equation*}
\hat{\Gamma}_{\mu_1\ldots\mu_k}^{a_1\ldots a_k}(x_1,\ldots,x_k)=
\frac{\delta}{\delta B_{\mu_k}^{a_k}(x_k)}\cdot\ldots\cdot
\frac{\delta}{\delta B_{\mu_1}^{a_1}(x_1)}W^{\mathrm{sc}}_{\mathrm{reg}}[B].
\end{equation*}
For example, in this case, the density of the quantum equation of motion is replaced by the total derivative
\begin{equation*}
	Q_\mu^a|_{\mathrm{reg.}}[B](x)=
	\frac{\delta W^{\mathrm{sc}}_{\mathrm{reg}}[B,e,r]}{\delta B_\mu^a(x)}\Bigg|_{r=e=B}
	\longrightarrow\,\,
	\frac{\delta W^{\mathrm{sc}}_{\mathrm{reg}}[B]}{\delta B_\mu^a(x)}.
\end{equation*}
Next, let us pay attention to the fact that the connected quantum action can be constructed as a sum in which the first term is the strongly connected action, and the next parts are constructed by density from the quantum equation and tree diagrams, at the nodes of which strongly connected Green's functions sit, see for example \eqref{ya-a-46}. Thus, $W^{\mathrm{c}}_{\mathrm{reg}}[B]$ is the functional that depends on the set $\{\Gamma_{\ldots}^{\ldots}\}$. Let us define the deformed connected action $\hat{W}^{\mathrm{c}}_{\mathrm{reg}}[B]$ as the functional $W^{\mathrm{c}}_{\mathrm{reg}}[B]$ in which the substitution was performed $\{\Gamma_{\ldots}^{\ldots}\}\to\{\hat{\Gamma}_{\ldots}^{\ldots}\}$. As a result, a deformed connected action is constructed by the strongly connected action on the diagonal $W^{\mathrm{sc}}_{\mathrm{reg}}[B]$. In this case, the solution $B_{d,\mu}^a$ of the deformed quantum equation of motion 
\begin{equation*}
\frac{\delta W^{\mathrm{sc}}_{\mathrm{reg}}[B]}{\delta B_\mu^a(x)}\Bigg|_{B=B_d}=0,
\end{equation*}
generally speaking, does not coincide with the solution $B_{f,\mu}^a$ of the standard quantum equation. Nevertheless, the following relations are valid in both situations
\begin{align*}
W^{\mathrm{c}}_{\mathrm{reg}}[B_f]&=W^{\mathrm{sc}}_{\mathrm{reg}}[B_f],\\
\hat{W}^{\mathrm{c}}_{\mathrm{reg}}[B_d]&=W^{\mathrm{sc}}_{\mathrm{reg}}[B_d].
\end{align*}
At the same time, there is no equality of strongly connected parts in both situations. It should be noted here that from the point of view of the procedure for finding the extremum, the solution $B_{d,\mu}^a$ is more natural, since it minimizes the functional rather than its part. Thus, forgetting about particular variations, the issue of renormalization narrows down to the functional $W^{\mathrm{sc}}_{\mathrm{reg}}[B]$, which transforms to $W^{\mathrm{sc}}_{\mathrm{ren}}[B]$, and according to which $\hat{W}^{\mathrm{c}}_{\mathrm{ren}}[B]$ is constructed, taking into account the rules described above.

Despite the obvious gain in the complexity of renormalization, there are some open questions along the way. First, are the solutions of $B_{f,\mu}^a$ and $B_{d,\mu}^a$ very different? After all, in the absence of regularization, they must match. Therefore, it is unclear whether their difference is small in any sense (according to $\Lambda$). Secondly, it is not at all clear whether the new renormalized connected action has a formulation using a functional integral. In other words, is it possible to supplement (extend) the regularized classical action \eqref{ya-z-25} in such a way that it reproduces $\hat{W}^{\mathrm{c}}_{\mathrm{ren}}[B]$?

\section{One loop: strong case}
\label{ym:sec:one}
\subsection{Scalar operator}
Let us start with the determinant for the operator $\mathfrak{G}_{0}$. The perturbative formula for such an object is written out as, for reference, see \cite{Vas-98} and the scalar case in Section 4 of \cite{Iv-2024-1},
\begin{align}\nonumber
	\ln\det\big(\mathfrak{G}_{0}/\mathfrak{G}_{0}|_{e=B=0}\big)\big|_{\mathrm{reg.}}=
	\sum_{k=1}^{+\infty}\frac{(-1)^k}{k}\int_{\mathbb{R}^{4\times k}}\mathrm{d}^4&z_1\ldots\mathrm{d}^4z_k\,
	\big(\mathfrak{M}_{0}^{ac_1}(z_1)+\delta^{ac_1}\partial_{z_1^\nu}\partial_{z_1^\nu}\big)R_0^{\Lambda}(z_1-z_2)\times\ldots\\\label{ya-a-48}\times&
	\big(\mathfrak{M}_{0}^{c_ka}(z_k)+\delta^{c_kb}\partial_{z_k^\nu}\partial_{z_k^\nu}\big)R_0^{\Lambda}(z_k-y)\big|_{y=z_1}.
\end{align}
The main task of this section consists of two calculations.
\begin{enumerate}
	\item Find the singular component for the case $\mathfrak{G}_0=G_0$.
	\item Find the linear contribution with respect to the field $e_\mu^a$ to the singular component for $\mathfrak{G}_0$.
\end{enumerate}
\underline{\textbf{The first part.}} It is convenient to divide this computational task into several parts. First, we write out all possible combinations appearing in expansion \eqref{ya-a-48} for $\mathfrak{G}_0=G_0$, according to the degrees of the background field. To do this, we define an auxiliary function of the form
\begin{equation}\label{ya-a-51}
\hat{R}_0^{\Lambda}(x)=R_0^{\Lambda}(x)\eta(|x|,\sigma),
\end{equation}
where $\eta(\,\cdot\,,\sigma)$ is a monotonic smooth function that is equal to one on $[0,1/\sigma]$ and equal to zero on the semi-infinite interval $[1/\sigma+1/\Lambda,+\infty)$. Next, we perform calculations for each degree of the background field.\\

\noindent\textbf{Fourth degree.} Let us write out all the possible coefficients for the integral
\begin{equation}\label{ya-a-53}
\frac{1}{4}\int_{\mathbb{R}^4}\mathrm{d}^4z_1\,B_\mu^{ab}(z_1)B_\nu^{bc}(z_1)B_\sigma^{cd}(z_1)B_\rho^{da}(z_1).
\end{equation}
First, we note that the singular component can only be given by parts with the functions $\hat{R}_0^{\Lambda}(x)$. Next, after replacing $R_0^{\Lambda}(x)\to\hat{R}_0^{\Lambda}(x)$, it is necessary to reexpand the smooth fields near the point $z_1$, and then perform a shift of $z_i\to z_i+z_1$ for all other variables. Then the singular contributions of interest are factorized by the product of the finite functional \eqref{ya-a-53}, which depends on the background field, by a singular coefficient constructed using the functions $\hat{R}_0^{\Lambda}(x)$. By direct calculation, we can verify that it is possible for five combinations to appear as singular coefficients
\begin{align}
2^4&\int_{\mathbb{R}^{4\times3}}\mathrm{d}^4z_2\mathrm{d}^4z_3\mathrm{d}^4z_4\,
\hat{R}_0^{\Lambda}(z_2)\hat{R}_0^{\Lambda}(z_2-z_3)\hat{R}_0^{\Lambda}(z_3-z_4)
\partial_{z_4^{\mu}}\partial_{z_4^{\nu}}\partial_{z_4^{\sigma}}\partial_{z_4^{\rho}}\hat{R}_0^{\Lambda}(z_4), \\
-\delta_{\mu\nu}2^43^{-1}&\int_{\mathbb{R}^{4\times2}}\mathrm{d}^4z_2\mathrm{d}^4z_3\,
\hat{R}_0^{\Lambda}(z_2)\hat{R}_0^{\Lambda}(z_2-z_3)\partial_{z_3^{\sigma}}\partial_{z_3^{\rho}}\hat{R}_0^{\Lambda}(z_3), \\
-\delta_{\nu\sigma}2^43^{-1}&\int_{\mathbb{R}^{4\times2}}\mathrm{d}^4z_2\mathrm{d}^4z_3\,
\hat{R}_0^{\Lambda}(z_2)\hat{R}_0^{\Lambda}(z_2-z_3)\partial_{z_3^{\mu}}\partial_{z_3^{\rho}}\hat{R}_0^{\Lambda}(z_3), \\
-\delta_{\sigma\rho}2^43^{-1}&\int_{\mathbb{R}^{4\times2}}\mathrm{d}^4z_2\mathrm{d}^4z_3\,
\hat{R}_0^{\Lambda}(z_2)\hat{R}_0^{\Lambda}(z_2-z_3)\partial_{z_3^{\mu}}\partial_{z_3^{\nu}}\hat{R}_0^{\Lambda}(z_3), \\
\delta_{\mu\nu}\delta_{\sigma\rho}2\,&\int_{\mathbb{R}^4}\mathrm{d}^4z_2\, 
\hat{R}_0^{\Lambda}(z_2)\hat{R}_0^{\Lambda}(z_2).
\end{align}
Considering the fact that the function from \eqref{ya-a-51} is smooth and has a compact support, the derivative flips performed above are quite reasonable. Let us use spherical\footnote{Next, the symbol $\mathbb{S}^3$ denotes a sphere of unit radius in $\mathbb{R}^4$ centered at the origin.} averaging formulas in four-dimensional space in the last equations 
\begin{equation}\label{ya-a-52}
\int_{\mathbb{S}^3}\mathrm{d}^3\sigma(\hat{x})\,
\hat{x}^\mu\hat{x}^\nu=\frac{S_3}{4}\delta^{\mu\nu}
,\,\,\,
\int_{\mathbb{S}^3}\mathrm{d}^3\sigma(\hat{x})\,
\hat{x}^\mu\hat{x}^\nu\hat{x}^\sigma\hat{x}^\rho=\frac{S_3}{4!}\Big(\delta^{\mu\nu}\delta^{\sigma\rho}+\delta^{\mu\sigma}\delta^{\nu\rho}+\delta^{\mu\rho}\delta^{\nu\sigma}\Big),
\end{equation}
where $\hat{x}^\mu=x^\mu/|x|$, as well as asymptotic decompositions\footnote{The main order can be found using the following set of transformations: scaling variables by $1/\Lambda$, differentiation by $\Lambda$, reverse change of variables, transition to the limit of $\Lambda\to+\infty$, as well as a final explicit calculation of remaining simplified integral.} for the following integrals
\begin{align}
&\int_{\mathbb{R}^{4\times3}}\mathrm{d}^4z_2\mathrm{d}^4z_3\mathrm{d}^4z_4\,
\hat{R}_0^{\Lambda}(z_2)\hat{R}_0^{\Lambda}(z_2-z_3)A_0(z_3)\hat{R}_0^{\Lambda}(z_3-z_4)
A_0(z_4)\hat{R}_0^{\Lambda}(z_4)\stackrel{\mathrm{s.p.}}{=}\frac{L}{8\pi^2}, \\
&\int_{\mathbb{R}^{4\times2}}\mathrm{d}^4z_2\mathrm{d}^4z_3\,
\hat{R}_0^{\Lambda}(z_2)\hat{R}_0^{\Lambda}(z_2-z_3)A_0(z_3)\hat{R}_0^{\Lambda}(z_3)\stackrel{\mathrm{s.p.}}{=}\frac{L}{8\pi^2}, \\
&\int_{\mathbb{R}^4}\mathrm{d}^4z_2\, 
\hat{R}_0^{\Lambda}(z_2)\hat{R}_0^{\Lambda}(z_2)\stackrel{\mathrm{s.p.}}{=}\frac{L}{8\pi^2},
\end{align}
where $L=\ln(\Lambda/\sigma)$. Then, after substitutions and summation of all contributions, the final combination is written out as
\begin{equation}\label{ya-a-54}
-\frac{c_2L}{96\pi^2}\int_{\mathbb{R}^4}\mathrm{d}^4z\,
f^{aec}B_\mu^{e}(z)B_\nu^{c}(z)
f^{adb}B_\mu^{d}(z)B_\nu^{b}(z).
\end{equation}

\noindent\textbf{Third degree.} Let us do similar calculations for the situation with the third degree of the background field. It is clear that in this case only logarithmic singularities can occur, and instead of the fourth degree of the background field, a derivative appears as well. The main functionality is
\begin{equation}\label{ya-a-55}
\frac{1}{3}\int_{\mathbb{R}^4}\mathrm{d}^4z_1\,\big(\partial_{z_1^\mu}B_\nu^{ac}(z_1)\big)B_\sigma^{cd}(z_1)B_\rho^{da}(z_1).
\end{equation}
In this case, the following integrals play the role of corresponding coefficients
\begin{align}
	-2^2&\int_{\mathbb{R}^{4\times2}}\mathrm{d}^4z_2\mathrm{d}^4z_3\,
\partial_{z_2^{\rho}}\hat{R}_0^{\Lambda}(z_2)
\Big(2(z_2+z_1)^\mu\partial_{z_2^{\nu}}+\delta^{\mu\nu}\Big)\hat{R}_0^{\Lambda}(z_2-z_3)
\partial_{z_3^{\sigma}}\hat{R}_0^{\Lambda}(z_3),\\
	-2^2&\int_{\mathbb{R}^{4\times2}}\mathrm{d}^4z_2\mathrm{d}^4z_3\,
\partial_{z_2^{\sigma}}\hat{R}_0^{\Lambda}(z_2)
\partial_{z_2^{\rho}}\hat{R}_0^{\Lambda}(z_2-z_3)
\Big(2(z_3+z_1)^\mu\partial_{z_3^{\nu}}+\delta^{\mu\nu}\Big)\hat{R}_0^{\Lambda}(z_3).
\end{align}
Note that during the writing out, some terms containing $B_\mu^{cd}(z_1)B_\mu^{da}(z_1)$ were omitted, since they would not lead to a singular contribution in the final expression due to the zeroing of the trace
\begin{equation}\label{ya-a-56}
\big(\partial_{z_1^\nu}B_\nu^{ac}(z_1)\big)B_\mu^{cd}(z_1)B_\mu^{da}(z_1)=0.
\end{equation}
For the same reasons, we can remove the parts with $\delta^{\mu\nu}$-symbol. Additionally, all terms containing the multiplier $z_1$ will also not contribute due to the presence of spherical symmetry in the form
\begin{equation}\label{ya-a-57}
\int_{\mathbb{S}^3}\mathrm{d}^3\sigma(\hat{x})\,
\hat{x}^\mu=0
,\,\,\,
\int_{\mathbb{S}^3}\mathrm{d}^3\sigma(\hat{x})\,
\hat{x}^\mu\hat{x}^\nu\hat{x}^\sigma=0.
\end{equation}
Next, using relations \eqref{ya-a-52} and noting that the following combination leads to zero
\begin{equation}\label{ya-a-58}
	\big(\partial_{z_1^\mu}B_\nu^{ac}(z_1)\big)B_\sigma^{cd}(z_1)B_\rho^{da}(z_1)
	\Big(\delta^{\mu\nu}\delta^{\sigma\rho}+\delta^{\mu\sigma}\delta^{\nu\rho}+\delta^{\mu\rho}\delta^{\nu\sigma}\Big)=0,
\end{equation}
we get only one non-zero coefficient
\begin{equation}\label{ya-a-59}
\delta^{\mu\rho}2^3\int_{\mathbb{R}^{4\times2}}\mathrm{d}^4z_2\mathrm{d}^4z_3\,
\partial_{z_2^{\rho}}\hat{R}_0^{\Lambda}(z_2)
\hat{R}_0^{\Lambda}(z_2-z_3)
\partial_{z_3^{\nu}}\partial_{z_3^{\sigma}}\hat{R}_0^{\Lambda}(z_3)\stackrel{\mathrm{s.p.}}{=}-\delta^{\mu\rho}\delta^{\nu\sigma}\frac{L}{4\pi^2}.
\end{equation}
As a result, we come to the final expression of the form
\begin{equation}\label{ya-a-60}
	-\frac{c_2L}{24\pi^2}
	\int_{\mathbb{R}^4}\mathrm{d}^4z\,\big(\partial_{z^\mu}B_\nu^{a}(z)\big)
	f^{adb}B_\mu^{d}(z)B_\nu^{b}(z).
\end{equation}

\noindent\textbf{Second degree.} In this case, not only logarithmic singularities are possible when two background fields and two derivatives appear, but also power-law contributions containing only two background fields without derivatives. Consider the first case, then for a combination of the form
\begin{equation}\label{ya-a-61}
\frac{1}{2}\int_{\mathbb{R}^4}\mathrm{d}^4z_1\,\big(\partial_{z_1^\mu}B_\nu^{ac}(z_1)\big)
\big(\partial_{z_1^\sigma}B_\rho^{ac}(z_1)\big)
\end{equation}
only the following coefficients are possible
\begin{align}
	\frac{1}{2}&\int_{\mathbb{R}^{4}}\mathrm{d}^4z_2\,
	\partial_{z_2^{\nu}}\hat{R}_0^{\Lambda}(z_2)
	\Big(\partial_{z_2^{\rho}}z_2^\mu z_2^\sigma+z_2^\mu z_2^\sigma\partial_{z_2^{\rho}}\Big)\hat{R}_0^{\Lambda}(z_2),\\ 
	-\frac{1}{2}&\int_{\mathbb{R}^{4}}\mathrm{d}^4z_2\,
	\hat{R}_0^{\Lambda}(z_2)
	\Big(\partial_{z_2^{\rho}}z_2^\mu z_2^\sigma+z_2^\mu z_2^\sigma\partial_{z_2^{\rho}}\Big)\partial_{z_2^{\nu}}\hat{R}_0^{\Lambda}(z_2).
\end{align}
After integration by parts, we come to the representation
\begin{equation}
\frac{1}{2}\int_{\mathbb{R}^{4}}\mathrm{d}^4z_2\,
	\partial_{z_2^{\nu}}\hat{R}_0^{\Lambda}(z_2)
	\Big(
	2\delta^{\rho\mu}z_2^\sigma+
	2\delta^{\rho\sigma}z_2^\mu+
	4z_2^\mu z_2^\sigma\partial_{z_2^{\rho}}\Big)
	\hat{R}_0^{\Lambda}(z_2),
\end{equation}
in which we can already use the averaging formulas from \eqref{ya-a-52}. Let us introduce the notation $C(x)$ for the operator $x^\mu\partial_{x^\mu}$, then using the asymptotic expressions\footnote{They can be obtained by the same method that was used in the analysis of integrals in the case of the fourth degree of the background field.} for auxiliary integrals
\begin{align}
\int_{\mathbb{R}^{4}}\mathrm{d}^4z\,\Big(C(z)\hat{R}_0^{\Lambda}(z)\Big)\hat{R}_0^{\Lambda}(z)
&\stackrel{\mathrm{s.p.}}{=}-\frac{L}{4\pi^2},
\\
\int_{\mathbb{R}^{4}}\mathrm{d}^4z\,\Big(C(z)\hat{R}_0^{\Lambda}(z)\Big)\Big(C(z)\hat{R}_0^{\Lambda}(z_2)\Big)
&\stackrel{\mathrm{s.p.}}{=}+\frac{L}{2\pi^2},
\\
\int_{\mathbb{R}^{4}}\mathrm{d}^4z\,\Big(\partial_{z^\mu}\hat{R}_0^{\Lambda}(z)\Big)|z|^2\Big(\partial_{z^\mu}\hat{R}_0^{\Lambda}(z)\Big)
&\stackrel{\mathrm{s.p.}}{=}+\frac{L}{2\pi^2},
\end{align}
we obtain a singular coefficient in the form
\begin{equation}
\frac{L}{48\pi^2}\Big(-\delta^{\rho\mu}\delta^{\nu\sigma}-
\delta^{\rho\sigma}\delta^{\nu\mu}+2
\delta^{\mu\sigma}\delta^{\nu\rho}\Big).
\end{equation}
Thus, the logarithmic contribution with functional \eqref{ya-a-61} is written as follows
\begin{equation}\label{ya-a-62}
-\frac{c_2L}{48\pi^2}\int_{\mathbb{R}^4}\mathrm{d}^4z\,\big(\partial_{z^\mu}B_\nu^{a}(z)\big)
\big(\partial_{z^\mu}B_\nu^{a}(z)-\partial_{z^\nu}B_\mu^{a}(z)\big).
\end{equation}
Returning to the contribution with a power-law behavior, we get the answer in the form
\begin{equation}\label{ya-a-63}
	\frac{c_2\Lambda^2}{2}\bigg(\int_{\mathbb{R}^4}\mathrm{d}^4z\,B_\nu^{a}(z)
	B_\nu^{a}(z)\bigg)\bigg(\int_{\mathrm{B}_1}\mathrm{d}^4x\,
\hat{R}_0^{1}(x)A_0(x)\hat{R}_0^{1}(x)-2\hat{R}_0^{1}(0)\bigg).
\end{equation}
Next, we use the auxiliary notation from Section \ref{ym:sec:two-vs}
\begin{equation}
\frac{\rho_0}{4\pi^2}=R_0^1(0),\,\,\,
\frac{\rho_3}{4\pi^2}=\int_{\mathrm{B}_{1}}\mathrm{d}^4x\,R_0^1(x)A_0^{\phantom{1}}(x)R_0^1(x).
\end{equation}
Then, performing the final substitutions and summing up all the answers outlined above, we obtain a relation of the form
\begin{equation}\label{ya-a-64}
\ln\det\big(G_{0}/G_{0}|_{B=0}\big)\big|_{\mathrm{reg.}}\stackrel{\mathrm{s.p.}}{=}
	\frac{c_2\Lambda^2(\rho_3-2\rho_0)}{8\pi^2}S_2[B]
		-\frac{c_2L}{96\pi^2}W_{-1}[B].
\end{equation}
\underline{\textbf{The second part.}} The calculation of a contribution linear with respect to the field from the gauge condition can be greatly simplified. To do this, we first write out functional derivatives with respect to the background field $B_\mu^a$ and to the field $e_\mu^a$ for the corresponding determinant at the point $e_\mu^a=B_\mu^a$. Using formula \eqref{ya-a-30} and performing a similar calculation, we obtain
\begin{align}
\frac{\delta}{\delta B_\mu^g(z)}\ln\det(\mathfrak{G}_0|_{\mathrm{reg.}})\bigg|_{e=B}&=-
D_\mu^{cd}(z)G_0^{ad}|_{\mathrm{reg.}}(x,z)\big|_{x=z}
f^{cga},\\
\frac{\delta}{\delta e_\mu^g(z)}\ln\det(\mathfrak{G}_0|_{\mathrm{reg.}})\bigg|_{e=B}&=+
D_\mu^{cd}(z)G_0^{da}|_{\mathrm{reg.}}(z,x)\big|_{x=z}
f^{agc}.
\end{align}
Next, we make two important observations.
\begin{enumerate}
	\item First, their sum is equal to the variation $\ln\det(G_0|_{\mathrm{reg.}})$ with respect to $B_\mu^g(z)$.
	\item Secondly, the functions are equal because of the symmetry of the kernel $G_0^{ab}(x,y)=G_0^{ba}(y,x)$.
\end{enumerate}
Thus, calculating the variation for \eqref{ya-a-64} and multiplying by $1/2$, we get the final answer in the form
\begin{equation}\label{ya-a-67}
\int_{\mathbb{R}^{4}}\mathrm{d}^4z\,e_\mu^g(z)\Bigg(\frac{\delta}{\delta e_\mu^g(z)}\ln\det(\mathfrak{G}_0|_{\mathrm{reg.}})\bigg|_{e=B}\Bigg)\stackrel{\mathrm{s.p.}}{=}
	\frac{c_2\Lambda^2(\rho_3-2\rho_0)}{8\pi^2}\bigg(\int_{\mathbb{R}^4}\mathrm{d}^4z\,e_\nu^{a}(z)
	B_\nu^{a}(z)\bigg)
-\frac{c_2L}{48\pi^2}\Gamma_1[e].
\end{equation}

\subsection{Vector operator}
Let us now consider the determinant for the regularized operator $\mathfrak{G}_{1}$. The perturbative formula for such an object is written out in the same way as
\begin{align}\nonumber
	\ln\det\big(\mathfrak{G}_{1}/\mathfrak{G}_{1}|_{e=B=0}\big)\big|_{\mathrm{reg.}}=
	\sum_{k=1}^{+\infty}\frac{(-1)^k}{k}\int_{\mathbb{R}^{4\times k}}\mathrm{d}^4&z_1\ldots\mathrm{d}^4z_k\,
	\big(\mathfrak{M}_{1\nu\mu_1}^{\,\,ac_1}(z_1)+\delta^{ac_1}\delta_{\nu\mu_1}\partial_{z_1^\nu}\partial_{z_1^\nu}\big)R_0^{\Lambda}(z_1-z_2)\times\\\label{ya-a-185}\ldots\times&
	\big(\mathfrak{M}_{1\mu_k\nu}^{\,\,c_ka}(z_k)+\delta^{c_kb}\delta_{\mu_k\nu}\partial_{z_k^\nu}\partial_{z_k^\nu}\big)R_0^{\Lambda}(z_k-y)\big|_{y=z_1}.
\end{align}
The main task of this section also consists of two calculations.
\begin{enumerate}
	\item Find the singular component for the case $\mathfrak{G}_1=G_1$.
	\item Find the linear contribution with respect to the field $e_\mu^a$ to a singular component for $\mathfrak{G}_1$.
\end{enumerate}
\underline{\textbf{The first part.}} To calculate it, it is convenient to use the previously obtained answer \eqref{ya-a-64} for the scalar operator. Indeed, given the relation $F_{\mu\mu}^{a}(z)=0$, we obtain the following decomposition
\begin{align}\label{ya-a-68}
\ln\det\big(G_{1}/G_{1}|_{B=0}\big)\big|_{\mathrm{reg.}}
\stackrel{\mathrm{s.p.}}{=}&4
\ln\det\big(G_{0}/G_{0}|_{B=0}\big)\big|_{\mathrm{reg.}}\\
+&2\int_{\mathbb{R}^{4\times 2}}\mathrm{d}^4z_1\mathrm{d}^4z_2\,
F_{\mu\nu}^{ac}(z_1)
R_0^{\Lambda}(z_1-z_2)
F_{\nu\mu}^{ca}(z_2)
R_0^{\Lambda}(z_2-z_1).
\end{align}
By shifting the variable in the second term, leaving only the singular part of the coefficient and using the result \eqref{ya-a-64}, we get the answer in the form
\begin{equation}\label{ya-a-69}
\ln\det\big(G_{1}/G_{1}|_{B=0}\big)\big|_{\mathrm{reg.}}\stackrel{\mathrm{s.p.}}{=}
\frac{c_2\Lambda^2(\rho_3-2\rho_0)}{2\pi^2}S_2[B]
+\frac{5c_2L}{24\pi^2}W_{-1}.
\end{equation}
\underline{\textbf{The second part.}} Using the definition of the operator from \eqref{ya-a-33}, we can write the formula for the first functional derivative
\begin{equation}\label{ya-a-70}
\frac{\delta}{\delta e_\mu^g(z)}\ln\det(\mathfrak{G}_1|_{\mathrm{reg.}})\bigg|_{e=B}=
2f^{agc}D_{\nu}^{cd}(z)G_{1\nu\mu}^{\,\,da}(z,x)\big|_{x=z}.
\end{equation}
Next, we note that it is convenient to use the decomposition of the Green's function by degrees of the field strength tensor. In this case, only the following five terms can provide a singular contribution to the right hand side of \eqref{ya-a-70}
\begin{align}
&2f^{agc}D_{\mu}^{cd}(z)G_{0}^{da}(z,x)\big|_{x=z}\\
+\,&8f^{agc}\partial_{z^\nu}\int_{\mathbb{R}^{4\times2}}\mathrm{d}^4z_2\mathrm{d}^4z_3\,
\hat{R}_0^{\Lambda}(z-z_2)
B_\sigma^{cd}(z_2)\partial_{z_2^\sigma}
\hat{R}_0^{\Lambda}(z_2-z_3)
F_{\nu\mu}^{da}(z_3)
\hat{R}_0^{\Lambda}(z_3-x)\big|_{x=z}\\
+\,&8f^{agc}\partial_{z^\nu}\int_{\mathbb{R}^{4\times2}}\mathrm{d}^4z_2\mathrm{d}^4z_3\,
\hat{R}_0^{\Lambda}(z-z_2)
F_{\nu\mu}^{cd}(z_2)
\hat{R}_0^{\Lambda}(z_2-z_3)
B_\sigma^{da}(z_3)\partial_{z_3^\sigma}
\hat{R}_0^{\Lambda}(z_3-x)\big|_{x=z}\\
+\,&4f^{agc}\partial_{z^\nu}\int_{\mathbb{R}^{4}}\mathrm{d}^4z_2\,
\hat{R}_0^{\Lambda}(z-z_2)
F_{\nu\mu}^{ca}(z_2)
\hat{R}_0^{\Lambda}(z_2-x)\big|_{x=z}\\
+\,&4f^{agc}B_\nu^{cd}(z)\int_{\mathbb{R}^{4}}\mathrm{d}^4z_2\,
\hat{R}_0^{\Lambda}(z-z_2)
F_{\nu\mu}^{da}(z_2)
\hat{R}_0^{\Lambda}(z_2-x)\big|_{x=z}.
\end{align}
Note that the first part is equal to the doubled variation of \eqref{ya-a-67} with respect to the field $e_\mu^a$, while the remaining terms are calculated using previously used formulas, see auxiliary integrals for calculating coefficients with the functional \eqref{ya-a-53}. Putting all the parts together, we get the result
\begin{equation}\label{ya-a-72}
\frac{c_2\Lambda^2(\rho_3-2\rho_0)}{4\pi^2}B_\mu^{a}(z)
-\frac{5c_2L}{24\pi^2}D_\nu^{ga}(z)F_{\nu\mu}^a(z),
\end{equation}
which leads to a final answer of the form
\begin{equation}\label{ya-a-73}
\int_{\mathbb{R}^{4}}\mathrm{d}^4z\,e_\mu^g(z)\Bigg(\frac{\delta}{\delta e_\mu^g(z)}\ln\det(\mathfrak{G}_1|_{\mathrm{reg.}})\bigg|_{e=B}\Bigg)\stackrel{\mathrm{s.p.}}{=}
\frac{c_2\Lambda^2(\rho_3-2\rho_0)}{4\pi^2}\bigg(\int_{\mathbb{R}^4}\mathrm{d}^4z\,e_\nu^{a}(z)
B_\nu^{a}(z)\bigg)
+\frac{5c_2L}{24\pi^2}\Gamma_1[e].
\end{equation}

\section{Two loops: strong case}
\label{ym:sec:two}
\subsection{Decomposition of Green's functions}
\label{ym:sec:two-1}

\textbf{Vector function.} For further calculations, it is convenient to use auxiliary decompositions for the Green's function, both for the scalar $\mathfrak{G}_0^\Lambda$ and for the vector $\mathfrak{G}_1^\Lambda$. Hereafter, the superscript $\Lambda$ indicates the presence of regularization. Note that the field $e_\mu^a$, responsible for fixing the gauge condition, is chosen equal to the background field $B_\mu^a$. Therefore, $\mathfrak{G}_0^\Lambda\to G_0^\Lambda$ and $\mathfrak{G}_1^\Lambda\to G_1^\Lambda$, see \eqref{ya-a-34} and \eqref{ya-a-35}. So, we decompose the function $G_{1\mu\nu}^{\Lambda ab}(x,y)$ with respect to the degrees of the field strength tensor with an allocation of singular components in a symmetrical form. Then we get
\begin{equation}\label{ya-a-74}
G_{1\mu\nu}^{\Lambda ab}(x,y)=\delta_{\mu\nu}G_{0}^{\Lambda ab}(x,y)+
\mathcal{N}_{\mu\nu}^{ab}(x,y)+
\mathcal{L}_{\mu\nu}^{ab}(x,y)+
PS_{1\mu\nu}^{\,\,ab}(x,y),
\end{equation}
where each part is determined based on a value of singularity and the degree of the field strength tensor. Let us explain the notation.\\

\noindent1) The contribution of $\delta_{\mu\nu}G_{0}^{\Lambda ab}$ corresponds to the summand proportional to the zero degree of the field strength tensor $F_{\mu\nu}^{a}$, and thus is the Green's function for the scalar operator.\\

\noindent2) The contribution of $\mathcal{N}_{\mu\nu}^{ab}$ corresponds to the summand proportional to the first degree of the field strength tensor $F_{\mu\nu}^{a}$ and containing the first three singular orders of the functional 
\begin{equation}
2\int_{\mathbb{R}^4}\mathrm{d}^4z\,
G_{0}^{\Lambda ac}(x,z)
F_{\mu\nu}^{cd}(z)
G_{0}^{\Lambda db}(z,y).
\end{equation}
It is convenient to fix its explicit form as follows
\begin{align}\label{ya-a-125}
\mathcal{N}_{\mu\nu}^{ab}(x,y)=&\,
\big(F_{\mu\nu}^{ab}(x)+F_{\mu\nu}^{ab}(y)\big)\theta(x-y)\\\nonumber
+&\big(B_\sigma^{ac}(x)F_{\mu\nu}^{cb}(x)+F_{\mu\nu}^{ac}(y)B_\sigma^{cb}(y)\big)2\partial_{x^\sigma}\tau(x-y)+n_{\mu\nu}^{ab}(x,y),
\end{align}
where auxiliary symmetric functions are defined by the equalities
\begin{equation}\label{ya-a-76}
\theta(x)=\int_{\mathrm{B}_{1/\sigma}}\mathrm{d}^4z\,R_0^{\Lambda}(x-z)R_0^{\Lambda}(z),
\end{equation}
\begin{equation}\label{ya-a-78}
	\tau(x)=\int_{\mathrm{B}_{1/\sigma}}\mathrm{d}^4z\int_{\mathrm{B}_{1/\sigma}}\mathrm{d}^4y\,
	\Big(R_0^{\Lambda}(x-z)R_0^{\Lambda}(z-y)R_0^{\Lambda}(y)-R_0^{\phantom{1}}(z)R_0^{\phantom{1}}(z-y)R_0^{\phantom{1}}(y)\Big).
\end{equation}
In this case, the symmetric function $n_{\mu\nu}^{ab}(x,y)$ is constructed in a similar way using $F_{\mu\nu}^a$, as well as either two background fields or a derivative of the background field. In this case, it may be required that for $y\to x$ in the main order, the function $n_{\mu\nu}^{ab}(x,y)$ should have the order $\ln(\Lambda/\sigma)/\Lambda^2$. Its explicit form is not important, since it does not lead to singular contributions due to the property $n_{\mu\mu}^{ab}(x,y)=0$.\\

\noindent3) The contribution of $\mathcal{L}_{\mu\nu}^{ab}$ is equal to the main singular part of the term proportional to the second degree of the field strength tensor, and is equal to
\begin{equation}\label{ya-a-77}
\mathcal{L}_{\mu\nu}^{ab}(x,y)=2
\big(F_{\mu\sigma}^{ac}(x)F_{\sigma\nu}^{cb}(x)+F_{\mu\sigma}^{ac}(y)F_{\sigma\nu}^{cb}(y)\big)\tau(x-y).
\end{equation}
4) The contribution $PS_{1\mu\nu}^{\,\,ab}$ contains all remaining parts. It has two finite derivatives and also depends on the boundary conditions. Considering the symmetry properties of all the previous parts, it can be argued that it has the property $PS_{1\mu\nu}^{\,\,ab}(x,y)=PS_{1\nu\mu}^{\,\,ba}(y,x)$.\\

\noindent{\textbf{Scalar function.}} Similarly, we decompose the Green's function for the ghost fields. However, only that part\footnote{Next decomposition orders must be additionally calculated.} of a singularity decomposition will be presented here, which is necessary for the tasks of this paper, see Section \ref{ym:sec:int}. First, we note that the function $G_0^{\Lambda ab}(x,y)$ can be divided into two parts
\begin{equation}\label{ya-a-128}
G_0^{\Lambda ab}(x,y)=G_{\mathrm{loc}}^{ab}(x,y)+PS_0^{ab}(x,y),
\end{equation}
where the non-linear part of $PS_0$ has two finite derivatives and is an analog of the function $PS_1$, while the first term has a local character, which can be described as follows: the function $G_{\mathrm{loc}}^{ab}(x,y)$ is a finite sum, each term of which is equal to the product of a special function that does not depend on the background field and has the argument $x-y$ by $p(x)+p(y)$, where $p$ is a polynomial of finite degree with respect to the background field and its first three derivatives. It can be shown that in the first orders with respect to singularities, the function $G_{\mathrm{loc}}^{ab}(x,y)$ can be decomposed as follows
\begin{align}\label{ya-a-127}
G_{\mathrm{loc}}^{ab}(x,y)=\delta^{ab}R_0^\Lambda(x-y)&+2B_\mu^{ab}(y)\partial_{x^\mu}\theta(x-y)\\
&+\partial_{y^\nu}B_\mu^{ab}(y)\big(\delta_{\mu\nu}\theta(x-y)+2\kappa_{\nu\mu}(x-y)\big)\\
&+B_\mu^{ac}(y)B_\nu^{cb}(y)\big(\delta_{\mu\nu}\theta(x-y)+4\partial_{x^\mu}\partial_{x^\nu}\tau(x-y)\big)+\ldots,
\end{align}
where an additional auxiliary function has the form
\begin{equation}\label{ya-a-79}
\kappa_{\nu\mu}(x)=\int_{\mathrm{B}_{1/\sigma}}\mathrm{d}^4z\,R_0^{\Lambda}(x-z)
z^\nu\partial_{z^\mu}R_0^{\Lambda}(z).
\end{equation}
Further, considering the symmetric combination, we obtain the relation
\begin{align}\label{ya-a-74-1}
	G_{\mathrm{loc}}^{ab}(x,y)=\delta^{ab}R_0^\Lambda(x-y)&+\big(B_\mu^{ab}(x)
	+B_\mu^{ab}(y)\big)\partial_{x^\mu}\theta(x-y)\\\nonumber
	&+B_\mu^{ac}(y)B_\nu^{cb}(y)\big(\delta_{\mu\nu}\theta(x-y)+4\partial_{x^\mu}\partial_{x^\nu}\tau(x-y)\big)+\ldots
\end{align}
In fact, within the framework of the decomposition, using symmetry was equivalent to replacing one type of singularity with another, that is,
\begin{equation}\label{ya-a-80}
\delta_{\mu\nu}\theta(x)+2\kappa_{\nu\mu}(x)\longleftrightarrow x^\nu\partial_{x^\mu}\theta(x),
\end{equation}
which can be verified by integration by parts. Note that the term proportional to the first degree of the background field has a symmetrical appearance, while the fields for the third term are written out at only one point. This is due to the fact that the main purpose of this work is to study singular components, for which the proposed decomposition is sufficient. When trying to search for explicit\footnote{Recall that this paper describes a deviation from the covariant case. At the same time, a structure of singularities is proposed for a number of counter-terms without calculating explicit coefficients.} coefficients for the vertices of renormalization it also be necessary to represent the terms for the third and fourth degrees of the background field.

Let us perform\footnote{Such transformations are standard in the analysis of diagrams with one and two integration operators. They allow us to express the singularity explicitly. For more information, see \cite{34,Iv-2024-1} for an example.} the shift of the variable $x\to x+y$ and factor the individual terms into a singular part depending on the variable $x$ and a part depending on the background fields at the point $y$. For convenience, we write out in the form of a table the decomposition by "values" of the singularities of the function $G_{\mathrm{loc}}^{ab}(x+y,y)$.
\begin{equation}\label{ya-a-104}
\renewcommand{\arraystretch}{1.2}
\begin{tabular}{|c|c|c|c|c|c|}
\hline
\multicolumn{2}{|c|}{$\sim r^{-2}$}&\multicolumn{2}{|c|}{$\sim r^{-1}$}&\multicolumn{2}{|c|}{$\sim\log|r|+\sim r^0$}\\
\hline
$1$&$R_0^\Lambda(x)$&$B_\mu(y)$&$2\partial_{x^\mu}\theta(x)$&$\partial_{y^\nu}B_\mu(y)$&$x^\nu\partial_{x^\mu}\theta(x)$\\
&&&&$B_\mu(y)B_\nu(y)$&$\delta_{\mu\nu}\theta(x)+4\partial_{x^\mu}\partial_{x^\nu}\tau(x)$\\
\hline
\end{tabular}
\renewcommand{\arraystretch}{1}
\end{equation}
Note that the group indices in the table above are omitted. A similar decomposition must also be written out for the second derivative of the Green's function. To do this, we first note that the covariant derivative can be rewritten as
\begin{equation}\label{ya-a-81}
D_\sigma(x)=\partial_{x^\sigma}+B_\sigma(y)+(x-y)^\eta\partial_{y^\eta}B_\sigma(y)+\ldots,
\end{equation}
where the group indices have been omitted again for convenience. After a series of auxiliary calculations, we obtain the following summary table.

\begin{equation}\label{ya-a-105}
	\renewcommand{\arraystretch}{1.3}
	\begin{tabular}{|c|c|c|c|c|c|}
		\hline
		\multicolumn{6}{|c|}{$D_{\sigma}^{ca}(x)D_\rho^{db}(y)G_{\mathrm{loc}}^{ab}(x,y)|_{x\to x+y}$ }\\
		\hline
		\multicolumn{2}{|c|}{$\sim r^{-4}$}&\multicolumn{2}{|c|}{$\sim r^{-3}$}&\multicolumn{2}{|c|}{$\sim r^{-2}$}\\
		\hline
		$1$&$-\partial_{x^\sigma}\partial_{x^\rho}R_0^\Lambda(x)$&$B_\mu(y)$
		&$-2\partial_{x^\sigma}\partial_{x^\rho}\partial_{x^\mu}\theta(x)$&$\partial_{y^\nu}B_\mu(y)$
		&$-\partial_{x^\sigma}\partial_{x^\rho}x^\nu\partial_{x^\mu}\theta(x)$\\
		&&$B_\sigma(y)$&$-\partial_{x^\rho}R_0^\Lambda(x)$&$B_\mu(y)B_\nu(y)$&$-\delta_{\mu\nu}\partial_{x^\sigma}\partial_{x^\rho}\theta(x)$\\
		&&$B_\rho(y)$&$-\partial_{x^\sigma}R_0^\Lambda(x)$&$B_\mu(y)B_\nu(y)$&$-4\partial_{x^\sigma}\partial_{x^\rho}\partial_{x^\mu}\partial_{x^\nu}\tau(x)$\\
		&&&&$\partial_{y^\rho}B_\mu(y)$&$2\partial_{x^\sigma}\partial_{x^\mu}\theta(x)$\\
		&&&&$B_\sigma(y)B_\mu(y)$&$-\partial_{x^\rho}\partial_{x^\mu}\theta(x)$\\
		&&&&$B_\mu(y)B_\rho(y)$&$-\partial_{x^\sigma}\partial_{x^\mu}\theta(x)$\\
		&&&&$\partial_{y^\eta}B_\sigma(y)$&$-x^\eta\partial_{x^\rho}R_0^\Lambda(x)$\\
		&&&&$B_\sigma(y)B_\rho(y)$&$-R_0^\Lambda(x)$\\
		\hline
	\end{tabular}
	\renewcommand{\arraystretch}{1}
\end{equation}

\subsection{Singular operators}
\label{ym:sec:two-sin}
For example, consider an arbitrary smooth matrix-valued function of two variables $f^{ab}(x,y)$. We assume that it is symmetric, that is, $f^{ab}(x,y)=f^{ba}(y,x)$. The main task of this section is to calculate the singular part of the expansion with respect to the regularizing parameter $\Lambda$ for the values
\begin{equation}\label{ya-a-100}
\mathrm{J}_{\sigma\rho}^{\phantom{1}}[f]=\int_{\mathbb{R}^4}\mathrm{d}^4x\int_{\mathbb{R}^4}\mathrm{d}^4y\,
\mathcal{R}_{\sigma\rho}^{ab}(x,y)f^{ba}(y,x),
\end{equation}
\begin{equation}\label{ya-a-100-1}
	\mathrm{\hat{J}}_{\sigma\rho}[f]=\int_{\mathbb{R}^4}\mathrm{d}^4x\int_{\mathbb{R}^4}\mathrm{d}^4y\,
	\mathcal{\hat{R}}_{\sigma\rho}^{ab}(x,y)f^{ba}(y,x),
\end{equation}
where the kernels of integro-differential operators are given by the equalities
\begin{equation}\label{ya-a-101}
\mathcal{R}_{\sigma\rho}^{ab}(x,y)=f^{acd}\Big(D_\sigma^{ch}(x)D_\rho^{gf}(y)G_{\mathrm{loc}}^{hf}(x,y)\Big)G_{\mathrm{loc}}^{de}(x,y)f^{egb},
\end{equation}
\begin{equation}\label{ya-a-101-1}
	\mathcal{\hat{R}}_{\sigma\rho}^{ab}(x,y)=f^{hcd}G_{\mathrm{loc}}^{cg}(x,y)G_{\mathrm{loc}}^{de}(x,y)f^{egf}
	D_\sigma^{ha}(x)D_\rho^{fb}(y).
\end{equation}
Taking into account the definition for the auxiliary number $\rho_3$ from Section \eqref{ya-a-94}, the answers are written as follows
\begin{align}\label{ya-a-102}
\mathrm{J}_{\sigma\rho}[f]
\stackrel{\mathrm{s.p.}}{=}&-\frac{\delta_{\sigma\rho}\Lambda^2c_2\rho_3}{16\pi^2}
\int_{\mathbb{R}^4}\mathrm{d}^4x\,f^{aa}(x,x)\\\nonumber
&+\frac{Lc_2}{96\pi^2}
\int_{\mathbb{R}^4}\mathrm{d}^4x\,\Big(\big(M_{1\rho\sigma}^{ba}(x)+4D_\rho^{bc}(x)D_\sigma^{ca}(x)\big)f^{ab}(x,y)\Big)\Big|_{y=x},\\\label{ya-a-103-1}
\mathrm{\hat{J}}_{\sigma\rho}[f]
\stackrel{\mathrm{s.p.}}{=}&+
\frac{Lc_2}{8\pi^2}
\int_{\mathbb{R}^4}\mathrm{d}^4x\,\Big(\big(D_\rho^{bc}(x)D_\sigma^{ca}(x)\big)f^{ab}(x,y)\Big)\Big|_{y=x}.
\end{align}
Let us first prove the first formulated relation. To do this, in integral \eqref{ya-a-100}, we shift the variable $x\to x+y$ and replace the integration area as follows $\mathbb{R}^4\times\mathbb{R}^4\to\mathrm{B}_{1/\sigma}\times\mathbb{R}^4$. Next, we decompose the densities in such a way that all the parts that depend on the background field have the variable $y$, and the parts containing a singular density contain only $x$. In this way, the desired factorization can be obtained for each term. Note that the Green's functions have already been redefined as necessary, see the tables in \eqref{ya-a-104} and \eqref{ya-a-105} in Section \ref{ym:sec:two-1}. In tern, the test function has a standard Taylor series expansion of the form
\begin{equation}\label{ya-a-106}
f^{ba}(y,x+y)=
f^{ba}(y,y)+
x^\nu \partial_{z^\nu}f^{ba}(y,z)\big|_{z=y}+\frac{x^\nu x^\mu}{2}
\partial_{z^\nu}\partial_{z^\mu}f^{ba}(y,z)\big|_{z=y}+\ldots
\end{equation}
Next, the analysis method is reduced to sorting through all possible combinations. Here are some supporting comments. For convenience, we will use the "value of singularity" shown in tables \eqref{ya-a-104} and \eqref{ya-a-105}.\\

\noindent 1) The Green's functions $DDG_{\mathrm{loc}}$ and $G_{\mathrm{loc}}$ have a total value of singularity of $r^{-6}$. There is only one such option, see the left columns of tables \eqref{ya-a-104} and \eqref{ya-a-105}. In this case, the operator from \eqref{ya-a-101} is converted to the form
\begin{equation}\label{ya-a-107}
\mathcal{R}_{\sigma\rho}^{ab}(x+y,y)\to f^{acd}\Big(-\delta^{cg}\partial_{x^\sigma}^{\phantom{1}}\partial_{x^\rho}^{\phantom{1}}R_0^\Lambda(x)\Big)\delta^{de}R_0^\Lambda(x)f^{egb}=c_2\delta^{ab}R_0^\Lambda(x)\partial_{x^\sigma}^{\phantom{1}}\partial_{x^\rho}^{\phantom{1}}R_0^\Lambda(x).
\end{equation}
Then we write out the interesting parts of the test function. It is clear that the degree of the variable $x$ must be even, otherwise the integral will be zeroed due to symmetry properties, see \eqref{ya-a-57}. Considering the fact that the total value of the singularity for the density should be $r^{-4}$ or stronger, we choose only two terms
\begin{equation}\label{ya-a-108}
f^{ba}(y,x+y)\to
f^{ba}(y,y)+\frac{x^\nu x^\mu}{2}
\partial_{z^\nu}\partial_{z^\mu}f^{ba}(y,z)\big|_{z=y}.
\end{equation}
Next, after substituting the obtained decompositions and using auxiliary formulas from the first part of Section \ref{ym:sec:two-vs}, we come to the answer
\begin{equation}\label{ya-a-109}
-\frac{\delta_{\sigma\rho}\Lambda^2c_2\rho_3}{16\pi^2}
\int_{\mathbb{R}^4}\mathrm{d}^4x\,f^{aa}(x,x)
+\frac{Lc_2}{96\pi^2}
\int_{\mathbb{R}^4}\mathrm{d}^4x\,\Big(\big(\delta_{\sigma\rho}A_0(x)+4\partial_{x^\rho}\partial_{x^\sigma}\big)f^{aa}(x,y)\Big)\Big|_{y=x}.
\end{equation}
\noindent 2) The Green's functions $DDG_{\mathrm{loc}}$ and $G_{\mathrm{loc}}$ have a total value of singularity of $r^{-5}$. According to the table, there can be four such options. In this case, the operator from \eqref{ya-a-101} can be converted to the following form
\begin{align}
\mathcal{R}_{\sigma\rho}^{ab}(x+y,y)\to&
\frac{c_2B_\mu^{ba}(y)}{2}\bigg[\Big(-\partial_{x^\sigma}^{\phantom{1}}\partial_{x^\rho}^{\phantom{1}}R_0^\Lambda(x)\Big)\Big(2\partial_{x^\mu}^{\phantom{1}}\theta(x)\Big)+\Big(-2\partial_{x^\sigma}^{\phantom{1}}\partial_{x^\rho}^{\phantom{1}}\partial_{x^\mu}^{\phantom{1}}\theta(x)\Big)R_0^\Lambda(x)\bigg]\\
-&\frac{c_2B_\sigma^{ba}(y)}{2}R_0^\Lambda(x)\partial_{x^\rho}^{\phantom{1}}R_0^\Lambda(x)
-\frac{c_2B_\rho^{ba}(y)}{2}R_0^\Lambda(x)\partial_{x^\sigma}^{\phantom{1}}R_0^\Lambda(x).
\end{align}
In this case, it is necessary to take only one term of the first degree from the test function
\begin{equation}\label{ya-a-111}
f^{ba}(y,x+y)\to
x^\nu \partial_{z^\nu}f^{ba}(y,z)\big|_{z=y}.
\end{equation}
After substituting and using the formulas from Section \ref{ym:sec:two-vs}, we arrive at an answer of the form
\begin{equation}\label{ya-a-112}
\frac{Lc_2}{96\pi^2}
\int_{\mathbb{R}^4}\mathrm{d}^4x\,\Big(\big(-2\delta_{\sigma\rho}^{\phantom{1}}B_\mu^{ba}(x)\partial_{x^\mu}^{\phantom{1}}+4B_\sigma^{ba}(x)\partial_{x^\rho}^{\phantom{1}}+4B_\rho^{ba}(x)\partial_{x^\sigma}^{\phantom{1}}\big)f^{ab}(x,y)\Big)\Big|_{y=x}.
\end{equation}
\noindent 3) The Green's functions $DDG_{\mathrm{loc}}$ and $G_{\mathrm{loc}}$ have a total value of singularity of $r^{-4}$. Taking into account the decompositions from the tables, a total of 13 terms are obtained. At the same time, only the main term must be saved for the test function. Omitting the routine calculations that are obtained taking into account the auxiliary integrals from Section \ref{ym:sec:two-vs}, we write only the final answer
\begin{equation}\label{ya-a-113}
\frac{Lc_2}{96\pi^2}
\int_{\mathbb{R}^4}\mathrm{d}^4x\,\Big(-\delta_{\sigma\rho}B_\mu^{bc}(x)B_\mu^{ca}(x)+2B_\rho^{bc}(x)B_\sigma^{ca}(x)+2B_\sigma^{bc}(x)B_\rho^{ca}(x)\Big)f^{ab}(x,x).
\end{equation}
Next, note that we can add an arbitrary number of derivatives of the background fields, since the convolution of symmetric $f^{ab}$ with antisymmetric $\partial B^{ba}$ leads to zero. Summing up all the answers and adding the missing number of $\partial B$, we get the stated answer \eqref{ya-a-102}. Relation \eqref{ya-a-103-1} is obtained after preserving the main order for all functions and integrating them by parts. Both formulas are proven.

\subsection{Nonlocal part}
\label{ym:sec:two-nl}

Let us study the singular nonlocal part appearing in the two-loop approximation of the quantum action. The main method is to sort through all possible combinations, taking into account the auxiliary decompositions for the Green's functions from \eqref{ya-a-74} and \eqref{ya-a-74-1}. In this case, the subsections will be named taking into account the parts of the Green's functions involved in the decomposition. For example, the two-loop part of $\mathbb{H}_0^{\mathrm{sc}}(\Gamma_3^2)$ consists of diagrams with three Green's functions. Then, using the expansion from \eqref{ya-a-74}, the contribution of $G_{\mathrm{loc}}G_{\mathrm{loc}}PS_1$ will be called all singular terms in diagrams obtained from $\mathbb{H}_0^{\mathrm{sc}}(\Gamma_3^2)$ by replacing two Green's functions with $G_{\mathrm{loc}}$ and one Green's function with $PS_1$ in all possible ways. Similar designations can be used for the rest of the parts. Here are all combinations that can lead to nonlocal singularities:
\begin{align*}
\mathbb{H}_0^{\mathrm{sc}}\big(\Gamma_3^2\big)&\to G_{\mathrm{loc}}G_{\mathrm{loc}}PS_1,\,\,\,G_{\mathrm{loc}}G_{\mathrm{loc}}PS_0,\,\,\,G_{\mathrm{loc}}\mathcal{N}PS_1
,\\
\mathbb{H}_0^{\mathrm{sc}}\big(\Omega_3^2\big)&\to G_{\mathrm{loc}}G_{\mathrm{loc}}PS_1,\,\,\,G_{\mathrm{loc}}G_{\mathrm{loc}}PS_0
,\\
\mathbb{H}_0^{\mathrm{sc}}\big(\Gamma_4\big)&\to G_{\mathrm{loc}}PS_1,\,\,\,G_{\mathrm{loc}}PS_0,\,\,\,\mathcal{N}PS_1
.
\end{align*}
The remaining combinations do not contain singularities from ultraviolet divergences. This fact is easily verified by counting "degrees" near the diagonal.

\subsubsection{Part $G_{\mathrm{loc}}G_{\mathrm{loc}}PS_i$ for $\mathbb{H}_0^{\mathrm{sc}}(\Gamma_3^2)$}
\label{ym:sec:two-nl-1}
The calculations in this section are based on the use of singular integrals from Section \ref{ym:sec:two-sin}. Let us start with the situation $i=1$, since $i=0$ is a special case. To do this, we recall that by construction, the part $PS_{1\mu\nu}^{\,\, ab}$ has two finite derivatives and therefore can be used as a test function. However, it is necessary to explain how the search should be performed. Suppose there is a two-loop diagram, then we need to cut one of the three lines and then connect the function $PS_{1\mu\nu}^{\,\, ab}$ to the cut point. In this case, the remaining two Green's functions should be replaced by $G_{\mathrm{loc}}^{ab}$. Thus, one diagram will turn into three, and $\mathbb{H}_0^{\mathrm{sc}}(\Gamma_3^2)$, consisting of 6 diagrams, will turn into 18 components. 

Mathematically, this can be rewritten as follows. First, we write out all strongly connected vertices that can be obtained by connecting two vertices $\Gamma_3$,
\begin{align*}
\mathbb{H}_2^{\mathrm{sc}}\big(\Gamma_3^2\big)=
&{\centering\adjincludegraphics[width = 1.5 cm, valign=c]{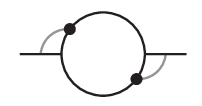}}
-{\centering\adjincludegraphics[width = 1.5 cm, valign=c]{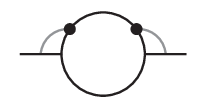}}
+2{\centering\adjincludegraphics[width = 1.5 cm, valign=c]{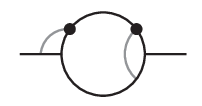}}
-2{\centering\adjincludegraphics[width = 1.5 cm, valign=c]{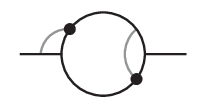}}
\\+&
{\centering\adjincludegraphics[width = 1.5 cm, valign=c]{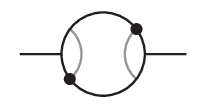}}
-{\centering\adjincludegraphics[width = 1.5 cm, valign=c]{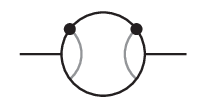}}
+2{\centering\adjincludegraphics[width = 1.5 cm, valign=c]{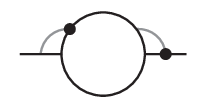}}
-2{\centering\adjincludegraphics[width = 1.5 cm, valign=c]{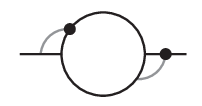}}
\\+&
{\centering\adjincludegraphics[width = 1.5 cm, valign=c]{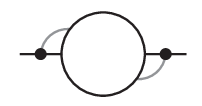}}
-{\centering\adjincludegraphics[width = 1.5 cm, valign=c]{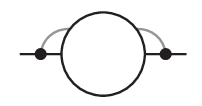}}
+2{\centering\adjincludegraphics[width = 1.5 cm, valign=c]{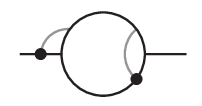}}
-2{\centering\adjincludegraphics[width = 1.5 cm, valign=c]{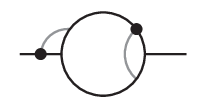}}.
\end{align*}
Next, replace $G_{1\mu\nu}^{\Lambda ab}$ with $\delta_{\mu\nu}G_{\mathrm{loc}}^{ab}$ and the resulting differential operator in the kernel notate by $r_{\sigma\rho}^{\,ab}(x,y)$, so we have
\begin{equation*}
\mathbb{H}_2^{\mathrm{sc}}\big(\Gamma_3^2\big)\Big|_{G_1^\Lambda\to G_{\mathrm{loc}}}=
\int_{\mathbb{R}^4}\mathrm{d}^4x\int_{\mathbb{R}^4}\mathrm{d}^4y\,
r_{\sigma\rho}^{\,ab}(x,y)a_{\sigma}^a(x)a_\rho^b(y).
\end{equation*}
Then a contribution in $G_{\mathrm{loc}}G_{\mathrm{loc}}PS_1$ we are interested is the singular part of the integral
\begin{equation}\label{ya-a-120}
	\int_{\mathbb{R}^4}\mathrm{d}^4x\int_{\mathbb{R}^4}\mathrm{d}^4y\,
	r_{\sigma\rho}^{\,ab}(x,y)PS_{1\sigma\rho}^{\,\,ab}(x,y).
\end{equation}
For convenience, we write out the density for each individual part after the appropriate replacement. We get 12 transitions in total: 
\begin{fleqn}
\begin{equation}\label{ya-a-114}
{\centering\adjincludegraphics[width = 1.5 cm, valign=c]{fig/ya16.eps}}\to
-4\mathcal{R}_{\sigma\rho}^{ab}(x,y)+2\mathcal{\hat{R}}_{\sigma\rho}^{ab}(x,y),\,\,\,\,\,\,
{\centering\adjincludegraphics[width = 1.5 cm, valign=c]{fig/ya22.eps}}\to
-\mathcal{R}_{\rho\sigma}^{ab}(x,y)+\frac{1}{2}\mathcal{\hat{R}}_{\rho\sigma}^{ab}(x,y),
\end{equation}
\end{fleqn}
\begin{fleqn}
	\begin{equation}\label{ya-a-115}
		{\centering\adjincludegraphics[width = 1.5 cm, valign=c]{fig/ya17.eps}}\to
	4\mathcal{R}_{\sigma\rho}^{ab}(x,y),\phantom{-+2\mathcal{\hat{R}}_{\sigma\rho}^{ab}(x,y)}\,\,\,\,\,\,
			{\centering\adjincludegraphics[width = 1.5 cm, valign=c]{fig/ya23.eps}}\to
			\delta_{\sigma\rho}^{\phantom{1}}\mathcal{R}_{\mu\mu}^{ab}(x,y),
	\end{equation}
\end{fleqn}
\begin{fleqn}
	\begin{equation}\label{ya-a-116}
		{\centering\adjincludegraphics[width = 1.5 cm, valign=c]{fig/ya18.eps}}\to
		\mathcal{R}_{\sigma\rho}^{ab}(x,y),\phantom{-4+2\mathcal{\hat{R}}_{\sigma\rho}^{ab}(x,y)}\,\,\,\,\,\,
		{\centering\adjincludegraphics[width = 1.5 cm, valign=c]{fig/ya24.eps}}\to
		\mathcal{\hat{R}}_{\rho\sigma}^{ab}(x,y),
	\end{equation}
\end{fleqn}
\begin{fleqn}
	\begin{equation}\label{ya-a-117}
		{\centering\adjincludegraphics[width = 1.5 cm, valign=c]{fig/ya19.eps}}\to
		-\mathcal{R}_{\sigma\rho}^{ab}(x,y)+\frac{1}{2}\mathcal{\hat{R}}_{\sigma\rho}^{ab}(x,y),
		\,\,\,\,\,\,\,\,
		{\centering\adjincludegraphics[width = 1.5 cm, valign=c]{fig/ya25.eps}}\to
		\delta_{\sigma\rho}^{\phantom{1}}\mathcal{\hat{R}}_{\mu\mu}^{ab}(x,y),
	\end{equation}
\end{fleqn}
\begin{fleqn}
	\begin{equation}\label{ya-a-118}
		{\centering\adjincludegraphics[width = 1.5 cm, valign=c]{fig/ya20.eps}}\to
		-\frac{1}{2}\mathcal{\hat{R}}_{\sigma\rho}^{ab}(x,y),\phantom{-\mathcal{R}_{\sigma\rho}^{ab}(x,y)}\,\,\,\,\,\,\,\,\,\,\,
		{\centering\adjincludegraphics[width = 1.5 cm, valign=c]{fig/ya26.eps}}\to
		-\frac{1}{2}\delta_{\sigma\rho}^{\phantom{1}}\mathcal{\hat{R}}_{\mu\mu}^{ab}(x,y),
	\end{equation}
\end{fleqn}
\begin{fleqn}
	\begin{equation}\label{ya-a-119}
		{\centering\adjincludegraphics[width = 1.5 cm, valign=c]{fig/ya21.eps}}\to
		-\frac{1}{2}\mathcal{\hat{R}}_{\sigma\rho}^{ab}(x,y),\phantom{-\mathcal{R}_{\sigma\rho}^{ab}(x,y)}\,\,\,\,\,\,\,\,\,\,\,
				{\centering\adjincludegraphics[width = 1.5 cm, valign=c]{fig/ya27.eps}}\to
				-\frac{1}{2}\mathcal{\hat{R}}_{\rho\sigma}^{ab}(x,y).
	\end{equation}
\end{fleqn}
Summarizing the expressions obtained, we arrive at the relation for the operator kernel of interest
\begin{align*}
r_{\sigma\rho}^{\,ab}(x,y)=&-4\mathcal{R}_{\sigma\rho}^{ab}(x,y)-\mathcal{R}_{\rho\sigma}^{ab}(x,y)-
\delta_{\sigma\rho}^{\phantom{1}}\mathcal{R}_{\mu\mu}^{ab}(x,y)\\
&+\mathcal{\hat{R}}_{\sigma\rho}^{ab}(x,y)+\frac{5}{2}\mathcal{\hat{R}}_{\rho\sigma}^{ab}(x,y)-2
\delta_{\sigma\rho}^{\phantom{1}}\mathcal{\hat{R}}_{\mu\mu}^{ab}(x,y).
\end{align*}
Therefore, taking into account the definitions of \eqref{ya-a-100} and \eqref{ya-a-100-1}, the integral \eqref{ya-a-120} is rewritten as the sum of six functionals
\begin{equation}\label{ya-a-135}
-4\mathrm{J}_{\sigma\rho}^{\phantom{1}}[PS_{1\sigma\rho}]-
\mathrm{J}_{\sigma\rho}^{\phantom{1}}[PS_{1\rho\sigma}]-
\mathrm{J}_{\sigma\sigma}^{\phantom{1}}[PS_{1\rho\rho}]+
\mathrm{\hat{J}}_{\sigma\rho}^{\phantom{1}}[PS_{1\sigma\rho}]+\frac{5}{2}
\mathrm{\hat{J}}_{\sigma\rho}^{\phantom{1}}[PS_{1\rho\sigma}]-2
\mathrm{\hat{J}}_{\sigma\sigma}^{\phantom{1}}[PS_{1\rho\rho}].
\end{equation}
Additionally, we note that similar expressions can be written out for each individual diagram. The representations for the corresponding integrals have the form:
\begin{align*}
	{\centering\adjincludegraphics[width = 1 cm, valign=c]{fig/ya8.eps}}&\longrightarrow
-\mathrm{J}_{\sigma\rho}^{\phantom{1}}[PS_{1\rho\sigma}]-
\mathrm{\hat{J}}_{\sigma\rho}^{\phantom{1}}[PS_{1\sigma\rho}]+\frac{1}{2}
\mathrm{\hat{J}}_{\sigma\rho}^{\phantom{1}}[PS_{1\rho\sigma}],\\
	-{\centering\adjincludegraphics[width = 1 cm, valign=c]{fig/ya7.eps}}&\longrightarrow
-4\mathrm{J}_{\sigma\rho}^{\phantom{1}}[PS_{1\sigma\rho}]-
\mathrm{J}_{\sigma\sigma}^{\phantom{1}}[PS_{1\rho\rho}]-
\mathrm{\hat{J}}_{\sigma\sigma}^{\phantom{1}}[PS_{1\rho\rho}],\\
	{\centering\adjincludegraphics[width = 1 cm, valign=c]{fig/ya9.eps}}&\longrightarrow
-4\mathrm{J}_{\sigma\rho}^{\phantom{1}}[PS_{1\sigma\rho}]+2
\mathrm{\hat{J}}_{\sigma\rho}^{\phantom{1}}[PS_{1\sigma\rho}]-
\mathrm{\hat{J}}_{\sigma\sigma}^{\phantom{1}}[PS_{1\rho\rho}],\\
	-2{\centering\adjincludegraphics[width = 1 cm, valign=c]{fig/ya10.eps}}&\longrightarrow
2\mathrm{J}_{\sigma\rho}^{\phantom{1}}[PS_{1\sigma\rho}]+
\mathrm{\hat{J}}_{\sigma\rho}^{\phantom{1}}[PS_{1\rho\sigma}],\\
	{\centering\adjincludegraphics[width = 1 cm, valign=c]{fig/ya11.eps}}&\longrightarrow
2\mathrm{J}_{\sigma\rho}^{\phantom{1}}[PS_{1\sigma\rho}]+
\mathrm{\hat{J}}_{\sigma\rho}^{\phantom{1}}[PS_{1\rho\sigma}].
\end{align*}
Summing up the last relations, we can check the consistency with the representation of \eqref{ya-a-135}. Let us use the obtained values \eqref{ya-a-102} and \eqref{ya-a-103-1} for singular parts, then
\begin{align}\label{ya-a-121}
\mathbb{H}_0^{\mathrm{sc}}\big(\Gamma_3^2\big)\Big|_{G_{\mathrm{loc}}G_{\mathrm{loc}}PS_1}\stackrel{\mathrm{s.p.}}{=}&\,\frac{9\Lambda^2c_2\rho_3}{16\pi^2}
	\int_{\mathbb{R}^4}\mathrm{d}^4x\,PS_{1\rho\rho}^{\,\,aa}(x,x)\\\nonumber
	+&\,\frac{Lc_2}{2\pi^2}
	\int_{\mathbb{R}^4}\mathrm{d}^4x\,
	\bigg(\bigg[-\frac{5}{48}M_{1\rho\sigma}^{\,\,ba}(x)+\frac{\delta_{\sigma\rho}}{2}M_{0}^{ba}(x)\\\nonumber
	&\,\,\,\,\,\,\,\,\,\,\,\,\,\,\,\,\,\,\,\,\,\,\,\,\,\,\,\,\,\,
	\,\,\,\,\,\,\,\,\,\,\,+\frac{11}{24}D_\rho^{bc}(x)D_\sigma^{ca}(x)
	-\frac{5}{8}F_{\rho\sigma}^{ba}(x)\bigg]PS_{1\sigma\rho}^{\,\,ab}(x,y)\bigg)\bigg|_{y=x}.
\end{align}
By construction, $PS_0$ has the same smoothness properties as the function $PS_1$. Therefore, we can use the resulting formula \eqref{ya-a-121} by substituting the function $\delta_{\sigma\rho}^{\phantom{1}}PS_{0}^{ab}(x,y)$ instead of $PS_{1\sigma\rho}^{\,\,ab}(x,y)$. Thus, we have
\begin{align}\label{ya-a-129}
	\mathbb{H}_0^{\mathrm{sc}}\big(\Gamma_3^2\big)\Big|_{G_{\mathrm{loc}}G_{\mathrm{loc}}PS_0}\stackrel{\mathrm{s.p.}}{=}&\,\frac{9\Lambda^2c_2\rho_3}{4\pi^2}
	\int_{\mathbb{R}^4}\mathrm{d}^4x\,PS_{0}^{aa}(x,x)\\\nonumber
	+&\,\frac{27Lc_2}{48\pi^2}
	\int_{\mathbb{R}^4}\mathrm{d}^4x\,
	\Big(M_{0}^{ba}(x)PS_{0}^{ab}(x,y)\Big)\Big|_{y=x}.
\end{align}

\subsubsection{Part $G_{\mathrm{loc}}\mathcal{N}PS_1$ for $\mathbb{H}_0^{\mathrm{sc}}(\Gamma_3^2)$}
\label{ym:sec:two-nl-2}
In this case, the counting process is similar to the previous one. Let us explain the differences and complete the calculations. Note that now, after the stage of cutting the line in a two-loop diagram from $\mathbb{H}_0^{\mathrm{sc}}\big(\Gamma_3^2\big)$, one of the two remaining Green's functions $G_1^\Lambda$ is replaced by $G_{\mathrm{loc}}$, and the second one is for $\mathcal{N}$, and vice versa. So the result is not 18 diagrams, but 36.

Next, note that in the main order near the diagonal, when $x\sim y$, the function $\mathcal{N}_{\mu\nu}^{ab}(x,y)$ is proportional to $2F_{\mu\nu}^{ab}(y)\theta(x-y)$, therefore, for a singular contribution to occur, both derivatives must be inside the loop. Thus, of the sum of $\mathbb{H}_2^{\mathrm{sc}}\big(\Gamma_3^2\big)$, only the following terms can eventually give a non-zero contribution
\begin{equation*}
{\centering\adjincludegraphics[width = 1.5 cm, valign=c]{fig/ya16.eps}}
	-{\centering\adjincludegraphics[width = 1.5 cm, valign=c]{fig/ya17.eps}}
	+2{\centering\adjincludegraphics[width = 1.5 cm, valign=c]{fig/ya18.eps}}
	-2{\centering\adjincludegraphics[width = 1.5 cm, valign=c]{fig/ya19.eps}}
	+{\centering\adjincludegraphics[width = 1.5 cm, valign=c]{fig/ya22.eps}}
	-{\centering\adjincludegraphics[width = 1.5 cm, valign=c]{fig/ya23.eps}}.
\end{equation*}
Note that the derivative can only be shifted inside the loop, since a derivative acting on the tail leads to a finite contribution. Therefore, we proceed to
\begin{equation*}
	-2{\centering\adjincludegraphics[width = 1.5 cm, valign=c]{fig/ya17.eps}}
	+4{\centering\adjincludegraphics[width = 1.5 cm, valign=c]{fig/ya18.eps}}
	+{\centering\adjincludegraphics[width = 1.5 cm, valign=c]{fig/ya22.eps}}
	-{\centering\adjincludegraphics[width = 1.5 cm, valign=c]{fig/ya23.eps}}.
\end{equation*}
Now, replacing as follows
\begin{equation*}
G_{1\mu\nu}^{\,\,ab}(x,y)\to\delta_{\mu\nu}^{\phantom{1}}\delta^{ab}R_0^\Lambda(x-y),\,\,\,
\mathcal{N}_{\mu\nu}^{ab}(x,y)\to2F_{\mu\nu}^{ab}(y)\theta(x-y),
\end{equation*}
note that the first diagram leads to zero due to $F_{\mu\mu}^{ab}(y)=0$. The remaining parts, after shifting the variable $x\to x+y$, add up to a density of the form
\begin{equation*}
-2c_2F_{\mu\rho}^{ba}(y)\Big(R_0^\Lambda(x)\partial_{x^\sigma}\partial_{x^\mu}\theta(x)\Big)
-4c_2F_{\rho\mu}^{ba}(y)\Big(R_0^\Lambda(x)\partial_{x^\sigma}\partial_{x^\mu}\theta(x)\Big)
+c_2F_{\sigma\rho}^{ba}(y)\Big(R_0^\Lambda(x)\partial_{x^\mu}\partial_{x^\mu}\theta(x)\Big).
\end{equation*}
Next, multiplying by $PS_{1\sigma\rho}^{\,\,ab}(y,y)$ and using the auxiliary relations from Section \ref{ym:sec:two-vs}, we get the following answer
\begin{equation}\label{ya-a-122}
	\mathbb{H}_0^{\mathrm{sc}}\big(\Gamma_3^2\big)\Big|_{G_{\mathrm{loc}}\mathcal{N}PS_1}\stackrel{\mathrm{s.p.}}{=}\frac{3Lc_2}{16\pi^2}
	\int_{\mathbb{R}^4}\mathrm{d}^4y\,
	\Big(F_{\rho\sigma}^{ba}(y)PS_{1\sigma\rho}^{\,\,ab}(y,y)\Big).
\end{equation}
To complete the picture, we note that the expressions for each individual diagram are written as follows:
\begin{equation*}
3{\centering\adjincludegraphics[width = 1 cm, valign=c]{fig/ya8.eps}}\Big|_{G_{\mathrm{loc}}\mathcal{N}PS_1}
\stackrel{\mathrm{s.p.}}{=}
-\frac{3}{2}{\centering\adjincludegraphics[width = 1 cm, valign=c]{fig/ya7.eps}}\Big|_{G_{\mathrm{loc}}\mathcal{N}PS_1}
\stackrel{\mathrm{s.p.}}{=}\mathbb{H}_0^{\mathrm{sc}}\big(\Gamma_3^2\big)\Big|_{G_{\mathrm{loc}}\mathcal{N}PS_1},
\end{equation*}
\begin{equation*}
{\centering\adjincludegraphics[width = 1 cm, valign=c]{fig/ya9.eps}}\Big|_{G_{\mathrm{loc}}\mathcal{N}PS_1}
\stackrel{\mathrm{s.p.}}{=}
{\centering\adjincludegraphics[width = 1 cm, valign=c]{fig/ya10.eps}}\Big|_{G_{\mathrm{loc}}\mathcal{N}PS_1}
\stackrel{\mathrm{s.p.}}{=}
{\centering\adjincludegraphics[width = 1 cm, valign=c]{fig/ya11.eps}}\Big|_{G_{\mathrm{loc}}\mathcal{N}PS_1}
\stackrel{\mathrm{s.p.}}{=}0.
\end{equation*}

\subsubsection{Parts for $\mathbb{H}_0^{\mathrm{sc}}(\Omega_3^2)$}
\label{ym:sec:two-nl-3}
Let us consider all three variants of the emerging nonlocal singularities in turn. First, let us recall the explicit form of the two-loop diagram
\begin{equation*}
	\mathbb{H}_0^{\mathrm{sc}}\big(\Omega_3^2\big)=
	-{\centering\adjincludegraphics[width = 1 cm, valign=c]{fig/ya12.eps}}.
\end{equation*}
Its special feature is that it contains only one Green's function $G_1^\Lambda$. The other two are equal to $G_0^\Lambda$. This fact leads to some changes in the computing process. \\

\noindent1) Consider the case of $G_{\mathrm{loc}}G_{\mathrm{loc}}PS_1$. Since the function $PS_1$ is present only in the middle line, it is it that should be cut. The resulting diagram is similar to the one already studied and is reproduced as follows
\begin{equation}\label{ya-a-153}
-\frac{1}{4}{\centering\adjincludegraphics[width = 1.5 cm, valign=c]{fig/ya16.eps}}\to
\mathcal{R}_{\sigma\rho}^{ab}(x,y)-\frac{1}{2}\mathcal{\hat{R}}_{\sigma\rho}^{ab}(x,y).
\end{equation}
Therefore, using \eqref{ya-a-102} and \eqref{ya-a-103-1}, we get the answer
\begin{align}\label{ya-a-130}
	\mathbb{H}_0^{\mathrm{sc}}\big(\Omega_3^2\big)\Big|_{G_{\mathrm{loc}}G_{\mathrm{loc}}PS_1}\stackrel{\mathrm{s.p.}}{=}&-\frac{\Lambda^2c_2\rho_3}{16\pi^2}
	\int_{\mathbb{R}^4}\mathrm{d}^4x\,PS_{1\rho\rho}^{\,\,aa}(x,x)\\\nonumber
	&+\frac{Lc_2}{2\pi^2}
	\int_{\mathbb{R}^4}\mathrm{d}^4x\,
	\bigg(\bigg[\frac{1}{48}M_{1\rho\sigma}^{\,\,ba}(x)-\frac{1}{24}D_\rho^{bc}(x)D_\sigma^{ca}(x)
\bigg]PS_{1\sigma\rho}^{\,\,ab}(x,y)\bigg)\bigg|_{y=x}.
\end{align}
2) Consider the case of $G_{\mathrm{loc}}G_{\mathrm{loc}}PS_0$. It is divided into two parts, since $PS_0$ is contained in both the midline and Green's scalar functions. In the first case, the answer is obtained by replacing the function $PS_{1\sigma\rho}$ in \eqref{ya-a-130} with $\delta_{\sigma\rho}PS_0$, that is
\begin{equation*}
-\frac{\Lambda^2c_2\rho_3}{4\pi^2}
\int_{\mathbb{R}^4}\mathrm{d}^4x\,PS_{0}^{aa}(x,x)
+\frac{Lc_2}{16\pi^2}
\int_{\mathbb{R}^4}\mathrm{d}^4x\,
\Big(M_{0}^{ba}(x)PS_{0}^{ab}(x,y)\Big)\bigg|_{y=x}.
\end{equation*}
In the second case, the side line is cut, and the resulting diagram is reduced to the one already studied, for which the transition is valid
\begin{equation*}
\frac{1}{4}{\centering\adjincludegraphics[width = 1.5 cm, valign=c]{fig/ya25.eps}}\bigg|_{\sigma=\rho}\to
\mathcal{\hat{R}}_{\mu\mu}^{ab}(x,y).
\end{equation*}
Then, integrating with function $PS_0^{ab}(x,y)$, we get the result in the form
\begin{equation*}
-\frac{Lc_2}{8\pi^2}
\int_{\mathbb{R}^4}\mathrm{d}^4x\,
\Big(M_{0}^{ba}(x)PS_{0}^{ab}(x,y)\Big)\bigg|_{y=x}.
\end{equation*}
Considering that there are two side lines, after summing the results, we get
\begin{equation}\label{ya-a-131}
	\mathbb{H}_0^{\mathrm{sc}}\big(\Omega_3^2\big)\Big|_{G_{\mathrm{loc}}G_{\mathrm{loc}}PS_0}\stackrel{\mathrm{s.p.}}{=}-\frac{\Lambda^2c_2\rho_3}{4\pi^2}
	\int_{\mathbb{R}^4}\mathrm{d}^4x\,PS_{0}^{aa}(x,x)
	-\frac{3Lc_2}{16\pi^2}
	\int_{\mathbb{R}^4}\mathrm{d}^4x\,
	\Big(M_{0}^{ba}(x)PS_{0}^{ab}(x,y)\Big)\bigg|_{y=x}.
\end{equation}

\subsubsection{Parts for $\mathbb{H}_0^{\mathrm{sc}}(\Gamma_4)$}
\label{ym:sec:two-nl-4}
In this case, the contribution from the quartic vertex consists of three diagrams. Let us recall their explicit form
\begin{equation*}
{\centering\adjincludegraphics[width = 1 cm, valign=c]{fig/ya13.eps}}-
{\centering\adjincludegraphics[width = 1 cm, valign=c]{fig/ya14.eps}}+
{\centering\adjincludegraphics[width = 2 cm, valign=c]{fig/ya15.eps}}.
\end{equation*}
They consist of two Green's functions on the diagonal, so first it is convenient to figure out exactly which singularities occur in the function $G_{\mathrm{loc}}^{ab}(x,x)$. Using the representation form \eqref{ya-a-74-1}, we see that only the main order $\delta^{ab}\Lambda^2R_0^1(0)$ is singular on the diagonal, since the remaining terms, which previously led to singularities during integration, in the case of $y=x$ are reduced due to the properties
\begin{equation}\label{ya-a-156}
\partial_{x^\mu}\theta(x-y)\big|_{y=x}=0,\,\,\,
\delta_{\mu\nu}\theta(0)+4\partial_{x^\mu}\partial_{x^\nu}\tau(x-y)\big|_{y=x}\stackrel{\mathrm{s.p.}}{=}0.
\end{equation}
Therefore, when calculating nonlocal components, we can use the substitution
\begin{equation*}
G_{\mathrm{loc}}^{ab}(x,x)\to\delta^{ab}\Lambda^2R_0^1(0).
\end{equation*}
1) Consider the cases with $G_{\mathrm{loc}}PS_i$. Let us immediately note that due to the symmetry of $\delta^{ab}$, the third diagram in this case gives zero contribution. Further, summing up the spatial indexes, we note that 
\begin{equation*}
\frac{1}{4}
{\centering\adjincludegraphics[width = 1 cm, valign=c]{fig/ya14.eps}}\bigg|_{G_{\mathrm{loc}}PS_1}=
{\centering\adjincludegraphics[width = 1 cm, valign=c]{fig/ya13.eps}}\bigg|_{G_{\mathrm{loc}}PS_1}=
-2c_2\Lambda^2R_0^1(0)
\int_{\mathbb{R}^4}\mathrm{d}^4x\,PS_{1\rho\rho}^{\,\,aa}(x,x).
\end{equation*}
Thus, the relation is valid
\begin{equation}\label{ya-a-132}
	\mathbb{H}_0^{\mathrm{sc}}\big(\Gamma_4\big)\Big|_{G_{\mathrm{loc}}PS_1}\stackrel{\mathrm{s.p.}}{=}6c_2\Lambda^2R_0^1(0)
	\int_{\mathbb{R}^4}\mathrm{d}^4x\,PS_{1\rho\rho}^{\,\,aa}(x,x).
\end{equation}
Next, replacing $PS_{1\mu\nu}$ with $\delta_{\mu\nu}PS_0$, we get a similar equality
\begin{equation}\label{ya-a-133}
\mathbb{H}_0^{\mathrm{sc}}\big(\Gamma_4\big)\Big|_{G_{\mathrm{loc}}PS_0}\stackrel{\mathrm{s.p.}}{=}24c_2\Lambda^2R_0^1(0)
\int_{\mathbb{R}^4}\mathrm{d}^4x\,PS_{0}^{aa}(x,x).
\end{equation}
2) Consider the case of $\mathcal{N}PS_1$. It can be immediately simplified by replacing $\mathcal{N}_{\mu\nu}^{ab}(x,x)$ with the main order $2F_{\mu\nu}^{ab}(x)\theta(0)$, and also by throwing out the second diagram, since the trace of the field strength tensor is zero. Next, recall that $\theta(0)\stackrel{\mathrm{s.p.}}{=}L/(8\pi^2)$, see Section \ref{ym:sec:two-vs}, then we get
\begin{align*}
{\centering\adjincludegraphics[width = 1 cm, valign=c]{fig/ya13.eps}}&\to
\frac{Lc_2}{4\pi^2}
	\int_{\mathbb{R}^4}\mathrm{d}^4x\,
F_{\rho\sigma}^{ba}(x)PS_{1\sigma\rho}^{\,\,ab}(x,x),
\\
{\centering\adjincludegraphics[width = 2 cm, valign=c]{fig/ya15.eps}}&\to
\frac{Lc_2}{2\pi^2}
\int_{\mathbb{R}^4}\mathrm{d}^4x\,
F_{\rho\sigma}^{ba}(x)PS_{1\sigma\rho}^{\,\,ab}(x,x).
\end{align*}
Thus, after summing up, we get an answer of the form
\begin{equation}\label{ya-a-134}
	\mathbb{H}_0^{\mathrm{sc}}\big(\Gamma_4\big)\Big|_{\mathcal{N}PS_1}\stackrel{\mathrm{s.p.}}{=}\frac{3Lc_2}{4\pi^2}
	\int_{\mathbb{R}^4}\mathrm{d}^4x\,
	F_{\rho\sigma}^{ba}(x)PS_{1\sigma\rho}^{\,\,ab}(x,x).
\end{equation}

\subsection{Local part}
\label{ym:sec:two-l}
Once again, we use the general idea of counting singular components, proposed in Section \ref{ym:sec:two-nl}. To do this, first, using the decompositions from \eqref{ya-a-74} and \eqref{ya-a-127}, we write out a list of combinations that can lead to nonzero singular components:
\begin{align*}
	\mathbb{H}_0^{\mathrm{sc}}\big(\Gamma_3^2\big)&\to G_{\mathrm{loc}}\mathcal{N}\mathcal{N},\,\,\,G_{\mathrm{loc}}G_{\mathrm{loc}}\mathcal{N},\,\,\,G_{\mathrm{loc}}G_{\mathrm{loc}}\mathcal{L},\,\,\,G_{\mathrm{loc}}G_{\mathrm{loc}}G_{\mathrm{loc}}
	,\\
	\mathbb{H}_0^{\mathrm{sc}}\big(\Omega_3^2\big)&\to
	G_{\mathrm{loc}}G_{\mathrm{loc}}\mathcal{N},\,\,\,
	 G_{\mathrm{loc}}G_{\mathrm{loc}}\mathcal{L},\,\,\,G_{\mathrm{loc}}G_{\mathrm{loc}}G_{\mathrm{loc}}
	,\\
	\mathbb{H}_0^{\mathrm{sc}}\big(\Gamma_4\big)&\to \mathcal{N}\mathcal{N},\,\,\,G_{\mathrm{loc}}\mathcal{L},\,\,\,G_{\mathrm{loc}}G_{\mathrm{loc}}
	.
\end{align*}
The remaining combinations do not lead to singularities of "ultraviolet" nature.

\subsubsection{Part $G_{\mathrm{loc}}\mathcal{N}\mathcal{N}$ for $\mathbb{H}_0^{\mathrm{sc}}(\Gamma_3^2)$}
\label{ym:sec:two-l-1}
Let us immediately note that the main order of the Green's function $G_{0}^\Lambda$ has the order of the singularity proportional $r^{-2}$, which is due to the behavior near the diagonal. Further, the main order of the function $\mathcal{N}$ has a logarithmic form, so it can be understood as $\ln|r|$. Thus, the maximum order of singularity for the combination $G_{\mathrm{loc}}\mathcal{N}\mathcal{N}$ does not exceed $\ln^2|r|/r^2$. Therefore, to search for singular parts in all diagrams from $\mathbb{H}_0^{\mathrm{sc}}\big(\Gamma_3^2\big)$, see \eqref{ya-a-19}, it is necessary to replace covariant derivatives with ordinary ones. Then it all comes down to a simple search of possible combinations. Let us explain the example of calculations in the first diagram from \eqref{ya-a-19}, for the rest we will write out only final answers in the form of a summary table.

As before, first we need to shift the variable $x\to x+y$, proceed to integration over the variables $(x,y)$ over the domain $\mathrm{B}_{1/\sigma}\times\mathbb{R}^4$, and then replace the densities into the main parts of the decompositions, that is 
\begin{equation}\label{ya-a-124}
\delta_{\mu\nu}^{\phantom{1}}G_{\mathrm{loc}}^{ab}(x+y,y)\to\delta_{\mu\nu}^{\phantom{1}}\delta^{ab}R_0^\Lambda(x),\,\,\,
\mathcal{N}_{\mu\nu}^{ab}(x+y,y)\to 2F_{\mu\nu}^{ab}(y)\theta(x).
\end{equation}
Note that only the main decomposition order is important for $\mathcal{N}_{\mu\nu}^{ab}(x+y,y)$, that is, only the first line from \eqref{ya-a-125}. The remaining parts will lead to contributions without singularities. Next, note that there are three possible non-zero combinations for each diagram, depending on what exactly each line is replaced with. Let the first case correspond to the situation when the bottom line is replaced by $G_{\mathrm{loc}}$, the second is the case when the middle line is replaced by $G_{\mathrm{loc}}$, and finally the third corresponds to the upper line. In this case, taking into account the auxiliary asymptotic expansions from Section \ref{ym:sec:two-vs}, we can check the following relation
\begin{equation*}
{\centering\adjincludegraphics[width = 1 cm, valign=c]{fig/ya8.eps}}\Big|_{G_{\mathrm{loc}}\mathcal{N}\mathcal{N}}
\stackrel{\mathrm{s.p.}}{=}
-\frac{c_2^2}{2}W_{-1}\bigg(\Big[-\hat{\mathrm{I}}_2/2\Big]+\Big[-\hat{\mathrm{I}}_2/2\Big]+\Big[\hat{\mathrm{I}}_1-\hat{\mathrm{I}}_2/2\Big]\bigg),
\end{equation*}
where each part in square brackets corresponds to a different situation. The integrals $\hat{\mathrm{I}}_i$ are defined in Section \ref{ym:sec:two-vs}. It turns out that for all diagrams, the answers can be written in the same way, replacing only the integral in square brackets with a suitable combination of $a\hat{\mathrm{I}}_1+b\hat{\mathrm{I}}_2$. Here are the answers:
\begin{align*}
-{\centering\adjincludegraphics[width = 1 cm, valign=c]{fig/ya7.eps}}\Big|_{G_{\mathrm{loc}}\mathcal{N}\mathcal{N}}&
	\stackrel{\mathrm{s.p.}}{=}
	-\frac{c_2^2}{2}W_{-1}\bigg(\Big[0\Big]+\Big[-4\hat{\mathrm{I}}_1\Big]+\Big[0\Big]\bigg),\\
{\centering\adjincludegraphics[width = 1 cm, valign=c]{fig/ya9.eps}}\Big|_{G_{\mathrm{loc}}\mathcal{N}\mathcal{N}}&
\stackrel{\mathrm{s.p.}}{=}
-\frac{c_2^2}{2}W_{-1}\bigg(\Big[0\Big]+\Big[-4\hat{\mathrm{I}}_1+2\hat{\mathrm{I}}_2\Big]+\Big[0\Big]\bigg),\\
-2{\centering\adjincludegraphics[width = 1 cm, valign=c]{fig/ya10.eps}}\Big|_{G_{\mathrm{loc}}\mathcal{N}\mathcal{N}}&
\stackrel{\mathrm{s.p.}}{=}
-\frac{c_2^2}{2}W_{-1}\bigg(\Big[-\hat{\mathrm{I}}_2\Big]+\Big[\hat{\mathrm{I}}_2\Big]+\Big[2\hat{\mathrm{I}}_1-\hat{\mathrm{I}}_2\Big]\bigg),\\
{\centering\adjincludegraphics[width = 1 cm, valign=c]{fig/ya11.eps}}\Big|_{G_{\mathrm{loc}}\mathcal{N}\mathcal{N}}&
\stackrel{\mathrm{s.p.}}{=}
-\frac{c_2^2}{2}W_{-1}\bigg(\Big[-\hat{\mathrm{I}}_2\Big]+\Big[\hat{\mathrm{I}}_1\Big]+\Big[\hat{\mathrm{I}}_1\Big]\bigg).
\end{align*}
Summing up the expressions, we get the following result for the combination of two-loop diagrams
\begin{equation}\label{ya-a-126}
\mathbb{H}_0^{\mathrm{sc}}\big(\Gamma_3^2\big)\Big|_{G_{\mathrm{loc}}\mathcal{N}\mathcal{N}}\stackrel{\mathrm{s.p.}}{=}
-\frac{c_2^2}{2}W_{-1}\bigg(-3\hat{\mathrm{I}}_1-\frac{3}{2}\hat{\mathrm{I}}_2\bigg)
\stackrel{\mathrm{s.p.}}{=}\frac{3c_2^2W_{-1}}{16(4\pi^2)^2}\Big(2L^2+L(1+4\rho_1+4\rho_2)\Big).
\end{equation}

\subsubsection{Part $G_{\mathrm{loc}}G_{\mathrm{loc}}\mathcal{N}$ for $\mathbb{H}_0^{\mathrm{sc}}(\Gamma_3^2)$}
\label{ym:sec:two-l-2}

In this case, two functions in the diagram remain $G_{\mathrm{loc}}$, so the general idea of counting is similar to the one from Section \ref{ym:sec:two-nl-1}. In this case, the third line is replaced by the function $\mathcal{N}$, and it is this fact that makes significant differences in the subsequent calculation, since $\mathcal{N}$ does not have two non-singular derivatives. Moreover, on the diagonal it behaves like $\theta(0)\sim L$. This means that it is not possible to use the results for singular operators from Section \ref{ym:sec:two-sin}. 

However, we can write out general answers. To do this, use the ready-made representations after formula \eqref{ya-a-135} and replace $PS_{1\sigma\rho}$ with $\mathcal{N}_{\sigma\rho}$ in them. Using the properties $\mathcal{N}_{\sigma\rho}=-\mathcal{N}_{\rho\sigma}$ and $\mathcal{N}_{\sigma\sigma}=0$ and the relations for auxiliary integrals from Section \ref{ym:sec:two-nl-1}, the intermediate results for individual diagrams can be represented as follows
\begin{align*}
{\centering\adjincludegraphics[width = 1 cm, valign=c]{fig/ya8.eps}}\Big|_{G_{\mathrm{loc}}G_{\mathrm{loc}}\mathcal{N}}&
\stackrel{\mathrm{s.p.}}{=}
\mathrm{J}_{\sigma\rho}^{\phantom{1}}[\mathcal{N}_{1\sigma\rho}]-\frac{3}{2}
\mathrm{\hat{J}}_{\sigma\rho}^{\phantom{1}}[\mathcal{N}_{1\sigma\rho}],\\
	-{\centering\adjincludegraphics[width = 1 cm, valign=c]{fig/ya7.eps}}\Big|_{G_{\mathrm{loc}}G_{\mathrm{loc}}\mathcal{N}}&
		\stackrel{\mathrm{s.p.}}{=}
-4\mathrm{J}_{\sigma\rho}^{\phantom{1}}[\mathcal{N}_{1\sigma\rho}],\\
		{\centering\adjincludegraphics[width = 1 cm, valign=c]{fig/ya9.eps}}\Big|_{G_{\mathrm{loc}}G_{\mathrm{loc}}\mathcal{N}}&
		\stackrel{\mathrm{s.p.}}{=}
-4\mathrm{J}_{\sigma\rho}^{\phantom{1}}[\mathcal{N}_{1\sigma\rho}]+2
\mathrm{\hat{J}}_{\sigma\rho}^{\phantom{1}}[\mathcal{N}_{1\sigma\rho}],\\
		-2{\centering\adjincludegraphics[width = 1 cm, valign=c]{fig/ya10.eps}}\Big|_{G_{\mathrm{loc}}G_{\mathrm{loc}}\mathcal{N}}&
		\stackrel{\mathrm{s.p.}}{=}
2\mathrm{J}_{\sigma\rho}^{\phantom{1}}[\mathcal{N}_{1\sigma\rho}]-
\mathrm{\hat{J}}_{\sigma\rho}^{\phantom{1}}[\mathcal{N}_{1\sigma\rho}],\\
		{\centering\adjincludegraphics[width = 1 cm, valign=c]{fig/ya11.eps}}\Big|_{G_{\mathrm{loc}}G_{\mathrm{loc}}\mathcal{N}}&
		\stackrel{\mathrm{s.p.}}{=}
2\mathrm{J}_{\sigma\rho}^{\phantom{1}}[\mathcal{N}_{1\sigma\rho}]-
\mathrm{\hat{J}}_{\sigma\rho}^{\phantom{1}}[\mathcal{N}_{1\sigma\rho}].
	\end{align*}
In turn, for their sum \eqref{ya-a-135}, we can write out the answer in the form
\begin{equation}\label{ya-a-138}
\mathbb{H}_0^{\mathrm{sc}}\big(\Gamma_3^2\big)\Big|_{G_{\mathrm{loc}}G_{\mathrm{loc}}\mathcal{N}}\stackrel{\mathrm{s.p.}}{=}
-3\mathrm{J}_{\sigma\rho}^{\phantom{1}}[\mathcal{N}_{1\sigma\rho}]-\frac{3}{2}
\mathrm{\hat{J}}_{\sigma\rho}^{\phantom{1}}[\mathcal{N}_{1\sigma\rho}].
\end{equation}
The last integrals need to be analyzed anew. Let us explain the first case, and for the second we will present only the final answer. To do this, we use the definition of integral \eqref{ya-a-100} and kernel \eqref{ya-a-101}. Next, we substitute the decompositions for the Green's functions $G_{\mathrm{loc}}$ from Tables \eqref{ya-a-104} and \eqref{ya-a-105}, as well as the decomposition for the function $\mathcal{N}$ from \eqref{ya-a-125}. Then we remove all terms that do not lead to singular contributions for reasons of zero trace or because of the convolution of a symmetric object with an antisymmetric one. As a result, after shifting $x\to x+y$ and transition to integration by $\mathrm{B}_{1/\sigma}\times\mathbb{R}^4$, the remaining parts are obtained by substitutions:
\begin{align*}
	G_{\mathrm{loc}}^{ab}(x,y)\big|_{x\to x+y}&\to\delta^{ab}R_0^\Lambda(x),
\\
\mathcal{N}_{\sigma\rho}^{ab}(x,y)\big|_{x\to x+y}&\to2F_{\sigma\rho}^{ab}(y)\theta(x),
\\
D_{\sigma}^{ca}(x)D_\rho^{db}(y)G_{\mathrm{loc}}^{ab}(x,y)\big|_{x\to x+y}&\to
-B_\sigma^{ca}(y)B_\rho^{ad}(y)R_0^\Lambda(x)
-4B_\sigma^{ca}(y)B_\mu^{ad}(y)\partial_{x^\rho}\partial_{x^\mu}\theta(x)\\
&\phantom{\to}\,\,
+2\partial_{y^\rho}B_\mu^{cd}(y)\partial_{x^\sigma}\partial_{x^\mu}\theta(x)
-\partial_{y^\eta}B_\mu^{cd}(y)x^\eta\partial_{x^\rho}R_0^\Lambda(x).
\end{align*}
Substituting the obtained transitions and using auxiliary integrals from Section \ref{ym:sec:two-vs}, we obtain the equality
\begin{align}\label{ya-a-151}
\mathrm{J}_{\sigma\rho}^{\phantom{1}}[\mathcal{N}_{1\sigma\rho}]\stackrel{\mathrm{s.p.}}{=}&
-\frac{c_2^2}{2}\Big(\hat{\mathrm{I}}_1-\hat{\mathrm{I}}_5\Big)\int_{\mathbb{R}^4}\mathrm{d}^4y\,f^{abc}
B_\sigma^a(y)B_\rho^b(y)
F_{\sigma\rho}^{c}(y)\\\nonumber&-
\frac{c_2^2}{2}\Big(\hat{\mathrm{I}}_1-\hat{\mathrm{I}}_6\Big)
\int_{\mathbb{R}^4}\mathrm{d}^4y\,
\Big(\partial_{y^\sigma}B_\rho^c(y)\Big)
F_{\sigma\rho}^{c}(y).
\end{align}
Further, repeating all the calculations performed for the second functional, we obtain an asymptotic decomposition for it in the form
\begin{align}\label{ya-a-152}
\mathrm{\hat{J}}_{\sigma\rho}^{\phantom{1}}[\mathcal{N}_{1\sigma\rho}]\stackrel{\mathrm{s.p.}}{=}&
-\frac{c_2^2}{2}\Big(2\hat{\mathrm{I}}_1-\hat{\mathrm{I}}_2-2\hat{\mathrm{I}}_5\Big)\int_{\mathbb{R}^4}\mathrm{d}^4y\,f^{abc}
B_\sigma^a(y)B_\rho^b(y)
F_{\sigma\rho}^{c}(y)\\\nonumber&-
\frac{c_2^2}{2}\Big(2\hat{\mathrm{I}}_1-\hat{\mathrm{I}}_2-4\hat{\mathrm{I}}_6\Big)
\int_{\mathbb{R}^4}\mathrm{d}^4y\,
\Big(\partial_{y^\sigma}B_\rho^c(y)\Big)
F_{\sigma\rho}^{c}(y).
\end{align}
Finally, substituting the results obtained in \eqref{ya-a-138} and using asymptotic expansions for integrals from Section \ref{ym:sec:two-vs}, we obtain the final expression in the form
\begin{align}\label{ya-a-139}
	\mathbb{H}_0^{\mathrm{sc}}\big(\Gamma_3^2\big)\Big|_{G_{\mathrm{loc}}G_{\mathrm{loc}}\mathcal{N}}\stackrel{\mathrm{s.p.}}{=}&
-\frac{3c_2^2}{16(4\pi^2)^2}\Big(L^2+L(2\rho_1-2\rho_2)\Big)\int_{\mathbb{R}^4}\mathrm{d}^4y\,f^{abc}
B_\sigma^a(y)B_\rho^b(y)
F_{\sigma\rho}^{c}(y)\\\nonumber&
-\frac{6c_2^2}{16(4\pi^2)^2}\Big(L^2+L(-1/4+2\rho_1-\rho_2)\Big)
\int_{\mathbb{R}^4}\mathrm{d}^4y\,
\Big(\partial_{y^\sigma}B_\rho^c(y)\Big)
F_{\sigma\rho}^{c}(y).
\end{align}
Note that the answer is not invariant with respect to the gauge transformations of the background field \eqref{ya-p-1} at the level of non-leading "logarithms". Thus, the classical action is split into two parts, each of which has its own correction. This fact is a consequence of the non-covariance of the deformation, see Section \ref{ym:sec:pr:gen3-2}, and leads to the need to calculate renormalization constants for individual parts of the classical action. Moreover, in the case of contributions proportional to $\sim L^2$, the invariance is preserved.

\subsubsection{Part $G_{\mathrm{loc}}G_{\mathrm{loc}}\mathcal{L}$ for $\mathbb{H}_0^{\mathrm{sc}}(\Gamma_3^2)$}
\label{ym:sec:two-l-3}

In this case, it is necessary to perform similar calculations after substituting $\mathcal{L}$ instead of $PS_1$ into the representation from \eqref{ya-a-135}, however, the replacements of the integral kernels will be different. Analyzing the orders of the singularity for all parts, we conclude that the covariant derivatives can be replaced by the ordinary ones, while the kernels can be replaced by the main parts. In other words, after shifting the variable, the following substitutions are valid for densities
\begin{equation*}
	\delta_{\mu\nu}^{\phantom{1}}G_{\mathrm{loc}}^{ab}(x+y,y)\to\delta_{\mu\nu}^{\phantom{1}}\delta^{ab}R_0^\Lambda(x),\,\,\,
	\mathcal{L}_{\mu\nu}^{ab}(x+y,y)\to 4F_{\mu\sigma}^{ac}(y)F_{\sigma\nu}^{cb}(y)\theta(x).
\end{equation*}
In this case, for example, the functional $\mathrm{J}_{\sigma\rho}^{\phantom{1}}[\mathcal{L}_{\sigma\rho}]$ after using the definition of \eqref{ya-a-101} for the operator kernel and appropriate simplification
\begin{equation*}
	\mathcal{R}_{\sigma\rho}^{ab}(x+y,y)\to c_2
	\Big(\partial_{x^\sigma}\partial_{x^\rho}R_0^\Lambda(x)\Big)R_0^\Lambda(x),
\end{equation*}
taking into account the auxiliary integral from Section \eqref{ym:sec:two-vs}, it is represented by the formula
\begin{equation*}
\mathrm{J}_{\sigma\rho}^{\phantom{1}}[\mathcal{L}_{\sigma\rho}]\stackrel{\mathrm{s.p.}}{=}
-c_2^2\hat{\mathrm{I}}_3^{\phantom{1}}W_{-1}^{\phantom{1}}.
\end{equation*}
In a similar manner, the relations are proved.
\begin{equation}\label{ya-a-154}
\mathrm{J}_{\sigma\rho}^{\phantom{1}}[\mathcal{L}_{\sigma\rho}]\stackrel{\mathrm{s.p.}}{=}
\mathrm{J}_{\sigma\rho}^{\phantom{1}}[\mathcal{L}_{\rho\sigma}]\stackrel{\mathrm{s.p.}}{=}\frac{1}{4}
\mathrm{J}_{\sigma\sigma}^{\phantom{1}}[\mathcal{L}_{\rho\rho}],
\end{equation}
\begin{equation}\label{ya-a-155}
	\mathrm{\hat{J}}_{\sigma\rho}^{\phantom{1}}[\mathcal{L}_{\sigma\rho}]\stackrel{\mathrm{s.p.}}{=}
	\mathrm{\hat{J}}_{\sigma\rho}^{\phantom{1}}[\mathcal{L}_{\rho\sigma}]\stackrel{\mathrm{s.p.}}{=}\frac{1}{4}
	\mathrm{\hat{J}}_{\sigma\sigma}^{\phantom{1}}[\mathcal{L}_{\rho\rho}]\stackrel{\mathrm{s.p.}}{=}
	-c_2^2\hat{\mathrm{I}}_4^{\phantom{1}}W_{-1}^{\phantom{1}}.
\end{equation}
Therefore, by summing up the results obtained above, we arrive at the answers for each individual diagram
\begin{align*}
	{\centering\adjincludegraphics[width = 1 cm, valign=c]{fig/ya8.eps}}\Big|_{G_{\mathrm{loc}}G_{\mathrm{loc}}\mathcal{L}}&\stackrel{\mathrm{s.p.}}{=}
	c_2^2W_{-1}^{\phantom{1}}\Big(2\hat{\mathrm{I}}_3^{\phantom{1}}+\hat{\mathrm{I}}_4^{\phantom{1}}\Big)\times\frac{1}{2},\\
	-{\centering\adjincludegraphics[width = 1 cm, valign=c]{fig/ya7.eps}}\Big|_{G_{\mathrm{loc}}G_{\mathrm{loc}}\mathcal{L}}&\stackrel{\mathrm{s.p.}}{=}
	c_2^2W_{-1}^{\phantom{1}}\Big(2\hat{\mathrm{I}}_3^{\phantom{1}}+\hat{\mathrm{I}}_4^{\phantom{1}}\Big)\times4,\\
	{\centering\adjincludegraphics[width = 1 cm, valign=c]{fig/ya9.eps}}\Big|_{G_{\mathrm{loc}}G_{\mathrm{loc}}\mathcal{L}}&\stackrel{\mathrm{s.p.}}{=}
	c_2^2W_{-1}^{\phantom{1}}\Big(2\hat{\mathrm{I}}_3^{\phantom{1}}+\hat{\mathrm{I}}_4^{\phantom{1}}\Big)\times2,\\
	-2{\centering\adjincludegraphics[width = 1 cm, valign=c]{fig/ya10.eps}}\Big|_{G_{\mathrm{loc}}G_{\mathrm{loc}}\mathcal{L}}&\stackrel{\mathrm{s.p.}}{=}
	c_2^2W_{-1}^{\phantom{1}}\Big(2\hat{\mathrm{I}}_3^{\phantom{1}}+\hat{\mathrm{I}}_4^{\phantom{1}}\Big)\times(-1),\\
	{\centering\adjincludegraphics[width = 1 cm, valign=c]{fig/ya11.eps}}\Big|_{G_{\mathrm{loc}}G_{\mathrm{loc}}\mathcal{L}}&\stackrel{\mathrm{s.p.}}{=}
	c_2^2W_{-1}^{\phantom{1}}\Big(2\hat{\mathrm{I}}_3^{\phantom{1}}+\hat{\mathrm{I}}_4^{\phantom{1}}\Big)\times(-1),
\end{align*}
and also for their specific sum in the form
\begin{equation}\label{ya-a-140}
\mathbb{H}_0^{\mathrm{sc}}\big(\Gamma_3^2\big)\Big|_{G_{\mathrm{loc}}G_{\mathrm{loc}}\mathcal{L}}\stackrel{\mathrm{s.p.}}{=}\frac{9c_2^2W_{-1}}{2}\Big(2\hat{\mathrm{I}}_3+\hat{\mathrm{I}}_4\Big)
\stackrel{\mathrm{s.p.}}{=}\frac{9c_2^2W_{-1}}{16(4\pi^2)^2}
\Big(L^2+L(1+2\rho_1+24\rho_3\rho_5-8\rho_4)\Big).
\end{equation}

\subsubsection{Part $G_{\mathrm{loc}}G_{\mathrm{loc}}G_{\mathrm{loc}}$ for $\mathbb{H}_0^{\mathrm{sc}}(\Gamma_3^2)$}
\label{ym:sec:two-l-4}

The last contribution of the local type, appearing in diagrams from $\mathbb{H}_0^{\mathrm{sc}}\big(\Gamma_3^2\big)$, can actually be reduced to the following single integral
\begin{equation}\label{ya-a-141}
\mathrm{J}_{\ominus}[B]=
\int_{\mathrm{B}_{1/\sigma}}\mathrm{d}^4x
\int_{\mathbb{R}^4}\mathrm{d}^4y\,\Big(
f^{ceg}
\Big(D_{\sigma}^{ca}(x)D_\sigma^{db}(y)G_{\mathrm{loc}}^{ab}(x,y)\Big)
G_{\mathrm{loc}}^{eh}(x,y)
G_{\mathrm{loc}}^{gf}(x,y)
f^{fhd}-\hat{\kappa}\Big)\Big|_{x\to x+y},
\end{equation}
where the constant $\hat{\kappa}$ subtracts the part from the first term that does not depend on the background field. To show this, it is necessary to replace the Green's function $G_{1\mu\nu}^\Lambda$ in each of the five diagrams in \eqref{ya-a-19} with the local component $\delta_{\mu\nu}G_{\mathrm{loc}}$, integrate by parts so that both derivatives act on the same local function, then shift the variable $x\to x+y$ so that each term is factorized into a singular part and some functional depending on the background field, and finally replace the integration domain $\mathbb{R}^{4\times2}$ with $\mathrm{B}_{1/\sigma}\times\mathbb{R}^4$. This set of manipulations transforms each diagram as follows
\begin{equation}\label{ya-a-142}
{\centering\adjincludegraphics[width = 1 cm, valign=c]{fig/ya8.eps}}\to-\mathrm{J}_{\ominus}/2
,\,\,\,
{\centering\adjincludegraphics[width = 1 cm, valign=c]{fig/ya7.eps}}\to4\mathrm{J}_{\ominus}
,\,\,\,
{\centering\adjincludegraphics[width = 1 cm, valign=c]{fig/ya9.eps}}\to-2\mathrm{J}_{\ominus}
,\,\,\,
{\centering\adjincludegraphics[width = 1 cm, valign=c]{fig/ya10.eps}}\to-\mathrm{J}_{\ominus}/2
,\,\,\,
{\centering\adjincludegraphics[width = 1 cm, valign=c]{fig/ya11.eps}}\to\mathrm{J}_{\ominus}.
\end{equation}
Therefore, substituting the transformed parts into \eqref{ya-a-19}, we arrive at the relation
\begin{equation}\label{ya-a-143}
\mathbb{H}_0^{\mathrm{sc}}\big(\Gamma_3^2\big)\Big|_{G_{\mathrm{loc}}G_{\mathrm{loc}}G_{\mathrm{loc}}}\stackrel{\mathrm{s.p.}}{=}-\frac{9\mathrm{J}_{\ominus}}{2}.
\end{equation}

\subsubsection{Parts for $\mathbb{H}_0^{\mathrm{sc}}(\Omega_3^2)$}
\label{ym:sec:two-l-5}
In this case, we can use calculations for the diagram group $\mathbb{H}_0^{\mathrm{sc}}\big(\Gamma_3^2\big)$, adapting some calculations as necessary. First, a local contribution of the type $G_{\mathrm{loc}}\mathcal{N}\mathcal{N}$ will be absent due to the presence of only one the Green's function for the vector operator. Let us comment on the remaining contributions in more detail.\\

\noindent1) The contribution of $G_{\mathrm{loc}}G_{\mathrm{loc}}\mathcal{N}$ can be found by cutting the midline in diagram \eqref{ya-a-20} and further moving to the prepared auxiliary integrals $\mathrm{J}_{\sigma\rho}[\,\cdot\,]$ and $\mathrm{\hat{J}}_{\sigma\rho}[\,\cdot\,]$. To do this, we can use the already obtained combination for the integral kernel \eqref{ya-a-153} from Section \ref{ym:sec:two-nl-3}. Only instead of the test function $PS_1$, it is necessary to substitute $\mathcal{N}$, then the desired singular combination is contained in the difference
\begin{equation}\label{ya-a-144}
\mathrm{J}_{\sigma\rho}[\mathcal{N}_{\sigma\rho}]-\frac{1}{2}\mathrm{\hat{J}}_{\sigma\rho}[\mathcal{N}_{\sigma\rho}].
\end{equation}
Next, using the answers from \eqref{ya-a-151} and \eqref{ya-a-152} and auxiliary asymptotic decompositions from Section \ref{ym:sec:two-vs}, we arrive at the answer
\begin{align}\label{ya-a-145}
	\mathbb{H}_0^{\mathrm{sc}}\big(\Omega_3^2\big)\Big|_{G_{\mathrm{loc}}G_{\mathrm{loc}}\mathcal{N}}\stackrel{\mathrm{s.p.}}{=}&
	-\frac{c_2^2}{16(4\pi^2)^2}\Big(L^2+L(2\rho_1+2\rho_2)\Big)\int_{\mathbb{R}^4}\mathrm{d}^4y\,f^{abc}
	B_\sigma^a(y)B_\rho^b(y)
	F_{\sigma\rho}^{c}(y)\\\nonumber&
	-\frac{2c_2^2}{16(4\pi^2)^2}\Big(L^2+L(1/4+2\rho_1+\rho_2)\Big)
	\int_{\mathbb{R}^4}\mathrm{d}^4y\,
	\Big(\partial_{y^\sigma}B_\rho^c(y)\Big)
	F_{\sigma\rho}^{c}(y).
\end{align}
\noindent2) Contribution $G_{\mathrm{loc}}G_{\mathrm{loc}}\mathcal{L}$ is searched in a similar way. It is enough to use representation \eqref{ya-a-144}, replace the function $\mathcal{N}_{\mu\nu}$ with $\mathcal{L}_{\mu\nu}$, and then apply the already known relations \eqref{ya-a-154} and \eqref{ya-a-155}. As a result, we get
\begin{align}\label{ya-a-146}
\mathbb{H}_0^{\mathrm{sc}}\big(\Omega_3^2\big)\Big|_{G_{\mathrm{loc}}G_{\mathrm{loc}}\mathcal{L}}&\stackrel{\mathrm{s.p.}}{=}
\mathrm{J}_{\sigma\rho}[\mathcal{L}_{\sigma\rho}]-\frac{1}{2}\mathrm{\hat{J}}_{\sigma\rho}[\mathcal{L}_{\sigma\rho}]\stackrel{\mathrm{s.p.}}{=}-
\frac{c_2^2W_{-1}^{\phantom{1}}}{2}\Big(2\hat{\mathrm{I}}_3-\hat{\mathrm{I}}_4\Big)\\\nonumber
&\stackrel{\mathrm{s.p.}}{=}\frac{c_2^2W_{-1}}{16(4\pi^2)^2}
\Big(L^2+L(1+2\rho_1-24\rho_3\rho_5+8\rho_4)\Big).
\end{align}
\noindent3) Contribution $G_{\mathrm{loc}}G_{\mathrm{loc}}G_{\mathrm{loc}}$ can be found using the procedure described in Section \ref{ym:sec:two-l-4}. It is clear that the diagram of \eqref{ya-a-20} is equal to $1/4$ of the middle diagram of \eqref{ya-a-142}. Therefore, we have
\begin{equation}\label{ya-a-147}
\mathbb{H}_0^{\mathrm{sc}}\big(\Omega_3^2\big)\Big|_{G_{\mathrm{loc}}G_{\mathrm{loc}}G_{\mathrm{loc}}}\stackrel{\mathrm{s.p.}}{=}\frac{\mathrm{J}_{\ominus}}{2}.
\end{equation}

\subsubsection{Parts for $\mathbb{H}_0^{\mathrm{sc}}(\Gamma_4)$}
\label{ym:sec:two-l-6}
As before, let us split the calculation into several parts. At the beginning of Section \ref{ym:sec:two-l}, contributions that may contain a non-zero singular part were highlighted. The remaining parts are not considered due to obvious reasons: either the singular part is missing in principle, or it is reduced due to the convolution of a symmetric object with an antisymmetric one.\\

\noindent1) Consider the contribution of $\mathcal{N}\mathcal{N}$. In this case, the function $\mathcal{N}_{\mu\nu}^{ab}(x,x)$ can be replaced by the main part of the expansion of $2F_{\mu\nu}^{ab}(x)\theta(0)$, since the remaining parts either vanish or are of the order $L/\Lambda^2$. In this case, after summing up the group indexes, we arrive at the following values
\begin{equation*}
{\centering\adjincludegraphics[width = 1 cm, valign=c]{fig/ya13.eps}}
\to2c_2^2W_{-1}^{\phantom{1}}\theta^2(0)
,\,\,\,
-
{\centering\adjincludegraphics[width = 1 cm, valign=c]{fig/ya14.eps}}
\to0
,\,\,\,
{\centering\adjincludegraphics[width = 2 cm, valign=c]{fig/ya15.eps}}
\to4c_2^2W_{-1}^{\phantom{1}}\theta^2(0)
.
\end{equation*}
Therefore, in total we get
\begin{equation}\label{ya-a-149}
\mathbb{H}_0^{\mathrm{sc}}\big(\Gamma_4\big)\Big|_{\mathcal{N}\mathcal{N}}\stackrel{\mathrm{s.p.}}{=}\frac{3c_2^2W_{-1}^{\phantom{1}}}{2(4\pi^2)^2}\Big(L^2+2L\rho_1\Big).
\end{equation}
\noindent2) Consider the case $G_{\mathrm{loc}}\mathcal{L}$. Note that due to the validity of the relations from \eqref{ya-a-156}, the local function $G_{\mathrm{loc}}^{ab}(x,x)$ can be replaced by the main part of the expansion $\delta^{ab}R_0^\Lambda(0)$. Substituting the function $\mathcal{L}$ from \eqref{ya-a-77}, we obtain the relations
\begin{equation*}
	{\centering\adjincludegraphics[width = 1 cm, valign=c]{fig/ya13.eps}}
	\to-8c_2^2W_{-1}^{\phantom{1}}R_0^\Lambda(0)\tau(0)
	,\,\,\,
	-
	{\centering\adjincludegraphics[width = 1 cm, valign=c]{fig/ya14.eps}}
	\to32c_2^2W_{-1}^{\phantom{1}}R_0^\Lambda(0)\tau(0)
	,\,\,\,
	{\centering\adjincludegraphics[width = 2 cm, valign=c]{fig/ya15.eps}}
	\to0
	,
\end{equation*}
from which, after summing up and using the decompositions from Section \ref{ym:sec:two-vs}, the answer follows
\begin{equation}\label{ya-a-148}
	\mathbb{H}_0^{\mathrm{sc}}\big(\Gamma_4\big)\Big|_{G_{\mathrm{loc}}\mathcal{L}}\stackrel{\mathrm{s.p.}}{=}\frac{9c_2^2W_{-1}^{\phantom{1}}}{\pi^2}LR_0^1(0)\rho_5^{\phantom{1}}.
\end{equation}
\noindent3) The latter case corresponds to the combination $G_{\mathrm{loc}}G_{\mathrm{loc}}$. Note that due to the same relations from \eqref{ya-a-156}, one of the functions can be replaced by the main part of the decomposition, and the result doubled. Let us introduce the auxiliary functionality
\begin{equation}\label{ya-a-150}
\mathrm{J}_{\odot}[B]=\int_{\mathbb{R}^4}\mathrm{d}^4x\,
\Big(G_{\mathrm{loc}}^{aa}(x,x)-\dot{\kappa}\Big),
\end{equation}
where the constant $\dot{\kappa}$ subtracts the part of the density that does not depend on the background field. Then, by direct calculation, we can verify the validity of the transitions
\begin{equation*}
	{\centering\adjincludegraphics[width = 1 cm, valign=c]{fig/ya13.eps}}
	\to-8c_2\mathrm{J}_{\odot}R_0^\Lambda(0)
	,\,\,\,
	-
	{\centering\adjincludegraphics[width = 1 cm, valign=c]{fig/ya14.eps}}
	\to32c_2\mathrm{J}_{\odot}R_0^\Lambda(0)
	,\,\,\,
	{\centering\adjincludegraphics[width = 2 cm, valign=c]{fig/ya15.eps}}
	\to0.
\end{equation*}
Summing up, we get the answer in the form
\begin{equation}\label{ya-a-157}
\mathbb{H}_0^{\mathrm{sc}}\big(\Gamma_4\big)\Big|_{G_{\mathrm{loc}}G_{\mathrm{loc}}}\stackrel{\mathrm{s.p.}}{=}24c_2\mathrm{J}_{\odot}R_0^\Lambda(0).
\end{equation}

\subsection{Counter-diagrams}
\label{ym:sec:two-co}
In this section, we study decompositions into local and nonlocal parts for emerging counter-diagrams and some additional auxiliary parts of the Green's functions.\\

\noindent\textbf{"Massive" counter-vertex.} When researching the first correction, it became necessary to introduce a massive\footnote{A more detailed analysis of the extension of the classical action is given in Section \ref{ym:sec:mas}.} term in the classical action of the Yang--Mills theory due to the appearance of a power singularity $\Lambda^2$ proportional to the second degree of the background field. This fact led to the appearance of the $S_2$ counter-vertex with two external lines. Since the introduced vertex is proportional to the second power of the coupling constant, a counter-diagram appears in the two-loop approximation equal to the trace of the Green's function for the vector operator. Let us break it down into parts using the decompositions from \eqref{ya-a-74} and \eqref{ya-a-128}. After substitution, we get
\begin{equation}\label{ya-a-158}
\Lambda^2\int_{\mathbb{R}^4}\mathrm{d}^4x\,\Big(G_{1\mu\mu}^{\Lambda aa}(x,x)-4G_{\mathrm{loc}}^{aa}(x,x)\Big)\stackrel{\mathrm{s.p.}}{=}
\Lambda^2\int_{\mathbb{R}^4}\mathrm{d}^4x\,\Big(PS_{1\mu\mu}^{\,\, aa}(x,x)+4PS_{0}^{aa}(x,x)\Big)
+\frac{3c_2L\rho_5}{2\pi^2}W_{-1},
\end{equation}
where a relation from Section \ref{ym:sec:two-vs} was used to calculate $\Lambda^2\tau(0)$.\\

\noindent\textbf{"Gauge" counter-vertex.} When studying the non-local parts for two-loop diagrams, in Section \ref{ym:sec:two-nl-1}, a term with density $D_\rho^{bc}(x)D_\sigma^{ca}(x)PS_{1\sigma\rho}^{\,\,ab}(x,y)|_{y=x}$ appeared, multiplied by $L=\ln(\Lambda/\sigma)$. This fact means that it is necessary to introduce a renormalization constant for term \eqref{ya-a-11}, which fixes the gauge condition. In particular, this leads to the appearance of an additional counter-diagram in the two-loop order. Let us represent decomposition for it
\begin{align}\label{ya-a-159}
	L\int_{\mathbb{R}^4}\mathrm{d}^4x\,\Big(D_\rho^{bc}(x)D_\sigma^{ca}(x)&G_{1\sigma\rho}^{\Lambda ab}(x,y)+M_0^{ab}(x)G_{\mathrm{loc}}^{ba}(x,y)\Big)\Big|_{y=x}\stackrel{\mathrm{s.p.}}{=}-\frac{c_2L\rho_2}{8\pi^2}W_{-1}\\
	&+
L\int_{\mathbb{R}^4}\mathrm{d}^4x\,\Big(D_\rho^{bc}(x)D_\sigma^{ca}(x)PS_{1\sigma\rho}^{\,\,ab}(x,y)+M_0^{ab}(x)PS_{0}^{ba}(x,y)\Big)\Big|_{y=x},
\end{align}
where again, taking into account the definitions from Section \ref{ym:sec:two-vs}, the auxiliary equality was used
\begin{equation}\label{ya-a-160}
\theta(0)-A_0(x)\tau(x)\big|_{x=0}=-\rho_2/(8\pi^2),
\end{equation}
which is true for all $\Lambda/\sigma>2$. Additionally, we introduce a definition for the local part
\begin{equation}\label{ya-a-187}
	\mathrm{J}_{\otimes}[B]=\int_{\mathbb{R}^4}\mathrm{d}^4x\,
	\Big(M_0^{ab}(x)G_{\mathrm{loc}}^{ba}(x,y)\Big|_{y=x}-\tilde{\kappa}\Big),
\end{equation}
where $\tilde{\kappa}$ subtracts the density independent of the background field.\\

\noindent\textbf{Decomposition for $LM_1PS_{1}$.} It can be noted that after summing all the nonlocal terms from Section \eqref{ym:sec:two-nl}, the operator $M_1$ is formed, acting on the part of $PS_1$ and multiplied by the logarithmic singularity $L$. Let us calculate the trace of this value. We use the definition of \eqref{ya-a-38} for $e_\mu^a=B_\mu^a$ and the definition for $PS_1$ from \eqref{ya-a-74}. Further, taking advantage of the antisymmetry of the field stress tensor, we note that the contribution to the integral
\begin{equation*}
L\int_{\mathbb{R}^4}\mathrm{d}^4x\,\Big(M_{1\mu\nu}^{\,\,ab}(x)PS_{1\nu\mu}^{\,\,ba}(x,y)\big|_{y=x}-\kappa_{ps_1}\Big)
\end{equation*}
is nonzero and depends on the background field, where $\kappa_{ps_1}$ subtracts a singular density, independent of the background field, and can be represented as three terms
\begin{align*}
4L&\int_{\mathbb{R}^4}\mathrm{d}^4x\int_{\mathbb{R}^4}\mathrm{d}^4y\int_{\mathbb{R}^4}\mathrm{d}^4z\,
\Big(A_0(x)R_0^\Lambda(x-y)\Big)F_{\mu\nu}^{ab}(y)R_0^\Lambda(y-z)F_{\nu\mu}^{ba}(z)R_0^\Lambda(z-x)\\
-4L&\int_{\mathbb{R}^4}\mathrm{d}^4x\int_{\mathbb{R}^4}\mathrm{d}^4y\,
F_{\mu\nu}^{ab}(x)R_0^\Lambda(x-y)F_{\nu\mu}^{ba}(y)R_0^\Lambda(y-x)
+4c_2LW_{-1}\Big(\theta(0)-A_0(x)\tau(x)\big|_{x=0}\Big).
\end{align*}
The remaining parts lead to finite correction terms due to the properties of the Green's function, see \eqref{ya-a-35}. Using formula \eqref{ya-a-160}, we see that the last part is proportional to $\rho_2$. Further, by adding, subtracting, and shifting a variable, it can be shown that the first two terms in the main order are also proportional to $\rho_2$, but with the opposite sign. Moreover, the correction parts are small with respect to the logarithm of $L$. Therefore, the final relation has the form
\begin{equation}\label{ya-a-161}
	L\int_{\mathbb{R}^4}\mathrm{d}^4x\,\Big(M_{1\mu\nu}^{\,\,ab}(x)PS_{1\nu\mu}^{\,\,ba}(x,y)\big|_{y=x}-\kappa_{ps_1}\Big)\stackrel{\mathrm{s.p.}}{=}0.
\end{equation}

\noindent\textbf{Decomposition for $LM_0PS_{0}$.} This statement is similar to the previous one and can be mathematically formulated as follows
\begin{equation}\label{ya-a-162}
	L\int_{\mathbb{R}^4}\mathrm{d}^4x\,\Big(M_{0}^{ab}(x)PS_{0}^{ba}(x,y)\big|_{y=x}-\kappa_{ps_0}\Big)\stackrel{\mathrm{s.p.}}{=}0,
\end{equation}
where $\kappa_{ps_0}$ subtracts a singular density independent of the background field. The main method of its proof is based on explicit application of the operator and further asymptotic expansion.


\subsection{Decomposition of local contributions}
\label{ym:sec:two-llo}
As the calculations progressed, three auxiliary functions appeared: $\mathrm{J}_{\otimes}$, $\mathrm{J}_{\odot}$, and $\mathrm{J}_{\ominus}$, which are constructed using the local component $G_{\mathrm{loc}}$ from expansion \eqref{ya-a-74-1} near the diagonal for the Green's function. Let us study them in more detail. To do this, we use the set of parts of the classical action $W_{-1}^i$, where $i\in\{\pm,1,\ldots,6\}$, from Theorem \ref{ya-t1} from Section \ref{ym:sec:pr:gen3-7}.\\

\noindent\textbf{Decomposition for $\mathrm{J}_{\otimes}$.} The definition for $\mathrm{J}_{\otimes}$ is given in formula \eqref{ya-a-187}. The basis here is the combination of $M_0^{ab}(x)G_{\mathrm{loc}}^{ba}(x,y)$ for $x=y$. Next, let us pay attention to two facts. First, the component $G_{\mathrm{loc}}$ in expansion \eqref{ya-a-128} was selected in such a way that $M_0^{ab}(x)PS_{0}^{ba}(x,y)$ for $x=y$ had behavior no worse than $L^{-1}$. Secondly, applying the operator $M_0^{ab}(x)$ to the unregularized scalar Green's function zeros (reduces) all nonzero orders in the background field. Therefore, the final part for $M_0^{ab}(x)G_{\mathrm{loc}}^{ba}(x,y)$ for $x=y$ should be represented as a linear combination of $S_2[B]$ and $W_{-1}^i[B]$ with coefficients that are differences depending on the deformed function $A_0^{\phantom{1}}(u)R_0^1(u-z)$. Using \eqref{ya-a-127} to calculate the exact coefficient for $S_2[B]$, the answer can be represented as follows
\begin{equation}\label{ya-a-188}
L\mathrm{J}_{\otimes}[B]\stackrel{\mathrm{s.p.}}{=}
\frac{c_2L\Lambda^2}{4\pi^2}S_2[B](\rho_0+\rho_6-2\rho_3)
+L\sum_{i=1}^6c_{\otimes}^iW_{-1}^i[B],
\end{equation}
where the notation from Section \eqref{ym:sec:two-vs} was used. Note that the coefficients $c_{\otimes}^i$ are finite and real, they depend on the function $\mathbf{f}(\cdot)$ from \eqref{ya-a-37}, and show a deviation from the covariant weak deformation from Section \ref{ym:sec:two1}.\\

\noindent\textbf{Decomposition for $\mathrm{J}_{\odot}$.} The definition for $\mathrm{J}_{\odot}$ is given in formula \eqref{ya-a-150}. It is important to pay attention to the fact that in the unregularized case, the logarithm of $\ln(|x-y|)$ near the degree of $|x-y|^2$ is actually absent. To understand this, it is enough to look at the structure of the Seeley--DeWitt coefficients near the diagonal, see for reference the decomposition of Green's functions near the diagonal \cite{29,30-1-1 } and the heat kernel method \cite{28,vas1,vas2}. Thus, the main contribution will be proportional to the linear combination of $S_2[B]$ and $W_{-1}^i[B]$. In this case, the coefficients will be represented by differences depending on the deformed function $A_0(u)R_0^1(u-z)$. Again, explicitly calculating the contribution proportional to $S[B]$, we get the answer in the form
\begin{equation}\label{ya-a-190}
\Lambda^2\mathrm{J}_{\odot}[B]\stackrel{\mathrm{s.p.}}{=}
	\frac{c_2\rho_2\Lambda^2}{8\pi^2}S_2[B]
	+L\sum_{i=1}^6c_{\odot}^iW_{-1}^i[B].
\end{equation}
As in the previous situation \eqref{ya-a-188}, the coefficients $c_{\odot}^i$ appeared due to "deviation" from the covariant case.\\

\noindent\textbf{Decomposition for $\mathrm{J}_{\ominus}$.} The definition for $\mathrm{J}_{\ominus}$ is given in formula \eqref{ya-a-141}. In this case, it is necessary to pay attention to the fact that the functional $\mathrm{J}_{\ominus}$ can contain either a power-law singularity proportional to $S_2[B]$, or a linear combination of the classical action $W_{-1}[B]$ and its parts $W_{-1}^i[B]$ multiplied by the first power of the logarithm of $L$. At the same time, the appearance of the second degree of the logarithm $L^2$ is impossible, since such a singularity does not depend on the type of regularization, which means that, for example, regularization for the covariant "weak" case can be used to determine it, see Section \ref{ym:sec:two1-2}. Moreover, it follows from the same section that the part proportional to the first power of the logarithm is represented as
\begin{equation*}
\frac{c_2^2L}{(4\pi)^4}W_{-1}[B]+L\sum_{i=1}^6c_{\ominus}^iW_{-1}^i[B].
\end{equation*}
It is equal to the sum of a covariant part and a set of additives containing coefficients $c_{\ominus}^i$, which should be considered as a deviation from the covariant case in violation of the gauge invariance, see \eqref{ya-p-1}. In the case of the term with $S_2[B]$, it is enough to sort through all possible variants using the decomposition from \eqref{ya-a-74-1}. Thus, taking into account the notation from Section \ref{ym:sec:two-si}, we get
\begin{equation}\label{ya-a-191}
	\mathrm{J}_{\ominus}[B]\stackrel{\mathrm{s.p.}}{=}
	\frac{c_2^2\rho_7\Lambda^2}{(4\pi^2)^2}S_2[B]
	+\frac{c_2^2L}{(4\pi)^4}W_{-1}[B]+L\sum_{i=1}^6c_{\ominus}^iW_{-1}^i[B].
\end{equation}

\subsection{Singular part: sum of contributions}
\label{ym:sec:two-si}
As the main result obtained in Section \ref{ym:sec:two}, we present the singular part for the first regularized non-renormalized quantum correction $W_1^{\Lambda}[B]$, the formula for which is given in \eqref{ya-a-163}. To do this, we use the formulas obtained for the singular parts for each individual component:
\begin{align*}
\mathbb{H}_0^{\mathrm{sc}}\big(\Gamma_3^2\big):\,\,\mbox{nonlocal part}&\longrightarrow\,
\mbox{see}\,\,\eqref{ya-a-121},\,\,\eqref{ya-a-129},\,\,\mbox{and}\,\,\eqref{ya-a-122};\\
\mbox{local part}&\longrightarrow\,
\mbox{see}\,\,\eqref{ya-a-126},\,\,\eqref{ya-a-139},\,\,\eqref{ya-a-140},\,\,\mbox{and}\,\,\eqref{ya-a-143};\\
\mathbb{H}_0^{\mathrm{sc}}\big(\Gamma_4^{\phantom{1}}\big):\,\,\mbox{nonlocal part}&\longrightarrow\,
\mbox{see}\,\,\eqref{ya-a-132},\,\,\eqref{ya-a-133},\,\,\mbox{and}\,\,\eqref{ya-a-134};\\
\mbox{local part}&\longrightarrow\,
\mbox{see}\,\,\eqref{ya-a-149},\,\,\eqref{ya-a-148},\,\,\mbox{and}\,\,\eqref{ya-a-157};\\
\mathbb{H}_0^{\mathrm{sc}}\big(\Omega_3^2\big):\,\,\mbox{nonlocal part}&\longrightarrow\,
\mbox{see}\,\,\eqref{ya-a-130},\,\,\mbox{and}\,\,\eqref{ya-a-131};\\
\mbox{local part}&\longrightarrow\,
\mbox{see}\,\,\eqref{ya-a-145},\,\,\eqref{ya-a-146},\,\,\mbox{and}\,\,\eqref{ya-a-147}.
\end{align*}
Additionally, we use the relations from \eqref{ya-a-161} and \eqref{ya-a-162}. Then, after summing up all the results obtained, taking into account the coefficients from \eqref{ya-a-163}, we get the answer in the form
\begin{align}\label{ya-a-166}
	-\frac{1}{2}\mathbb{H}_0^{\mathrm{sc}}\big(\Gamma_3^2\big)
+\frac{1}{4}\mathbb{H}_0^{\mathrm{sc}}\big(\Gamma_4^{\phantom{1}}\big)-\frac{1}{2}\mathbb{H}_0^{\mathrm{sc}}\big(\Omega_3^2\big)\stackrel{\mathrm{s.p.}}{=}&
-\frac{\Lambda^2c_2(2\rho_3-3\rho_0)}{8\pi^2}
\int_{\mathbb{R}^4}\mathrm{d}^4x\,\Big(PS_{1\rho\rho}^{\,\,aa}(x,x)+4PS_{0}^{aa}(x,x)\Big)\,\,\,\,\,\,\,\,\,\\\nonumber
&-\frac{5Lc_2}{48\pi^2}
\int_{\mathbb{R}^4}\mathrm{d}^4x\,
\Big(D_\rho^{bc}(x)D_\sigma^{ca}(x)PS_{1\sigma\rho}^{\,\,ab}(x,y)\Big)\Big|_{y=x}
\\\nonumber
&-\frac{c_2^2W_{-1}L}{32(4\pi^2)^2}\Big(13+16\rho_2\Big)
+\frac{c_2^2W_{-1}^{+}L}{32(4\pi^2)^2}\Big(2\rho_2-1/2\Big)\\\nonumber
&-\frac{c_2^2W_{-1}L}{(4\pi^2)^2}
\Big(3\rho_5(2\rho_3-3\rho_0)-2\rho_4\Big)+2\mathrm{J}_{\ominus}+\frac{3c_2\rho_0\Lambda^2\mathrm{J}_{\odot}}{2\pi^2}.
\end{align}
Auxiliary designations were used here for the parts of the classical action
\begin{equation}\label{ya-c-1}
W_{-1}^{+}=2
\int_{\mathbb{R}^4}\mathrm{d}^4y\,
\Big(\partial_{y^\sigma}B_\rho^c(y)\Big)
F_{\sigma\rho}^{c}(y)
,\,\,\,
W_{-1}^{-}=
\int_{\mathbb{R}^4}\mathrm{d}^4y\,f^{abc}
B_\sigma^a(y)B_\rho^b(y)
F_{\sigma\rho}^{c}(y).
\end{equation}
Note that the latter values obey the equality $W_{-1}^-+W_{-1}^+=W_{-1}^{\phantom{1}}$.

\subsection{Auxiliary relations}
\label{ym:sec:two-vs}

\noindent\textbf{Integrals for nonlocal part.} All the integrals listed below, when asymptotically decomposed with respect to the variable $\Lambda$, are proportional in the main order to $L=\ln(\Lambda/\sigma)$. This can be confirmed by direct differentiation, that is, using the operator $\Lambda\partial_\Lambda$, and further moving to the limit of $\Lambda\to+\infty$. For all integrals, such limits exist and are finite. Thus, the answer is to take the product of $L$ by the limit value. We have

\begin{fleqn}
	\begin{equation}\label{ya-a-84}
		\int_{\mathrm{B}_{1/\sigma}}\mathrm{d}^4x\,
		\Big(-\partial_{x^\sigma}\partial_{x^\rho}R_0^\Lambda(x)\Big)\partial_{x^\mu}\partial_{x^\nu}\tau(x)\stackrel{\mathrm{s.p.}}{=}-\frac{\delta_{\rho\nu}\delta_{\sigma\mu}+\delta_{\sigma\nu}\delta_{\mu\rho}+\delta_{\sigma\rho}\delta_{\nu\mu}}{24}\frac{L}{8\pi^2},
	\end{equation}
\end{fleqn}
\begin{fleqn}
	\begin{equation}\label{ya-a-83}
\int_{\mathrm{B}_{1/\sigma}}\mathrm{d}^4x\,
\Big(-\partial_{x^\sigma}\partial_{x^\rho}R_0^\Lambda(x)\Big)x^\nu\partial_{x^\mu}\theta(x)\stackrel{\mathrm{s.p.}}{=}\frac{2\delta_{\rho\nu}\delta_{\sigma\mu}+2\delta_{\sigma\nu}\delta_{\mu\rho}-\delta_{\sigma\rho}\delta_{\nu\mu}}{12}\frac{L}{8\pi^2},
	\end{equation}
\end{fleqn}
\begin{fleqn}
	\begin{equation}\label{ya-a-86-1}
		\int_{\mathrm{B}_{1/\sigma}}\mathrm{d}^4x\,
		\Big(-\partial_{x^\sigma}\partial_{x^\rho}R_0^\Lambda(x)\Big)x^\nu x^\mu R_0^\Lambda(x)\stackrel{\mathrm{s.p.}}{=}-
		\frac{2\delta_{\rho\nu}\delta_{\sigma\mu}2\delta_{\sigma\nu}\delta_{\mu\rho}-\delta_{\sigma\rho}\delta_{\nu\mu}}{6}\frac{L}{8\pi^2},
	\end{equation}
\end{fleqn}
\begin{fleqn}
	\begin{equation}\label{ya-a-851}
		\int_{\mathrm{B}_{1/\sigma}}\mathrm{d}^4x\,
		\Big(-\partial_{x^\sigma}\partial_{x^\rho}\partial_{x^\mu}\theta(x)\Big)x^\nu R_0^\Lambda(x)\stackrel{\mathrm{s.p.}}{=}-\frac{\delta_{\rho\nu}\delta_{\sigma\mu}+\delta_{\sigma\nu}\delta_{\mu\rho}+\delta_{\sigma\rho}\delta_{\nu\mu}}{12}\frac{L}{8\pi^2},
	\end{equation}
\end{fleqn}
\begin{fleqn}
	\begin{equation}\label{ya-a-85}
\int_{\mathrm{B}_{1/\sigma}}\mathrm{d}^4x\,
\Big(\partial_{x^\sigma}\partial_{x^\rho}\theta(x)\Big)\partial_{x^\mu}\partial_{x^\nu}\theta(x)\stackrel{\mathrm{s.p.}}{=}\frac{\delta_{\rho\nu}\delta_{\sigma\mu}+\delta_{\sigma\nu}\delta_{\mu\rho}+\delta_{\sigma\rho}\delta_{\nu\mu}}{24}\frac{L}{8\pi^2},
	\end{equation}
\end{fleqn}
\begin{fleqn}
	\begin{equation}\label{ya-a-82}
		\int_{\mathrm{B}_{1/\sigma}}\mathrm{d}^4x\,
		\Big(-\partial_{x^\sigma}\partial_{x^\rho}R_0^\Lambda(x)\Big)\theta(x)\stackrel{\mathrm{s.p.}}{=}
		\frac{\delta_{\sigma\rho}}{4}\frac{L}{8\pi^2},
	\end{equation}
\end{fleqn}
\begin{fleqn}
	\begin{equation}\label{ya-a-86}
\int_{\mathrm{B}_{1/\sigma}}\mathrm{d}^4x\,
\Big(-x^\sigma\partial_{x^\rho}R_0^\Lambda(x)\Big)R_0^\Lambda(x)\stackrel{\mathrm{s.p.}}{=}\frac{\delta_{\sigma\rho}}{2}\frac{L}{8\pi^2}.
	\end{equation}
\end{fleqn}
Note that in the calculation process, it is convenient to use spherical symmetry firstly in the form \eqref{ya-a-52}.\\

\noindent\textbf{Integrals for local part.} To formulate the answers, it is necessary to define several auxiliary functionals depending on the deformed free Green's function $R_0^1(x)$:

\begin{fleqn}
	\begin{equation}\label{ya-a-92}
\rho_1=\int_{\mathrm{B}_{1}}\mathrm{d}^4x\,R_0^1(x)R_0^1(x)\times8\pi^2,
	\end{equation}
\end{fleqn}
\begin{fleqn}
	\begin{equation}\label{ya-a-93}
\rho_2=
\int_{\mathrm{B}_{1}}\mathrm{d}^4x\,A_0^{\phantom{1}}(x)R_0^1(x)
\int_{\mathbb{R}^4}\mathrm{d}^4y\,\Big(R_0^1(x-y)R_0^1(y)-R_0^1(y)R_0^1(y)\Big)\times8\pi^2,
	\end{equation}
\end{fleqn}
\begin{fleqn}
	\begin{equation}\label{ya-a-94}
\rho_3=\int_{\mathrm{B}_{1}}\mathrm{d}^4x\,R_0^1(x)A_0^{\phantom{1}}(x)R_0^1(x)\times4\pi^2,
	\end{equation}
\end{fleqn}
\begin{fleqn}
\begin{equation}\label{ya-a-96}
\rho_4=\int_{\mathrm{B}_{1}}\mathrm{d}^4x\,R_0^1(x)|x|^2A_0^{\phantom{1}}(x)R_0^1(x)\times\frac{\pi^2}{2},
\end{equation}
\end{fleqn}
\begin{fleqn}
	\begin{equation}\label{ya-a-95}
\rho_5=\int_{\mathrm{B}_{1}}\mathrm{d}^4x\,\Big(R_0^1(x)-R_0^{\phantom{1}}(x)\Big),
	\end{equation}
\end{fleqn}
\begin{fleqn}
\begin{equation}\label{ya-a-189}
\rho_6=\int_{\mathrm{B}_{1}}\mathrm{d}^4x\,A_0^{\phantom{1}}(x)R_0^1(x)
\int_{\mathbb{R}^4}\mathrm{d}^4y\,R_0^1(x-y)A_0^{\phantom{1}}(y)R_0^1(y)\times4\pi^2,
\end{equation}
\end{fleqn}
as well as
\begin{fleqn}
\begin{align*}
\rho_7=16\pi^4\times\int_{\mathbb{R}^4}\mathrm{d}^4x\,\bigg(
\Big(&-2\partial_{x_\mu}\tilde{\theta}_\mu(x)-\tilde{\tau}(x)-R_0^1(x)\Big)\Big(R_0^1(x)\Big)^2\\&
+2R_0^1(x)\Big(A_0^{\phantom{1}}(x)R_0^1(x)\Big)
\Big(\theta(x)-A_0^{\phantom{1}}(x)\tau(x)\Big)\Big|_{\Lambda=1,\,\sigma\to+0}\\&+
R_0^1(x)\tilde{\theta}_\mu(x)\Big(\partial_{x_\mu}R_0^1(x)-A_0^{\phantom{1}}(x)\tilde{\theta}_\mu(x)\Big)
-\frac{1}{2}\tilde{\theta}_\mu(x)\tilde{\theta}_\mu(x)A_0^{\phantom{1}}(x)R_0^1(x)
\bigg)
,
\end{align*}
\end{fleqn}
where
\begin{align*}
\tilde{\theta}_\mu(x)&=\int_{\mathbb{R}^4}
\mathrm{d}^4z\,R_0^1(x-z)\partial_{z^\mu}R_0^1(z),\\
\tilde{\tau}(x)&=
\int_{\mathbb{R}^4}\mathrm{d}^4z\int_{\mathbb{R}^4}\mathrm{d}^4y\,R_0^{1}(x-z)A_0^{\phantom{1}}(z)R_0^{1}(z-y)A_0^{\phantom{1}}(y)R_0^{1}(y).
\end{align*}
Then we can write out the following asymptotic expansions with respect to the parameter $\Lambda$ for integrals arising in calculations:
\begin{fleqn}
\begin{equation}\label{ya-a-98}
\theta(0)=\frac{L+\rho_1}{8\pi^2},\,\,\,\mbox{see the definition in}\,\,\,\eqref{ya-a-76},
\end{equation}
\end{fleqn}
\begin{fleqn}
\begin{equation}\label{ya-a-99}
\tau(0)=\frac{3L\rho_5}{8\pi^2\Lambda^2}+\mathcal{O}\big(1/\Lambda^2\big),
\,\,\,\mbox{see the definition in}\,\,\,\eqref{ya-a-78},
\end{equation}
\end{fleqn}
\begin{fleqn}
	\begin{equation}\label{ya-a-87}
\hat{\mathrm{I}}_1=\int_{\mathrm{B}_{1/\sigma}}\mathrm{d}^4x\,
R_0^\Lambda(x)
\Big(A_0^{\phantom{1}}(x)\theta(x)\Big)
\theta(x)\stackrel{\mathrm{s.p.}}{=}\frac{L^2+L(1+2\rho_1+2\rho_2)}{8(4\pi^2)^2},
	\end{equation}
\end{fleqn}
\begin{fleqn}
	\begin{equation}\label{ya-a-88}
\hat{\mathrm{I}}_2=\int_{\mathrm{B}_{1/\sigma}}\mathrm{d}^4x\,
\Big(A_0^{\phantom{1}}(x)R_0^\Lambda(x)\Big)
\theta(x)\theta(x)\stackrel{\mathrm{s.p.}}{=}\frac{L^2+L(2\rho_1+2\rho_2)}{4(4\pi^2)^2},
	\end{equation}
\end{fleqn}
\begin{fleqn}
	\begin{equation}\label{ya-a-89}
\hat{\mathrm{I}}_3=\int_{\mathrm{B}_{1/\sigma}}\mathrm{d}^4x\,
R_0^\Lambda(x)
\Big(A_0^{\phantom{1}}(x)R_0^\Lambda(x)\Big)
\tau(x)\stackrel{\mathrm{s.p.}}{=}\frac{L(3\rho_3\rho_5-\rho_4)}{2(4\pi^2)^2},
	\end{equation}
\end{fleqn}
\begin{fleqn}
	\begin{equation}\label{ya-a-90}
\hat{\mathrm{I}}_4=\int_{\mathrm{B}_{1/\sigma}}\mathrm{d}^4x\,
R_0^\Lambda(x)R_0^\Lambda(x)A_0^{\phantom{1}}(x)
\tau(x)\stackrel{\mathrm{s.p.}}{=}\frac{L^2+L(1+2\rho_1)}{8(4\pi^2)^2},
	\end{equation}
\end{fleqn}
\begin{fleqn}
\begin{equation}\label{ya-a-136}
\hat{\mathrm{I}}_5=\int_{\mathrm{B}_{1/\sigma}}\mathrm{d}^4x\,
\theta(x)R_0^\Lambda(x)R_0^\Lambda(x)\stackrel{\mathrm{s.p.}}{=}\frac{L^2+L(1+2\rho_1)}{8(4\pi^2)^2},
	\end{equation}
\end{fleqn}
\begin{fleqn}
\begin{equation}\label{ya-a-137}
	\hat{\mathrm{I}}_6=\int_{\mathrm{B}_{1/\sigma}}\mathrm{d}^4x\,
	\theta(x)R_0^\Lambda(x)x^\mu\partial_{x^\mu}R_0^\Lambda(x)\times(-1/2)\stackrel{\mathrm{s.p.}}{=}
	\frac{L^2+L(1/2+2\rho_1)}{8(4\pi^2)^2},
\end{equation}
\end{fleqn}
\begin{fleqn}
	\begin{equation}\label{ya-a-137-1}
\hat{\mathrm{I}}_7=\int_{\mathrm{B}_{1/\sigma}}\mathrm{d}^4x\,
|x|^2\Big(R_0^\Lambda(x)\Big)^3\stackrel{\mathrm{s.p.}}{=}
		\frac{L}{2(4\pi^2)^2}.
	\end{equation}
\end{fleqn}
All integrals $\hat{\mathrm{I}}_i$ are analyzed by the same method -- differentiation by the regularizing parameter $\Lambda$. Only unlike the values at the beginning of the section, after applying the operator $\Lambda\partial_\Lambda$, we do not need to go to the limit, because it may not exist, but find an asymptotic expansion of $\Lambda$ while preserving a logarithmic part ($\sim L$) and a constant correction.

Let us demonstrate an example of calculating the asymptotics for $\hat{\mathrm{I}}_1$. To do this, we first write out an explicit expression 
\begin{equation}\label{ya-a-91}
\int_{\mathrm{B}_{1/\sigma}}\mathrm{d}^4x_1\int_{\mathrm{B}_{1/\sigma}}\mathrm{d}^4x_2\int_{\mathrm{B}_{1/\sigma}}\mathrm{d}^4x_3\,
R_0^\Lambda(x_1)
	\Big(A_0^{\phantom{1}}(x_1)R_0^\Lambda(x_1-x_2)\Big)R_0^\Lambda(x_2)R_0^\Lambda(x_1-x_3)R_0^\Lambda(x_3).
\end{equation}
Next, we scale the variables $x_i\to x_i/\Lambda$, apply the operator $\Lambda\partial_\Lambda$, and then rescale $x_i\to x_i\Lambda/\sigma$ them again. Since after the first transformation, all dependence on $\Lambda$ is in $\mathrm{B}_{\Lambda/\sigma}$, differentiation only affects the limit of integration over the radius. As a result, we get three contributions:
\begin{align*}
	t_1&=\int_{\mathbb{S}^3}\mathrm{d}^3\sigma(\hat{x}_1)\int_{\mathrm{B}_{1}}\mathrm{d}^4x_2\int_{\mathrm{B}_{1}}\mathrm{d}^4x_3\,
	R_0^{\Lambda/\sigma}(\hat{x}_1)
	\Big(A_0^{\phantom{1}}R_0^{\Lambda/\sigma}(\hat{x}_1-x_2)\Big)R_0^{\Lambda/\sigma}(x_2)R_0^{\Lambda/\sigma}(\hat{x}_1-x_3)R_0^{\Lambda/\sigma}(x_3),\\
	t_2&=\int_{\mathrm{B}_{1}}\mathrm{d}^4x_1\int_{\mathbb{S}^3}\mathrm{d}^3\sigma(\hat{x}_2)\int_{\mathrm{B}_{1}}\mathrm{d}^4x_3\,
	R_0^{\Lambda/\sigma}(x_1)
	\Big(A_0^{\phantom{1}}R_0^{\Lambda/\sigma}(x_1-\hat{x}_2)\Big)R_0^{\Lambda/\sigma}(\hat{x}_2)R_0^{\Lambda/\sigma}(x_1-x_3)R_0^{\Lambda/\sigma}(x_3),\\
	t_3&=\int_{\mathrm{B}_{1}}\mathrm{d}^4x_1\int_{\mathrm{B}_{1}}\mathrm{d}^4x_2\int_{\mathbb{S}^3}\mathrm{d}^3\sigma(\hat{x}_3)\,
	R_0^{\Lambda/\sigma}(x_1)
	\Big(A_0^{\phantom{1}}R_0^{\Lambda/\sigma}(x_1-x_2)\Big)R_0^{\Lambda/\sigma}(x_2)R_0^{\Lambda/\sigma}(x_1-\hat{x}_3)R_0^{\Lambda/\sigma}(\hat{x}_3).
\end{align*}
Let us consider the main transitions for each term separately.\\

\noindent1) Recall that $\hat{x}=x/|x|$, therefore $R_0^{\Lambda/\sigma}(\hat{x}_1)=1/(4\pi^2)$, since $R_0^{\Lambda/\sigma}(\hat{x}_1)=R_0^{\phantom{1}}(\hat{x}_1)$, if $\Lambda/\sigma>1$ is satisfied. Next, note that the support of the function $A_0^{\phantom{1}}R_0^{\Lambda/\sigma}(\cdot)$ is located in $\mathrm{B}_{\sigma/\Lambda}$, so the variable $x_2$ is located in a ball of radius $\sigma/\Lambda$ centered at $\hat{x}_1$. When searching for the main order, this fact allows us to replace $R_0^{\Lambda/\sigma}(x_2)$ with $1/(4\pi^2)$, and then make a substitution
\begin{equation*}
	\int_{\mathrm{B}_{1}}\mathrm{d}^4x_2\,A_0^{\phantom{1}}R_0^{\Lambda/\sigma}(\hat{x}_1-x_2)\to\frac{1}{2},
\end{equation*}
	where the multiplier $1/2$ appeared due to the partial overlap of the support of the function and the area of integration. The final answer is obtained after the limit transition and explicit calculation of the integral over the variable $x_3$
\begin{equation*}
	\lim_{\Lambda\to+\infty}t_1=\frac{1}{2(4\pi^2)^2}
	\int_{\mathbb{S}^3}\mathrm{d}^3\sigma(\hat{x}_1)\int_{\mathrm{B}_{1}}\mathrm{d}^4x_3\,R_0(\hat{x}_1-x_3)R_0(x_3)=\frac{1}{16(4\pi^2)^2},
\end{equation*}
	where the auxiliary relation was used
\begin{equation}\label{ya-a-97}
\int_{\mathrm{B}_{1}}\mathrm{d}^4z\,R_0(x-z)R_0(z)=\frac{1}{(4\pi^2)^2}
\begin{cases}
1-2\ln|x|, & |x|\leqslant1;\\
	\,\,\,\,\,\,\,\,\,\,\,\,1,& |x|>1.
\end{cases}
\end{equation}
\noindent2) The second part is analyzed in a similar way. First, we replace $R_0^{\Lambda/\sigma}(\hat{x}_2)\to1/(4\pi^2)$, then $R_0^{\Lambda/\sigma}(x_1)\to1/(4\pi^2)$, as well as the integral with density $A_0^{\phantom{1}}R_0^{\Lambda/\sigma}(x_1-\hat{x}_2)$ over "half" of the ball leads to the value $1/2$. The remaining integral completely reproduces the one that was obtained in the first case, therefore
\begin{equation*}
\lim_{\Lambda\to+\infty}t_2=
\lim_{\Lambda\to+\infty}t_1=\frac{1}{16(4\pi^2)^2}.
\end{equation*}
\noindent3) The calculation sequence for the third function is different. To begin with, an obvious replacement is made as follows $R_0^{\Lambda/\sigma}(\hat{x}_3)\to1/(4\pi^2)$. Then it should be noted that one of the functions can be replaced by an undeformed analog and the corresponding integral over the sphere can be explicitly calculated
\begin{equation*}
\int_{\mathbb{S}^3}\mathrm{d}^3\sigma(\hat{x}_3)\,R_0^{\Lambda/\sigma}(x_1-\hat{x}_3)\to
\int_{\mathbb{S}^3}\mathrm{d}^3\sigma(\hat{x}_3)\,R_0^{\phantom{1}}(x_1-\hat{x}_3)=\frac{2\pi^2}{4\pi^2}
\begin{cases}
	\,\,\,\,\,1, & |x_1|\leqslant1;\\
	|x_1|^{-2},& |x_1|>1,
\end{cases}\to\frac{1}{2},
\end{equation*}
where $x_1\in\mathrm{B}_1$ was used. After all the substitutions in the remaining integral
\begin{equation*}
\frac{1}{8\pi^2}\int_{\mathrm{B}_{1}}\mathrm{d}^4x_1\int_{\mathrm{B}_{1}}\mathrm{d}^4x_2\,
R_0^{\Lambda/\sigma}(x_1)
\Big(A_0^{\phantom{1}}R_0^{\Lambda/\sigma}(x_1-x_2)\Big)R_0^{\Lambda/\sigma}(x_2)
\end{equation*}
it is possible to make the transition
\begin{equation*}
\Big(A_0^{\phantom{1}}R_0^{\Lambda/\sigma}(x_1-x_2)\Big)R_0^{\Lambda/\sigma}(x_2)
\to\Big(A_0^{\phantom{1}}R_0^{\Lambda/\sigma}(x_2)\Big)R_0^{\Lambda/\sigma}(x_1-x_2),
\end{equation*}
and also add and subtract the density of $\big(R_0^{\Lambda/\sigma}(x_1)\big)^2$. Then we come to the expression
\begin{equation*}
\frac{1}{8\pi^2}\int_{\mathrm{B}_{\Lambda/\sigma}}\mathrm{d}^4x_1\int_{\mathrm{B}_{\Lambda/\sigma}}\mathrm{d}^4x_2\,
\Big(R_0^{1}(x_1)R_0^{1}(x_1-x_2)-R_0^{1}(x_1)R_0^{1}(x_1)\Big)A_0^{\phantom{1}}R_0^{1}(x_2)+
\frac{1}{8\pi^2}\int_{\mathrm{B}_{\Lambda/\sigma}}\mathrm{d}^4x_1\,
\Big(R_0^{1}(x_1)\Big)^2.
\end{equation*}
Finally, using the notation suggested above and passing to the limit in the first integral, we obtain 
\begin{equation*}
\lim_{\Lambda\to+\infty}t_3=
	\frac{2L+1+2\rho_1+2\rho_2}{8(4\pi^2)^2}+o(1).
\end{equation*}
Summing up all the terms and calculating the integral, we arrive at the stated answer. The integrals $\hat{\mathrm{I}}_2$ and $\hat{\mathrm{I}}_4$ are analyzed using exactly the same method. Before proceeding to the comments for $\hat{\mathrm{I}}_3$, we note that the result for $\theta(0)$ is obtained by direct calculation, while for $\tau(0)$ it is obtained by rejecting corrections of the order of $1/\Lambda^2$. It is convenient to make a representation of all deformed functions in the form of $R_0^\Lambda(\cdot)=R_0^\Lambda(\cdot)\pm R_0^{\phantom{1}}(\cdot)$. Then the main part of the expansion, containing only undeformed functions, will be reduced, and the next correction, proportional to $L/\Lambda^2$, will be contained only in terms consisting of two undeformed functions and the difference. There are only three such contributions. They lead to the stated answer after using relation \eqref{ya-a-97}. For reference, here is the value of the subtracted part
\begin{equation*}
\int_{\mathrm{B}_{1/\sigma}}\mathrm{d}^4z\int_{\mathrm{B}_{1/\sigma}}\mathrm{d}^4y\,
R_0^{\phantom{1}}(z)R_0^{\phantom{1}}(z-y)R_0^{\phantom{1}}(y)=\frac{1}{32\pi^2\sigma^2}.
\end{equation*}
In the case of the integral $\hat{\mathrm{I}}_3$, it is easier to add and subtract density at zero, that is, $\tau(x)=\tau(x)\pm\tau(0)$. Then the summand with $\tau(0)$ will result in a part with $\rho_5$, and the term with the difference after using the auxiliary relation
\begin{equation*}
\int_{\mathrm{B}_{2}}\mathrm{d}^4z\,\Big(R_0^{1}(x-z)-R_0^{1}(z)\Big)=-\frac{|x|^2}{8},
\end{equation*}
which is true for all $x\in\mathrm{B}_1$, will result in the part with $\rho_4$. Note that the last equality is checked using the operator $A_0(x)$, taking into account the equality to zero at $x=0$.

\section{Two loops: weak case}
\label{ym:sec:two1}
\subsection{Decomposition of Green's functions}
\label{ym:sec:two1-1}
The basic idea of calculating two-loop contributions is the same as the one in Section \ref{ym:sec:two}. Thus, the analysis is related to the decomposition of the Green's function near the diagonal. Since the operator has a similar appearance in general, we will focus only on significant differences. First of all, we note that the form of the representations \eqref{ya-a-74} and \eqref{ya-a-128} remains valid for the new case, since the decomposition is obtained in a pertubative manner, see for example \eqref{ya-a-38} and \eqref{ya-a-39}. Nevertheless, the definitions for each part change individually, since the Green's function for a deformed scalar operator has a new form. In this section, the same designations will be used for the components, as this does not create confusion. So, the deformed Green's function for the ghost fields can be divided into two parts
\begin{equation}\label{ya-z-10}
G_0^{\Lambda ab}(x,y)=G_{\mathrm{loc}}^{ab}(x,y)+PS_0^{ab}(x,y),
\end{equation}
where the second part has two finite derivatives on the diagonal and is essentially determined by the last relation. In this case, $G_{\mathrm{loc}}^{ab}(x,y)$ differs significantly from \eqref{ya-a-74-1} and can be written out explicitly. To do this, first consider the simplified case in which the background field is $B_\mu^a(x)\to\hat{B}_\mu^a(x)$ and depends linearly on the coordinate as follows
\begin{equation}\label{ya-z-6}
	B^a_{\mu}(x)=\frac{s}{2}x^\nu \xi^a_{\nu\mu},\,\,\,\mbox{where}\,\,\,
	\big(\xi^a\big)_{\mu\nu}=\frac{1}{\sqrt{8\dim\mathfrak{g}}}
	\begin{pmatrix}
		0&1&0&1\\
		-1&0&1&0\\
		0&-1&0&1\\
		-1&0&-1&0
	\end{pmatrix}\,\,\,
	\mbox{for all}\,\,\,a\in\{1,\ldots,\dim\mathfrak{g}\}.
\end{equation}
This example was used to construct a covariant regularization in \cite{Ivanov-2022,Ivanov-Kharuk-20222}. Here, the parameter $s>0$ is small for the convenience of formal series expansion. Then, denoting with the symbol $\hat{M}^{ab}_0(x)$ the scalar operator \eqref{ya-a-10} after substitution of $B_\mu^a(x)\to\hat{B}_\mu^a(x)$ and using formula (84) from \cite{Ivanov-2022}, we can write out the following decomposition near the diagonal
\begin{multline}\label{ya-z-3}
\rho\bigg(\sqrt{\hat{M}_0}/\Lambda\bigg)^{ab}G_0^{bc}(x,y)=\Phi^{ab}_{\hat{B}}(x,y)\bigg(
\delta^{bc}R_0^\Lambda\big|_{\mathbf{f}=0}(x-y)\\+
s^2\xi_{\sigma\beta}^{bd}\xi_{\sigma\beta}^{dc}
\bigg[-\frac{|x-y|^2}{2^9\pi^2}-\frac{\hat{R}_0^\Lambda(x-y)}{2^6\Lambda^2}\bigg]+\mathcal{O}(s^3)
\bigg).
\end{multline}
Here we used relation \eqref{ya-a-37}, the definition of the function $\rho(r)=2J_1(r)/r$ from formula (50) in \cite{Ivanov-2022}, where $J_1(\cdot)$ is the Bessel function of the first kind, the definition of the function
\begin{equation}\label{ya-z-4}
	\hat{R}_0^\Lambda(x)=\frac{1}{4\pi^2}\begin{cases}
		\,\,\,\,\,\,\,\,\,\,\,\,\,\,\,\,\,\,\,\,\,\,\,\,\,\,
		0,&|x|>1/\Lambda;\\
		\frac{1}{6}-\frac{1}{3}|x|^2\Lambda^2+\frac{1}{6}|x|^4\Lambda^4,&|x|\leqslant1/\Lambda,
	\end{cases}
\end{equation}
from formula (77) of \cite{Ivanov-2022}, as well as the definition for the ordered exponent
\begin{equation}\label{ya-z-7}
\Phi^{ab}_{B}(x,y)=1+\sum_{k=1}^{+\infty}(-1)^k
\int_0^1\mathrm{d}z^{\mu_1}(s_1)\ldots
\int_0^{s_{k-1}}\mathrm{d}z^{\mu_k}(s_k)\,
B_{\mu_1}^{ac_1}\big(z(s_1)\big)\cdot\ldots\cdot
B_{\mu_k}^{c_{k-1}b}\big(z(s_k)\big),
\end{equation}
where $z^\mu(s)=(1-s)y^\mu+sx^\mu$ is a parametrization of a geodesic in the flat space (straight line), see \cite{sig1,33}. Note that the correction term is $o(s^2)$ for small values of the parameter. Let us adapt the representation from \eqref{ya-z-3} to the case under consideration.

First, the function $\hat{\omega}(\cdot)=\hat{\Omega}^2(\cdot)$ was used to construct the weak deformation, see formula \eqref{ya-z-15}, instead of the special case of $\rho(\cdot)$, which in \cite{Ivanov-2022} referred to single averaging over a unit sphere. Thus, it is necessary to make the transition
\begin{equation}\label{ya-z-8}
\rho(r)=\frac{2J_1(r)}{r}=\frac{1}{S_{3}}\int_{\mathrm{S}^{3}}\mathrm{d}^4k\,e^{i(k,x)}
\longrightarrow
\int_{\mathrm{B}_{1}}\mathrm{d}^4z\,e^{i(z,x)}
\bigg(\int_{\mathrm{B}_{1/2}}\mathrm{d}^4y\,\omega(|z-y|)\omega(|y|)\bigg)=\hat{\omega}(r),
\end{equation}
where the notation $|x|=r$ was used. Next, note that the function in parentheses depends on the absolute value of $|z|=t$. Therefore, moving to spherical coordinates with respect to the variable $z$, we obtain
\begin{equation}\label{ya-z-9}
\hat{\omega}(r)=S_3\int_0^1\mathrm{d}t\,t^3\upsilon_\omega(t)\rho(tr)
,\,\,\,\mbox{where}\,\,\,
\upsilon_\omega(t)=\int_{\mathrm{B}_{1/2}}\mathrm{d}^4y\,\omega(|z-y|)\omega(|y|).
\end{equation}
Thus, the case of interest can be obtained by applying the integral operator to both sides of formula \eqref{ya-z-3} after replacing $1/\Lambda\to t/\Lambda$. In particular, the formulas are valid
\begin{equation*}
S_3\int_0^1\mathrm{d}t\,t^3\upsilon_\omega(t)=1,\,\,\,
S_3\int_0^1\mathrm{d}t\,t^3\upsilon_\omega(t)
R_0^{\Lambda/t}\big|_{\mathbf{f}=0}(x-y)=R_0^{\Lambda}(x-y),
\end{equation*}
where the last function corresponds to the twice-averaged fundamental solution \eqref{ya-a-36} for the free Laplace operator. For the sake of definiteness, such a transformation of the function $\hat{R}_0^\Lambda(\cdot)/\Lambda^2$ will be denoted by $\hat{R}_{0,\omega}^\Lambda(\cdot)$. Let us note that $\mathrm{supp}\big(\hat{R}_{0,\omega}^\Lambda(\cdot)\big)\in\mathrm{B}_{1/\Lambda}$, because
\begin{equation*}
\mbox{if}\,\,\,\hat{R}_{0}^\Lambda(x)=0,\,\,\,\mbox{then}\,\,\,\hat{R}_{0}^{\Lambda/t}(x)=0\,\,\,
\mbox{for all}\,\,\,t\in[0,1].
\end{equation*}

Secondly, it is necessary to explain the applicability of the special case of the field, since the final formula \eqref{ya-z-3} contains a rather special combination of the field strength tensors. Indeed, using invariance with respect to the gauge transformations of the background field and considerations related to dimensional conservation, it can be noted that all singular components in the weak case should be proportional to the classical action. Thus, the quadratic combination of the field strength tensors can be written out in any convenient form, including the averaged one. Therefore, the function $G_{\mathrm{loc}}^{ab}(x,y)$ from \eqref{ya-z-10} for an arbitrary smooth background field $B_\mu^a$ is chosen as
\begin{equation}\label{ya-z-5}
G_{\mathrm{loc}}^{ab}(x,y)=\Phi^{ab}_{B}(x,y)
R_0^\Lambda(x-y)-\frac{1}{2^7}
\Big(F_{\sigma\beta}^{ad}(x)F_{\sigma\beta}^{db}(x)+F_{\sigma\beta}^{ad}(y)F_{\sigma\beta}^{db}(y)\Big)
\bigg(\hat{R}_{0,\omega}^\Lambda(x-y)+
\frac{|x-y|^2}{8\pi^2}
\bigg)
.
\end{equation}
At the same time, the part $PS_0^{ab}(x,y)$ may contain inconsistencies, which, after averaging over indices and arguments, reduce to zeros and, thus, do not affect the final results.

Turning to the Green's vector function, we also use the notation from \eqref{ya-a-74}, that is, we have again
\begin{equation}\label{ya-z-11}
	G_{1\mu\nu}^{\Lambda ab}(x,y)=\delta_{\mu\nu}G_{0}^{ab}(x,y)+
	\mathcal{N}_{\mu\nu}^{ab}(x,y)+
	\mathcal{L}_{\mu\nu}^{ab}(x,y)+
	PS_{1\mu\nu}^{\,\,ab}(x,y).
\end{equation}
We need to provide supporting comments. Here, the first part of $G_{0}^{\Lambda ab}(x,y)$ coincides with the new scalar deformed Green's function from \eqref{ya-z-10}. The second part of $\mathcal{N}_{\mu\nu}^{ab}(x,y)$ is not equal to the one in Section \ref{ym:sec:two-1}. At the same time, it exactly matches in the main order, and differs in the rest. Moreover, their explicit form is not important, since it does not lead to a singular contribution. The third part of $\mathcal{L}_{\mu\nu}^{ab}(x,y)$, as before, is determined by formula \eqref{ya-a-77}. The last part is defined as a residual term and does not match the one in \eqref{ya-a-74}. Note that all parts are symmetric with respect to the simultaneous permutation of indices and arguments.

\subsection{Adapting calculations}
\label{ym:sec:two1-2}
In this section, we discuss the issue related to the use of existing calculations for the "strong" case, see Section \ref{ym:sec:two}. It turns out that many of the answers look exactly the same. Let us start with the parts containing the so-called nonlocal parts of the Green's functions. In the "strong" case, Section \ref{ym:sec:two-nl} was dedicated to them. Note that the calculations did not use an explicit form for $PS_0$ and $PS_1$, only their smoothness properties. In addition, we note that the logarithmic singularity occurred only to the first degree and did not depend on the regularizing function $\mathbf{f}(\cdot)$. Moreover, repeating the calculations, we can see that when working with asymptotics, it was allowed to ignore some unimportant\footnote{Small in a certain sense. As a rule, by the regularizing parameter $\Lambda$.} terms, which eventually led to the possibility of using invariance with respect to the gauge transformations of the background field. Summing up all the observations, it can be argued that the logarithmic nonlocal parts in the "weak" case have the same form. Next, we note that the Green's functions in the main order in the "strong" and "weak" cases coincide and are proportional to \eqref{ya-a-37}. Consequently, the power-law singularities in both approaches also coincide. Making a summary table, we get the following results:
\begin{align*}
\mbox{\textbf{Nonlocal singular parts:}}\,\,\,\mathbb{H}_0^{\mathrm{sc}}(\Gamma_3^2)&\longrightarrow\eqref{ya-a-121}+
\eqref{ya-a-129}+\eqref{ya-a-122}
;\\
\mathbb{H}_0^{\mathrm{sc}}(\Omega_3^2)&\longrightarrow
\eqref{ya-a-130}+\eqref{ya-a-131}
;\\
\mathbb{H}_0^{\mathrm{sc}}(\Gamma_4)&\longrightarrow\eqref{ya-a-132}+\eqref{ya-a-133}+\eqref{ya-a-134}
.
\end{align*}

Let us move on to the local parts. There will already be discrepancies here. Let us first note the terms that remain unchanged. These include those that were calculated using either part of the Green's function $\mathcal{L}_{\mu\nu}^{ab}(x,y)$, or the main order of the component $\mathcal{N}_{\mu\nu}^{ab}(x,y)$, or the answer was basically written out in the general \footnote{This implies functionals of the type $\mathrm{J}_{\odot}$, $\mathrm{J}_{\otimes}$, and $\mathrm{J}_{\ominus}$.} terms. Using the proposed criteria, we provide links to the results in the form of a summary table:

\begin{align*}
	\mbox{\textbf{Local parts with the same answer:}}\,\,\,\mathbb{H}_0^{\mathrm{sc}}(\Gamma_3^2)&\longrightarrow\eqref{ya-a-126}+
	\eqref{ya-a-140}+\eqref{ya-a-143}
	;\\
	\mathbb{H}_0^{\mathrm{sc}}(\Omega_3^2)&\longrightarrow
	\eqref{ya-a-146}+\eqref{ya-a-147}
	;\\
	\mathbb{H}_0^{\mathrm{sc}}(\Gamma_4)&\longrightarrow\eqref{ya-a-149}+\eqref{ya-a-148}+\eqref{ya-a-150}
	.
\end{align*}
Thus, the calculation of singular contributions for the main three diagrams is reduced to the adaptation of two terms:
\begin{align*}
	\mbox{\textbf{Local parts with a different answer:}}\,\,\,\mathbb{H}_0^{\mathrm{sc}}(\Gamma_3^2)&\longrightarrow\eqref{ya-a-139}
	;\phantom{+
		\eqref{ya-a-140}+\eqref{ya-a-143}\,\,\,}\\
	\mathbb{H}_0^{\mathrm{sc}}(\Omega_3^2)&\longrightarrow
	\eqref{ya-a-145}
	.
\end{align*}
Here are the main arguments and results. First, note that both contributions are linear combinations of $\mathrm{J}_{\sigma\rho}^{\phantom{1}}[\mathcal{N}_{1\sigma\rho}]$ and $\mathrm{\hat{J}}_{\sigma\rho}^{\phantom{1}}[\mathcal{N}_{1\sigma\rho}]$, see formulas \eqref{ya-a-138} and \eqref{ya-a-144}. They are the ones that should be studied. Next, we recall that the ordered exponential has a set of important properties, see for example \cite{33}. In our case, only the differentiation formulas in the first orders are noteworthy
\begin{align*}
D_{\sigma}^{ca}(x)\Phi^{ab}(x,y)&=\frac{1}{2}\Phi^{ca}(x,y)(x-y)^\mu F_{\mu\sigma}^{ab}(y)+\ldots,\\
D_{\rho}^{ba}(y)\Phi^{ca}(x,y)&=\frac{1}{2}\Phi^{ca}(x,y)(x-y)^\mu F_{\mu\rho}^{ab}(y)+\ldots,
\end{align*}
where ellipsis denote terms containing full covariant derivatives of the field strength tensor. Based on these relations, it is possible to prove an auxiliary statement. Let $g(\cdot)$ be a scalar smooth function, then the equality holds
\begin{multline}\label{ya-z-2}
D_{\sigma}^{ca}(x)D_{\rho}^{bd}(y)\Big(\Phi^{ad}(x,y)g(x-y)\Big)=\Phi^{ac}(x,y)
\bigg(-\delta^{cb}\partial_{x^\sigma}\partial_{x^\rho}+\frac{1}{2}F_{\sigma\rho}^{cb}(y)\\+\frac{(x-y)^\mu}{2}\Big[ F_{\mu\rho}^{cb}(y)\partial_{x^\sigma}-F_{\mu\sigma}^{cb}(y)\partial_{x^\rho}\Big]+
\frac{1}{4}(x-y)^{\mu\nu}F_{\mu\rho}^{cd}(y)F_{\nu\sigma}^{db}(y)+\ldots\bigg)g(x-y),
\end{multline}
where again, the ellipsis denotes combinations with covariant derivatives and with large number of monomials $x-y$. This formula is quite remarkable, because after replacing the function $\Phi^{ad}(x,y)g(x-y)$ with $G_{\mathrm{loc}}^{ad}(x,y)$ from \eqref{ya-z-5}, it can be seen that the double derivative does not lead to a singular contribution due to the antisymmetry of $\mathcal{N}_{1\sigma\rho}$ by subscripts in $\mathrm{J}_{\sigma\rho}^{\phantom{1}}[\mathcal{N}_{1\sigma\rho}]$ and $\mathrm{\hat{J}}_{\sigma\rho}^{\phantom{1}}[\mathcal{N}_{1\sigma\rho}]$. Therefore, $\mathcal{N}_{1\sigma\rho}$ can be replaced by the main order, while in \eqref{ya-z-2} it is possible to leave only two terms proportional to the field strength tensor in the first degree. By quoting similar terms and using auxiliary integrals from Section \ref{ym:sec:two-vs}, we arrive at the following answers
\begin{align}\label{ya-y-6}
\mathrm{J}_{\sigma\rho}^{\phantom{1}}[\mathcal{N}_{1\sigma\rho}]&
\stackrel{\mathrm{s.p.}}{=}\frac{c_2^2}{2}W_{-1}(\hat{\mathrm{I}}_6-\hat{\mathrm{I}}_5)
,\\\label{ya-y-7}
\mathrm{\hat{J}}_{\sigma\rho}^{\phantom{1}}[\mathcal{N}_{1\sigma\rho}]&
\stackrel{\mathrm{s.p.}}{=}c_2^2W_{-1}(2\hat{\mathrm{I}}_6-\hat{\mathrm{I}}_5)
.
\end{align}
Therefore, the formulas for the "strong" case, taking into account \eqref{ya-a-138} and \eqref{ya-a-144}, are transformed according to the following substitutions:
\begin{align*}
\eqref{ya-a-139}&\longrightarrow
	\mathbb{H}_0^{\mathrm{sc}}\big(\Gamma_3^2\big)\Big|_{G_{\mathrm{loc}}G_{\mathrm{loc}}\mathcal{N}}\stackrel{\mathrm{s.p.}}{=}-\frac{3c_2^2}{2}W_{-1}(3\hat{\mathrm{I}}_6-2\hat{\mathrm{I}}_5)\stackrel{\mathrm{s.p.}}{=}
	-\frac{3c_2^2W_{-1}}{16(4\pi^2)^2}\Big(L^2+L(-1/2+2\rho_1)\Big)
	,\\
\eqref{ya-a-145}&\longrightarrow
	\mathbb{H}_0^{\mathrm{sc}}\big(\Omega_3^2\big)\Big|_{G_{\mathrm{loc}}G_{\mathrm{loc}}\mathcal{N}}\stackrel{\mathrm{s.p.}}{=}-\frac{
	c_2^2}{2}W_{-1}\hat{\mathrm{I}}_6\stackrel{\mathrm{s.p.}}{=}
	-\frac{c_2^2W_{-1}}{16(4\pi^2)^2}\Big(L^2+L(1/2+2\rho_1)\Big)
	.
\end{align*}
Additionally, we draw attention to the fact that both answers are gauge invariant, in contrast to the "strong" case, in which the correction singularities depend on the degree of the background field. Thus, the last formulas conclude the discussion of contributions for all strongly connected two-loop diagrams.

The next step is to consider the relations for the counter-diagrams from Section \ref{ym:sec:two-co}. Repeating the main stages of the reasoning given above, and also paying attention to the fact that in the "weak" case, the entire averaged trace part, proportional to the second degree of the field strength tensor, is contained precisely in $G_{\mathrm{loc}}^{ad}(x,y)$, see formulas \eqref{ya-z-3} and \eqref{ya-z-5}, we obtain the validity of all available relations
\begin{equation*}
	\mbox{\textbf{Equalities}}\,\,\eqref{ya-a-158},\,
	\eqref{ya-a-159},\,\eqref{ya-a-161},\,\mbox{\textbf{and}}\,\,\eqref{ya-a-162}\,\,
	\mbox{\textbf{have no changes}}
	.
\end{equation*}

The last point of discussion is Section \ref{ym:sec:two-llo}. In this case, all three functions $\mathrm{J}_{\otimes}$, $\mathrm{J}_{\odot}$, and $\mathrm{J}_{\ominus}$ can be calculated explicitly using formula \eqref{ya-z-5} for $G_{\mathrm{loc}}^{ad}(x,y)$. Using the definitions for the functionals \eqref{ya-a-141}, \eqref{ya-a-150}, and \eqref{ya-a-187}, after direct substitution we obtain
\begin{equation}\label{ya-p-6}
\mathrm{J}_{\otimes}=-\frac{c_2W_{-1}}{48(4\pi^2)}+o(1)
,\,\,\,
\mathrm{J}_{\odot}\stackrel{\mathrm{s.p.}}{=}\frac{c_2W_{-1}}{2^6}\hat{R}_{0,\omega}^\Lambda(0)=
\mathcal{O}(1/\Lambda^2)
,\,\,\,
\mathrm{J}_{\ominus}\stackrel{\mathrm{s.p.}}{=}\frac{c_2^2W_{-1}L}{16(4\pi^2)^2}+\mathcal{O}(1)
,
\end{equation}
where the correction terms are given relative to the regularizing parameter $\Lambda$. Note that when deriving the last relation, property \eqref{ya-z-2} was used for the ordered exponential. As a result, summing up all the above facts, the analogue of formula \eqref{ya-a-166} for weak deformation is written out in the following form
\begin{align}\label{ya-z-12}
	-\frac{1}{2}\mathbb{H}_0^{\mathrm{sc}}\big(\Gamma_3^2\big)
	+\frac{1}{4}\mathbb{H}_0^{\mathrm{sc}}\big(\Gamma_4^{\phantom{1}}\big)-\frac{1}{2}\mathbb{H}_0^{\mathrm{sc}}\big(\Omega_3^2\big)\stackrel{\mathrm{s.p.}}{=}&
	-\frac{\Lambda^2c_2(2\rho_3-3\rho_0)}{8\pi^2}
	\int_{\mathbb{R}^4}\mathrm{d}^4x\,\Big(PS_{1\rho\rho}^{\,\,aa}(x,x)+4PS_{0}^{aa}(x,x)\Big)\\\nonumber
	&-\frac{5Lc_2}{48\pi^2}
	\int_{\mathbb{R}^4}\mathrm{d}^4x\,
	\Big(D_\rho^{bc}(x)D_\sigma^{ca}(x)PS_{1\sigma\rho}^{\,\,ab}(x,y)\Big)\Big|_{y=x}
	\\\nonumber
	&-\frac{c_2^2W_{-1}L}{32(4\pi^2)^2}\Big(10+12\rho_2+96\rho_5(2\rho_3-3\rho_0)-64\rho_4\Big).
\end{align}
Next, using relations \eqref{ya-a-158} and \eqref{ya-a-159}, we obtain the final relation in the form
\begin{align}\label{ya-z-13}
	-\frac{1}{2}\mathbb{H}_0^{\mathrm{sc}}\big(\Gamma_3^2\big)
	+\frac{1}{4}\mathbb{H}_0^{\mathrm{sc}}\big(\Gamma_4^{\phantom{1}}\big)-\frac{1}{2}\mathbb{H}_0^{\mathrm{sc}}\big(\Omega_3^2\big)\stackrel{\mathrm{s.p.}}{=}&
	-\frac{\Lambda^2c_2(2\rho_3-3\rho_0)}{8\pi^2}
	\mathbb{H}_0^{\mathrm{sc}}\big(S_2[\,\cdot\,]\big)\\\nonumber
	&-\frac{5Lc_2}{48\pi^2}\mathbb{H}_0^{\mathrm{sc}}\big(S_{\mathrm{f}}[\,\cdot\,,B]\big)
	\\\nonumber
	&-\frac{c_2^2W_{-1}L}{32(4\pi^2)^2}\Big(175/18+56\rho_2/3-64\rho_4\Big).
\end{align}
It is clear that the first two terms are reduced by counter-diagrams, and the third part leads to the two-loop coefficient for the renormalization constant.

\subsection{Answers for diagrams}
\label{ym:sec:two1-3}
\begin{table}[h!]
	\centering
	\setlength\arrayrulewidth{0.6pt}
	\renewcommand{\arraystretch}{1.8}
	\begin{tabular}{|r||c|c|c|c|c|c|}
		\hline
		\multicolumn{1}{|c||}{\mbox{Diagram}}&$\frac{c_2^2W_{-1}}{(4\pi)^4}$&$\frac{\Lambda^2c_2\mathbb{H}_0^{\mathrm{sc}}(S_2)}{(4\pi)^2}$&$\frac{Lc_2\mathrm{V}_1}{48\pi^2}$&$\frac{Lc_2\mathrm{V}_2}{48\pi^2}$&$\frac{Lc_2\mathrm{V}_3}{48\pi^2}$&$\frac{Lc_2\mathrm{V}_4}{48\pi^2}$\\
		\hline
		${\centering\adjincludegraphics[width = 1 cm, valign=c]{fig/ya8.eps}}$&$L(-1+4\rho_2-8\rho_4)$&$\rho_3$&$-5$&$-1/2$&$0$&$0$\\
		\hline
		$-{\centering\adjincludegraphics[width = 1 cm, valign=c]{fig/ya7.eps}}$&$12L^2+L(10+24\rho_1+8\rho_2-64\rho_4)$&$8\rho_3$&$-8$&$4$&$18$&$0$\\
		\hline
		${\centering\adjincludegraphics[width = 1 cm, valign=c]{fig/ya9.eps}}$&$8L^2+L(8+16\rho_1-32\rho_4)$&$4\rho_3$&$4$&$4$&$12$&$0$\\
		\hline
		$-2{\centering\adjincludegraphics[width = 1 cm, valign=c]{fig/ya10.eps}}$&$-4L^2+L(-4-8\rho_1+16\rho_4)$&$-2\rho_3$&$10$&$1$&$-6$&$0$\\
		\hline
		${\centering\adjincludegraphics[width = 1 cm, valign=c]{fig/ya11.eps}}$&$-4L^2+L(-4-8\rho_1+16\rho_4)$&$-2\rho_3$&$10$&$1$&$-6$&$0$\\
		\hline
\rowcolor[gray]{.9}			$\mathbb{H}_0^{\mathrm{sc}}\big(\Gamma_3^2\big)$&$12L^2+L(9+24\rho_1+12\rho_2-72\rho_4)$&$9\rho_3$&$11$&$19/2$&$18$&$0$\\
		\hline
\rowcolor[gray]{.9}			$\mathbb{H}_0^{\mathrm{sc}}\big(\Omega_3^2\big)$&$L(1+8\rho_4)$&$-\rho_3$&$-1$&$1/2$&$0$&$-6$\\
		\hline
\rowcolor[gray]{.9}		$-\frac{1}{2}\mathbb{H}_0^{\mathrm{sc}}\big(\Gamma_4\big)$&$-12L^2+L(-24\rho_1)$&$-12\rho_0$&$0$&$0$&$-18$&$0$\\
		\hline
\rowcolor[gray]{.7}		\multicolumn{1}{|c||}{$-\frac{\mbox{\small{sum}}}{2}$}&$L(-5-6\rho_2+32\rho_4)$&$6\rho_0-4\rho_3$&$-5$&$-5$&$0$&$3$\\
		\hline
	\end{tabular}
	\renewcommand{\arraystretch}{1}
	\caption{The table shows the singular components for individual diagrams, as well as for their linear combinations. In the sixth line, the result for expression \eqref{ya-a-19} is written out, which is equal to the sum of the first five diagrams (lines). The last line contains a linear combination for \eqref{ya-z-13}.}
	\label{ya:table:1}
\end{table}

\noindent The main result of the previous section was the derivation of a formula for the singular component of the linear combination of diagrams from \eqref{ya-z-13}, both as a whole and separately for each component. By performing similar steps, we can get answers for each diagram. To do this, it is enough to adapt the calculations, as was done in Section \ref{ym:sec:two1-3}, taking into account the auxiliary calculations from Sections \ref{ym:sec:two-nl} and \ref{ym:sec:two-1} for individual diagrams. Using the above notation and auxiliary functionals of the form
\begin{align*}
\mathrm{V}_1&=\int_{\mathbb{R}^4}\mathrm{d}^4x\,
\Big(D_\rho^{ab}(x)D_\sigma^{bc}(x)\big(PS_{1\sigma\rho}^{ca}(x,y)+\delta_{\sigma\rho}
PS_{0}^{ca}(x,y)\big)\Big)\Big|_{y=x},\\
\mathrm{V}_2&=\int_{\mathbb{R}^4}\mathrm{d}^4x\,
\Big(M_{1\rho\sigma}^{ac}(x)\big(PS_{1\sigma\rho}^{ca}(x,y)+\delta_{\sigma\rho}
PS_{0}^{ca}(x,y)\big)\Big)\Big|_{y=x},\\
\mathrm{V}_3&=\int_{\mathbb{R}^4}\mathrm{d}^4x\,
\Big(F_{\rho\sigma}^{ac}(x)\big(PS_{1\sigma\rho}^{ca}(x,x)+\delta_{\sigma\rho}
PS_{0}^{ca}(x,x)\big)\Big),\\
\mathrm{V}_4&=\int_{\mathbb{R}^4}\mathrm{d}^4x\,
\Big(M_{0}^{ac}(x)PS_{0}^{ca}(x,y)\Big)\Big|_{y=x},
\end{align*}
we provide the answer in the form of Table \ref{ya:table:1}.

\subsection{Comparison of parts $G_{\mathrm{loc}}G_{\mathrm{loc}}\mathcal{N}$ }
\label{ym:sec:two1-5}

In Section \ref{ym:sec:two1-2}, a variant of converting formulas for singular parts in the case of strong deformation to formulas for the "weak" approach was shown. It turned out that most of the "singular" integrals have the same form and do not depend on the choice of deformation. The main difference lies in local contributions of the type $G_{\mathrm{loc}}G_{\mathrm{loc}}\mathcal{N}$. To see this, just compare the formulas \eqref{ya-a-151}$\longleftrightarrow$\eqref{ya-y-6} and \eqref{ya-a-152}$\longleftrightarrow$\eqref{ya-y-7}. 

Let us discuss the first pair. It is clear that in the "strong" case, the classical action is split into two parts, therefore, it is more correct to make a comparison for the second part with the derivative, since in this case there are fewer deforming averaging operators. For convenience, we choose a field in the form $B_\mu^{ab}(x)=x^{\nu}F_{\nu\mu}^{ab}(x)$ that satisfies the Fock--Schwinger gauge condition. Then, using relation \eqref{ya-z-2}, we obtain that in the case of weak deformation, the main contribution to the functional $\mathrm{J}_{\sigma\rho}^{\phantom{1}}[\mathcal{N}_{1\sigma\rho}]$ is given by the combination
\begin{equation*}
	D_{\sigma}^{ac}(x)D_{\rho}^{bd}(y)\Big(\Phi^{cd}(x,y)R_0^{\Lambda,\mathbf{f}}(x-y)\Big)\longrightarrow
	\frac{1}{2}F_{\sigma\rho}^{ab}(y)R_0^{\Lambda,\mathbf{f}}(x)+\frac{1}{4} F_{\sigma\rho}^{ab}(y)x^\mu\partial_{x^\mu}R_0^{\Lambda,\mathbf{f}}(x),
\end{equation*}
where additionally variable shifting and spherical averaging were used. In the case of strong deformation, the Green's function is not invariant with respect to the gauge transformations of the background field, so the part in large parentheses has a different form. Let us use the decomposition near the diagonal from Section \ref{ym:sec:two-1}, see formula \eqref{ya-a-74-1}. It follows that for the transition\footnote{This transition is valid for searching for the local singular part. In the general case, this is incorrect.} to the case of strong deformation, it is necessary to replace
\begin{equation*}
\Phi^{cd}(x,y)R_0^{\Lambda,\mathbf{f}}(x-y)\longrightarrow
\delta^{cd}R_0^{\Lambda,\mathbf{f}}(x-y)+y^{\nu}F_{\nu\mu}^{ab}(y)\partial_{x^\mu}\theta(x-y).
\end{equation*} 
Then applying the two derivatives $D_{\sigma}^{ac}(x)D_{\rho}^{bd}(y)$, making a shift of the variable followed by spherical averaging, and also removing combinations symmetric in indices $\sigma$ and $\rho$, we obtain the following final terms
\begin{equation*}
\frac{1}{4}F_{\sigma\rho}^{ab}(y)\Big(A_0(x)\theta(x)+x^\mu\partial_{x^\mu}R_0^{\Lambda,\mathbf{f}}(x)/2\Big).
\end{equation*}
Thus, comparing the obtained combinations in both approaches, we see that the transition between the strong deformation and the weak one in the functional $\mathrm{J}_{\sigma\rho}^{\phantom{1}}[\mathcal{N}_{1\sigma\rho}]$ is carried out by the following replacement
\begin{equation*}
A_0(x)\theta(x)\longleftrightarrow 2R_0^{\Lambda,\mathbf{f}}(x)+x^\mu\partial_{x^\mu}R_0^{\Lambda,\mathbf{f}}(x)/2.
\end{equation*}
It is precisely this transition that leads to the connection of singular integrals in the pair \eqref{ya-a-151}$\longleftrightarrow$\eqref{ya-y-6}, since integrals also change after the density transformation, see Section \ref{ym:sec:two-vs},
\begin{equation*}
	\hat{\mathrm{I}}_1\stackrel{\mathrm{s.p.}}{\longleftrightarrow} 2\hat{\mathrm{I}}_5-\hat{\mathrm{I}}_6.
\end{equation*}
For completeness, we note that in the second pair \eqref{ya-a-152}$\longleftrightarrow$\eqref{ya-y-7}, integrals of the form are replaced
\begin{equation*}
\hat{\mathrm{I}}_2-2\hat{\mathrm{I}}_1\stackrel{\mathrm{s.p.}}{\longleftrightarrow} 4\hat{\mathrm{I}}_6-4\hat{\mathrm{I}}_5.
\end{equation*}
Considering the validity of the transition from the first pair, the second substitution reduces to
$\hat{\mathrm{I}}_2\stackrel{\mathrm{s.p.}}{\longleftrightarrow} 2\hat{\mathrm{I}}_6$.

\subsection{Comparison with dimensional regularization}
\label{ym:sec:two1-4}

Considering the singularities for the first loop, see for example \cite{12} and the first point of Theorem \ref{ya-t1}, it may seem that "elementary" logarithmic singularities that $\sim1/\varepsilon$ in dimensional regularization and $\sim L=\ln(\Lambda/\sigma)$ in the case of cutoff, can be compared with each other using some linear relationship $1/\varepsilon\longleftrightarrow aL+b$. However, this impression is deceptive, since an analysis of the two-loop corrections shows that there is no such dependence.

In this paper, we managed to make a very important observation about two-loop singularities. It turns out that it is possible to identify a set of master-integrals in which the terms independent of the deforming function have a direct relationship with the singularities for dimensional regularization. In this case, the difference between the two approaches arises from the fact that in the case of dimensional regularization, the coefficients with which the master-integrals enter the diagrams are deformed, whereas in the case of cutoff, the integrals themselves are deformed. We demonstrate this explicitly in several steps.\\

\noindent\underline{\textbf{Stage 1.}} First, let us compare the elementary blocks involved in constructing the decomposition of the Green's functions near the diagonal. The first function is the main part of the asymptotic expansion of $R_0(x)$
\begin{equation*}
\frac{\Gamma(d/2-1)}{4\pi^{d/2}|x|^{d-2}}=R_0^\varepsilon(x)\longleftarrow 
R_0^{\phantom{1}}(x)\longrightarrow
R_{0}^{\Lambda,\mathbf{f}}(x)\,\,\,\mbox{from}\,\,\,\eqref{ya-a-37}.
\end{equation*}
In this case, $d=4-\varepsilon$, where $\varepsilon$ is a small parameter of dimensional regularization, the transition of which $\varepsilon\to +0$ removes the regularization. As the second function, we choose the coefficient $R_1(x)$ about $2F_{\mu\nu}^{ab}$ when decomposing the Green's vector function near the diagonal, which obeys the relation\footnote{In the case of dimensional regularization, the operator is also deformed, since the quadratic form also changes.} $A_0(x)R_1(x)=R_0(x)$ in some fixed neighborhood of the origin. In the absence of regularization, it is equal to $-\ln(|x|^2\sigma^2)/(16\pi^2)$. After the deformation, we get\footnote{For dimensional regularization, the parameter designation $\mu$ is usually used instead of $\sigma$. We violated this agreement for convenience.}
\begin{equation*}
\frac{1}{16\pi^2}\bigg(\frac{2\sigma^{-\varepsilon}}{\varepsilon}+
\frac{\Gamma(d/2-2)}{\pi^{d/2-2}}\vert x\vert^{4-d}\bigg)=R_1^\varepsilon(x)\longleftarrow 
R_1^{\phantom{1}}(x)\longrightarrow
\theta(x)\,\,\,\mbox{from}\,\,\,\eqref{ya-a-76}.
\end{equation*}
The third function is a coefficient of about $2F_{\mu\sigma}^{ac}(x)F_{\sigma\nu}^{cb}(x)$, which, in the case of cutoff, is contained in the term $\mathcal{L}_{\mu\nu}^{ab}$, see \eqref{ya-a-77}. By construction, such a function should result\footnote{It is possible to have small corrections with respect to the regularizing parameter. This depends on the definition of the additive.} to $2R_1(x)$ after the action of the operator $A_0(x)$ and is usually denoted by $R_2(x)-|x|^2/(2^7\pi^2)$, where
\begin{equation*}
R_2(x)=\frac{|x|^2\big(\ln(|x|^2\sigma^2)-1\big)}{64\pi^2}.
\end{equation*}
Comparing such a function with the case of dimensional regularization \cite{12,13} and with the cutoff from Section \ref{ym:sec:two-1}, we obtain 
\begin{equation*}
\frac{1}{32\pi^2}\bigg(-\frac{\vert x\vert^2\sigma^{-\varepsilon}}{\varepsilon}+
\frac{\Gamma(d/2-3)}{2\pi^{d/2-2}}\vert x\vert^{6-d}\bigg)=R_2^\varepsilon(x)\longleftarrow 
	R_2^{\phantom{1}}(x)\longrightarrow
	2\tau(x)+\frac{|x|^2}{2^7\pi^2}\,\,\,\mbox{from}\,\,\,\eqref{ya-a-78}.
\end{equation*}

\noindent\underline{\textbf{Stage 2.}} At the next stage, we need to compare the master-integrals. In the case of dimensional regularization, eight pieces of $\mathrm{I}_1$--$\mathrm{I}_8$ were used, see Section 2.3 in \cite{Ivanov-Kharuk-20222}. In the case of the cutoff regularization, seven functionals were used $\hat{\mathrm{I}}_1$--$\hat{\mathrm{I}}_7$, see Section \ref{ym:sec:two-vs}. Let us present the results of the comparisons in the form of Table \ref{ya:table:2}. In this case, the relations are not a purposeful selection of linear combinations, but a consequence of the density comparisons from the previous step.

\begin{table}[h!]
	\centering
	\setlength\arrayrulewidth{0.6pt}
	\renewcommand{\arraystretch}{1.8}
	\begin{tabular}{|c|c|c||c|c|}
		\hline
		\multicolumn{3}{|c||}{\mbox{Cutoff regularization}}&\multicolumn{2}{c|}{\mbox{Dimensional regularization}}\\
		\hline
		\multicolumn{1}{|c|}{\mbox{Int.}$\times(4\pi)^4$}&\multicolumn{1}{c|}{\mbox{Main part}}&\multicolumn{1}{c||}{\mbox{Correction}}&\multicolumn{1}{c|}{\mbox{Int.}$\times(4\pi)^4\sigma^{2\varepsilon}$}&\multicolumn{1}{c|}{\mbox{Singularity}}\\
		\hline
		$\hat{\mathrm{I}}_1$&$2L^2+L(2+4\rho_1)$&$L(4\rho_2)$&$-(\mathrm{I}_4+d\mathrm{I}_3)/(2c_2^2)$&$2/\varepsilon^2+1/\varepsilon$\\
		\hline
		$\hat{\mathrm{I}}_2$&$4L^2+L(8\rho_1)$&$L(8\rho_2)$&$-d\mathrm{I}_3/c_2^2$&$4/\varepsilon^2$\\
		\hline
		$\hat{\mathrm{I}}_3$&$0$&$L(24\rho_3\rho_5-8\rho_4)$&$-d(\mathrm{I}_1+\mathrm{I}_2)/(4c_2^2)$&$0$\\
		\hline
		$\hat{\mathrm{I}}_4$&$2L^2+L(2+4\rho_1)$&$0$&$-d(\mathrm{I}_1-\mathrm{I}_7)/(2c_2^2)$&$2/\varepsilon^2+1/\varepsilon$\\
		\hline
		$\hat{\mathrm{I}}_5$&$2L^2+L(2+4\rho_1)$&$0$&$-2(\mathrm{I}_6-\mathrm{I}_5)/c_2^2$&$2/\varepsilon^2+1/\varepsilon$\\
		\hline
		$\hat{\mathrm{I}}_6$&$2L^2+L(1+4\rho_1)$&$0$&$2\mathrm{I}_5/c_2^2$&$2/\varepsilon^2+1/(2\varepsilon)$\\
		\hline
		$\hat{\mathrm{I}}_7$&$8L$&$0$&$8d\mathrm{I}_7/c_2^2=8d\mathrm{I}_8/c_2^2$&$4/\varepsilon$\\
		\hline
	\end{tabular}
	\renewcommand{\arraystretch}{1}
\caption{The first column contains the master-integrals for the cutoff regularization. The fourth column contains expressions using master-integrals for the dimensional regularization. Here $L=\ln(\Lambda/\sigma)$, and $\rho_i$ are auxiliary functionals (numbers) from Section \ref{ym:sec:two-vs}. }
\label{ya:table:2}
\end{table}
\noindent The table shows that in the case of cutoff regularization, the "base" singularity is the combination of $L+\rho_1$, and not just $L$, which is equal to $8\pi^2\theta(0)$. This value is more natural. A similar situation was observed in the two-dimensional sigma model, see \cite{AIK-25}, and played an important role\footnote{The suggested choice significantly reduces the number of calculations.} when studying three-loop contributions. Note that, if necessary, the sum of $L+\rho_1$ can be reduced to $L$ by shifting the parameter $\sigma$. Thus, the main "poles" have the direct connection
\begin{equation*}
(L+\rho_1)^2\longleftrightarrow\frac{1}{\varepsilon^2},
\end{equation*}
whereas all first-order singularities that do not depend on the choice of the regularizing function $\mathbf{f}(\cdot)$ differ by two. The "correction" column contains the singularities that occur due to the smoothing of the delta-functional. Some of them are reduced after summing all the diagrams, which is a consequence of the invariance of the two-loop contribution with respect to the shift by a constant with a multiplier of $L/\Lambda^2$, and some can be neutralized by introducing quasi-local vertices, see Section \ref{ym:sec:quas}. Thus, only the functional $\rho_2$ remains, which is responsible for the deformation of the master-integrals, leading to a change in the two-loop corrections.\\

\noindent\underline{\textbf{Stage 3.}} At this stage, individual diagrams and their linear combinations can be discussed. Note that in this paper, the contribution of $\mathbb{H}_0^{\mathrm{sc}}\big(\Gamma_3^2\big)$, see formula \eqref{ya-a-19}, consists of five diagrams. A similar contribution in the works \cite{Ivanov-Kharuk-2020,Ivanov-Kharuk-20222} was represented by the sum of four parts
\begin{equation*}
\mathbb{H}_0^{\mathrm{sc}}\big(\Gamma_3^2\big)=\sum_{i=1}^42\mathcal{J}_i,
\end{equation*}
where, using formulas (48) and (49) from \cite{Ivanov-Kharuk-2020}, the following connecting relations are true
\begin{equation*}
	2\mathcal{J}_1=-2
	{\centering\adjincludegraphics[width = 1 cm, valign=c]{fig/ya10.eps}},\,\,\,
	2\mathcal{J}_2=2{\centering\adjincludegraphics[width = 1 cm, valign=c]{fig/ya9.eps}},\,\,\,
	2\mathcal{J}_3=-
	{\centering\adjincludegraphics[width = 1 cm, valign=c]{fig/ya7.eps}}-
	{\centering\adjincludegraphics[width = 1 cm, valign=c]{fig/ya10.eps}},
\end{equation*}
\begin{equation*}
	2\mathcal{J}_4=2
	{\centering\adjincludegraphics[width = 1 cm, valign=c]{fig/ya10.eps}}+
	{\centering\adjincludegraphics[width = 1 cm, valign=c]{fig/ya11.eps}}+ {\centering\adjincludegraphics[width = 1 cm, valign=c]{fig/ya8.eps}}.
\end{equation*}
At the same time, similar designations were used for the remaining parts
\begin{equation*}
\mathbb{H}_0^{\mathrm{sc}}\big(\Gamma_3^2\big)=2\mathcal{J}_5,\,\,\,
\mathbb{H}_0^{\mathrm{sc}}\big(\Gamma_4\big)=-4\mathcal{J}_6.
\end{equation*}
The question of matching the master-integrals concerns diagrams with two integrations, that is $\mathcal{J}_1$--$\mathcal{J}_5$. Let us use the formulas derived earlier, see (61)--(66) in \cite{Ivanov-Kharuk-20222}, for the case of dimensional regularization and rewrite them, dividing them into two parts: the main one, which allows direct comparison with the case of the cutoff according to Table \ref{ya:table:2}, and the correction, which is responsible for the deformation of the coefficients with which the master-integrals enter. We have the following formulas
\begin{fleqn}
\begin{equation}\label{ya-y-1}
2\mathcal{J}_1\times(4\pi)^4\sigma^{2\varepsilon}\stackrel{\mathrm{s.p.}}{=}
\Big(2d\mathrm{I}_1+d\mathrm{I}_2+\mathrm{I}_4-4\mathrm{I}_5+4\mathrm{I}_7\Big)+\varepsilon
\Big(2\mathrm{I}_1+\mathrm{I}_2\Big)
\end{equation}
\begin{equation*}
\phantom{2\mathcal{J}_1\times(4\pi)^4\sigma^{2\varepsilon}}\,\to\Big(-2\hat{\mathrm{I}}_1+\hat{\mathrm{I}}_2-4\hat{\mathrm{I}}_3-2\hat{\mathrm{I}}_4-2\hat{\mathrm{I}}_6+\hat{\mathrm{I}}_7/4\Big)\times c_2^2,
\end{equation*}
\end{fleqn}
\begin{fleqn}
	\begin{equation}\label{ya-y-2}
		2\mathcal{J}_2\times(4\pi)^4\sigma^{2\varepsilon}\stackrel{\mathrm{s.p.}}{=}
		\Big(-4d\mathrm{I}_1-2d\mathrm{I}_2-2\mathrm{I}_4+8\mathrm{I}_5-8\mathrm{I}_7\Big)+\varepsilon\Big(-2\mathrm{I}_5\Big)
	\end{equation}
\begin{equation*}
\phantom{2\mathcal{J}_2\times(4\pi)^4\sigma^{2\varepsilon}}\,\to\Big(4\hat{\mathrm{I}}_1-2\hat{\mathrm{I}}_2+8\hat{\mathrm{I}}_3+4\hat{\mathrm{I}}_4+4\hat{\mathrm{I}}_6-\hat{\mathrm{I}}_7/2\Big)\times c_2^2,
\end{equation*}
\end{fleqn}
\begin{fleqn}
	\begin{equation}\label{ya-y-3}
	2\mathcal{J}_3\times(4\pi)^4\sigma^{2\varepsilon}\stackrel{\mathrm{s.p.}}{=}
	\Big(-2d\mathrm{I}_1-d\mathrm{I}_2-d\mathrm{I}_3-4\mathrm{I}_5-d\mathrm{I}_6-d\mathrm{I}_7\Big)+\varepsilon\Big(\mathrm{I}_5\Big)
\end{equation}
\begin{equation*}
	\phantom{2\mathcal{J}_3\times(4\pi)^4\sigma^{2\varepsilon}}\,\to\Big(\hat{\mathrm{I}}_2+4\hat{\mathrm{I}}_3+2\hat{\mathrm{I}}_4+2\hat{\mathrm{I}}_5-4\hat{\mathrm{I}}_6-\hat{\mathrm{I}}_7/4\Big)\times c_2^2,
\end{equation*}
\end{fleqn}
\begin{fleqn}
	\begin{equation}\label{ya-y-4}
	2\mathcal{J}_4\times(4\pi)^4\sigma^{2\varepsilon}\stackrel{\mathrm{s.p.}}{=}
	\Big(-d\mathrm{I}_1/2-d\mathrm{I}_2/4-d\mathrm{I}_3/2+\mathrm{I}_4/4-3\mathrm{I}_5-6\mathrm{I}_6-\mathrm{I}_7\Big)+\varepsilon\Big(-\mathrm{I}_1/2-\mathrm{I}_2/4-\mathrm{I}_3/2\Big)
\end{equation}
\begin{equation*}
	\phantom{2\mathcal{J}_4\times(4\pi)^4\sigma^{2\varepsilon}}\,\to\Big(-\hat{\mathrm{I}}_1/2+3\hat{\mathrm{I}}_2/4+\hat{\mathrm{I}}_3+\hat{\mathrm{I}}_4/2+\hat{\mathrm{I}}_5-5\hat{\mathrm{I}}_6/2-\hat{\mathrm{I}}_7/16\Big)\times c_2^2,
\end{equation*}
\end{fleqn}
\begin{fleqn}
	\begin{equation}\label{ya-y-5}
	2\mathcal{J}_5\times(4\pi)^4\sigma^{2\varepsilon}\stackrel{\mathrm{s.p.}}{=}
	\Big(d\mathrm{I}_2/4-\mathrm{I}_5+\mathrm{I}_7+2\mathrm{I}_8\Big)+\varepsilon\Big(\mathrm{I}_2/4\Big)
\end{equation}
\begin{equation*}
	\phantom{2\mathcal{J}_5\times(4\pi)^4\sigma^{2\varepsilon}}\,\to\Big(-\hat{\mathrm{I}}_3+\hat{\mathrm{I}}_4/2-\hat{\mathrm{I}}_6/2+\hat{\mathrm{I}}_7/16\Big)\times c_2^2.
\end{equation*}
\end{fleqn}
Here, the first line corresponds to the case of dimensional regularization, and the second line corresponds to the case of the cutoff. At the same time, when switching from the "dimensional" situation to the cutoff, the second part of the first row resets to zero, and then the coefficients in the first part are deformed, and vice versa. Let us substitute the values for the master integrals into the obtained ratios, explicitly dividing the answers into the main part and the correction. The values obtained are shown in Table \ref{ya:table:3}.

Let us note an interesting property possessed by correction terms for dimensional regularization. If we sum up all the parts that do not contain $\mathrm{I}_5$, we get
\begin{equation*}
\varepsilon
\Big(2\mathrm{I}_1+\mathrm{I}_2\Big)+
\varepsilon\Big(-\mathrm{I}_1/2-\mathrm{I}_2/4-\mathrm{I}_3/2\Big)+
\varepsilon\Big(\mathrm{I}_2/4\Big)\stackrel{\mathrm{s.p.}}{=}0.
\end{equation*}
Consequently, such additions are similar to terms with $\rho_1$ in the case of the cutoff regularization, which eventually cancel each other out. In turn, the parts with the integral $\mathrm{I}_5$ are not reduced
\begin{equation*}
\varepsilon\Big(-2\mathrm{I}_5\Big)+\varepsilon\Big(\mathrm{I}_5\Big)\stackrel{\mathrm{s.p.}}{=}-\frac{c_2^2\sigma^{-2\varepsilon}}{(4\pi)^4}\frac{1}{\varepsilon}
\end{equation*}
and they are an analog of contributions with the functional $\rho_2$ for the cutoff case, which also remains.

\begin{table}[h!]
	\centering
	\setlength\arrayrulewidth{0.6pt}
	\renewcommand{\arraystretch}{1.8}
	\begin{tabular}{|r||c|c|c|c|}
		\hline
		\multicolumn{1}{|c||}{}&\multicolumn{2}{c|}{\mbox{Cutoff reg.}}&\multicolumn{2}{c|}{\mbox{Dimensional reg.}}\\
		\hline
		\multicolumn{1}{|c||}{\mbox{Diagram}}&$\mbox{Main part}$&$\mbox{Correction}$&$\mbox{Main part}$&$\mbox{Correction}$\\
		\hline
		$2\mathcal{J}_1$&$-8L^2+L(-8+16\rho_1)$&$L(32\rho_4)$&$-8/\varepsilon^2-4/\varepsilon$&$-1/\varepsilon$\\
		\hline
		$2\mathcal{J}_2$&$16L^2+L(16+32\rho_1)$&$L(-64\rho_4)$&$16/\varepsilon^2+8/\varepsilon$&$-2/\varepsilon$\\
		\hline
		$2\mathcal{J}_3$&$4L^2+L(2+8\rho_1)$&$L(8\rho_2-32\rho_4)$&$4/\varepsilon^2+1/\varepsilon$&$1/\varepsilon$\\
		\hline
		$2\mathcal{J}_4$&$-L$&$L(4\rho_2-8\rho_4)$&$-1/(2\varepsilon)$&$3/(4\varepsilon)$\\
		\hline
		\rowcolor[gray]{.9}			$\mathbb{H}_0^{\mathrm{sc}}\big(\Gamma_3^2\big)$&$12L^2+L(9+24\rho_1)$&$L(12\rho_2-72\rho_4)$&$12/\varepsilon^2+9/(2\varepsilon)$&$-5/(4\varepsilon)$\\
		\hline
		\rowcolor[gray]{.9}			$\mathbb{H}_0^{\mathrm{sc}}\big(\Omega_3^2\big)$&$L$&$L(8\rho_4)$&$1/(2\varepsilon)$&$1/(4\varepsilon)$\\
		\hline
		\rowcolor[gray]{.9}		$-\frac{1}{2}\mathbb{H}_0^{\mathrm{sc}}\big(\Gamma_4\big)$&$-12L^2+L(-24\rho_1)$&$0$&$-12/\varepsilon^2$&$0$\\
		\hline
		\rowcolor[gray]{.7}		\multicolumn{1}{|r||}{$-\frac{\mbox{\small{sum}}}{2}$}&$-5L$&$L(-6\rho_2+32\rho_4)$&$-5/(2\varepsilon)$&$1/(2\varepsilon)$\\
		\hline
		\rowcolor[gray]{.9}\multicolumn{1}{|r||}{$-\mathcal{J}_7$}&$0$&$L(5/36-10\rho_2/3)$&$0$&$-5/(6\varepsilon)$\\
		\hline
		\rowcolor[gray]{.7}\multicolumn{1}{|c||}{$\mbox{Answer}$}&
		\multicolumn{2}{c|}{$L(-175/36-28\rho_2/3+32\rho_4)$}
		&\multicolumn{2}{c|}{$-17/(6\varepsilon)$}\\
		\hline
	\end{tabular}
	\renewcommand{\arraystretch}{1}
	\caption{The table shows the local singular components corresponding to the second column of Table \ref{ya:table:1}. In this case, the second and third columns must be multiplied by $c_2^2W_{-1}/(4\pi)^4$, and the last two by $c_2^2W_{-1}\sigma^{-2\varepsilon}/(4\pi)^4$.}
	\label{ya:table:3}
\end{table}

\noindent\underline{\textbf{Stage 4.}} In addition, Table \ref{ya:table:3} shows the values and the division into the main part and the correction part for the counter-vertex. First, we give a complete singular value decomposition for each type of regularization
\begin{align*}
	-\frac{5Lc_2}{48\pi^2}
	\mathbb{H}_0^{\mathrm{sc}}\big(S_{\mathrm{f}}[\,\cdot\,,B]\big)\bigg|_{\mathrm{cutoff}}&\stackrel{\mathrm{s.p.}}{=}-\frac{5}{6}\frac{c_2^2W_{-1}}{(4\pi)^4}
	\times L(4\rho_2-1/6)
	+\frac{5Lc_2\mathrm{V}_1}{48\pi^2},
	\\
	-\frac{5c_2}{48\pi^2\varepsilon}
	\mathbb{H}_0^{\mathrm{sc}}\big(S_{\mathrm{f}}[\,\cdot\,,B]\big)\bigg|_{\mathrm{d.reg.}}^{\phantom{\mathrm{cutoff}}}&\stackrel{\mathrm{s.p.}}{=}-\frac{5}{6}\frac{c_2^2W_{-1}\sigma^{-2\varepsilon}}{(4\pi)^4}\times\frac{1}{\varepsilon}
	+\frac{5c_2\mathrm{V}_1}{48\pi^2\varepsilon}.
\end{align*}
The latter value, following the definitions from \cite{Ivanov-Kharuk-20222}, is denoted by $-\mathcal{J}_7$. In both cases, the local singularity proportional to $W_{-1}$ is recorded as a correction part. Let us compare the corresponding terms with the case without regularization. First, the part of $M_0^{ab}G_{\mathrm{loc}}^{ba}$ proportional to $F_{\nu\sigma}^{ab}F_{\mu\sigma}^{ba}$ contains the following multiplier
\begin{equation*}
	e_{1}^{\mu\nu}(x)=-\frac{x^{\mu\nu}}{4}R_0(x)+A_0(x)\bigg(x^{\mu\nu}R_1(x)+R_2(x)\delta^{\mu\nu}-\frac{|x|^2\delta^{\mu\nu}}{2^7\pi^2}\bigg)\frac{1}{12}.
\end{equation*}
In the absence of regularization, it is zero $e_{1}^{\mu\nu}(x)=0$ for all values of $x$. However, when regularization is introduced, it may deform. It is not difficult to verify that 
\begin{equation*}
		e_{1}^{\mu\nu}\Big|_{\mathrm{d.reg.}}^{\phantom{\mathrm{cutoff}}}(0)=0
	,\,\,\,
	e_{1}^{\mu\nu}\Big|_{\mathrm{cutoff}}(0)=\frac{\delta^{\mu\nu}}{48(4\pi^2)}
	.
\end{equation*}
In the second case, there is a combination with the field strength tensor of the form $F_{\nu\sigma}^{ab}F_{\sigma\mu}^{ba}$. Again, in the unregulated case, it has the form
\begin{equation*}
	e_{2}^{\mu\nu}(x)=-R_1(x)\delta^{\mu\nu}-2\partial_{x_\mu}\partial_{x_\nu}\bigg(R_2(x)-\frac{|x|^2}{2^7\pi^2}\bigg)=\frac{\delta^{\mu\nu}-4x^\mu x^\nu/|x|^2}{2^5\pi^2}.
\end{equation*}
In this case, we have an ambiguity at zero. It depends on the limit (the direction of approach to zero). Turning to the dimensional regularization or to the cutoff, we are convinced of the validity of the following relations
\begin{equation*}
	e_{2}^{\mu\nu}\Big|_{\mathrm{d.reg.}}^{\phantom{\mathrm{cutoff}}}(0)=\frac{\delta^{\mu\nu}}{32\pi^2}
	,\,\,\,
	e_{2}^{\mu\nu}\Big|_{\mathrm{cutoff}}(0)=\frac{\rho_2\delta^{\mu\nu}}{8\pi^2}
	.
\end{equation*}
Thus, in the case of dimensional regularization, the ambiguity is revealed by zeroing the second term, while in the case of cutoff, both parts are important, and their total value has the opposite sign, since $\rho_2\leqslant0$. Nevertheless, it is the second option that is more natural, since in the case of cutoff in the momentum representation reach a maximum of $\rho_2\to0$, which is consistent with the fact that the trace of the matrix $e_2^{\mu\nu}$ is zero in the absence of regularization.

\section{On renormalization of mass}
\label{ym:sec:mas}

Section \ref{ym:sec:int} discussed the renormalization process and emphasized the fact that renormalizability may depend on the type of regularization. In this case, it is necessary either to expand the classical action by adding new terms with new parameters, or to expand the renormalization rules by adding new counter-terms. It should be noted right away that the second option is a special case of the first one, since when setting the newly added parameters in the classical action, both results must match. In Sections \ref{ym:sec:pr:gen3} and \ref{ym:sec:re}, the second approach was used, as it is easier to implement. The main objective of this section is to demonstrate the first approach within the framework of the first two quantum corrections for the weak deformation and to discuss additional changes in the case of the strong deformation.\\

\noindent\textbf{Weak deformation.} Regularization in this approach is introduced in a covariant manner \eqref{ya-z-15}, therefore, the counter-vertices are also invariant with respect to the gauge transformations of the background field $B_\mu^a$, see \eqref{ya-p-1}. This fact means that new additional parts of the classical action \eqref{ya-z-25} must also have this property. Now note that within the first two quantum corrections, see Theorem \ref{ya-t2}, only one new counter-vertex was introduced
\begin{equation}\label{ya-s-1}
S_2[a]=\int_{\mathbb{R}^4}\mathrm{d}^4x\,a_\mu^a(x)a_\mu^a(x)
\end{equation}
thus, the classical action from \eqref{ya-z-25} was shifted by the value
\begin{equation*}
g^2\frac{\Lambda^2c_2(2\rho_3-3\rho_0)}{8\pi^2}S_2[a].
\end{equation*}
For verification, it is enough to take into account formula \eqref{ya-a-173} and its special case from Section \ref{ym:sec:re-5}. The definitions of the coefficients $\rho_0$ and $\rho_3$ are given in Theorem \ref{ya-t2}. Therefore, the extended classical action should differ only in the mass\footnote{Intuitively, it would be more correct to call the mass term $S_2[B+ga]$. For this reason, the parameter is designated $\mu^2$, not $m^2$. The authors believe that the term with $S_2[a]$ has more to do with a renormalization of a measure rather than the mass. Anyway, it is an open question.} summand and have the form
\begin{equation}\label{ya-p-5}
S_{\mathrm{tot}}[B,e]+\frac{\mu^2}{2}S_2[a].
\end{equation}
Note that covariance is present here, since the field $a_\mu^a$ is transformed according to the law from \eqref{ya-p-4}. After such a change, the renormalization procedure with the addition of a new counter-vertex can be reformulated in a standard way by multiplying the parameter $\mu^2$ by the auxiliary renormalization constant $Z_\mu$, which decomposes into a formal series\footnote{Usually, an ansatz with $z_{\mu,1}=0$ is used. We consider a more general case, since $z_{\mu,1}=0$ is not obvious. It follows from the calculations.} by the coupling constant
\begin{equation*}
Z_\mu^{\phantom{1}}=1+z_{\mu,1}^{\phantom{1}}+g^2_{\mathrm{ren}}z_{\mu,2}^{\phantom{1}}+\mathcal{O}(g^4_{\mathrm{ren}}).
\end{equation*}
In this case, the coefficients $z_{\mu,1}$ and $z_{\mu,2}$ are determined based on the first two corrections. They contain the "power-law" parts $z_{\mu,1}^{\Lambda}$ and $z_{\mu,2}^\Lambda$, which follow from Theorem \ref{ya-t2} and are equal to
\begin{equation*}
z_{\mu,1}^{\Lambda}=0,\,\,\,
z_{\mu,2}^{\Lambda}=\frac{\Lambda^2c_2(2\rho_3-3\rho_0)}{4\mu^2\pi^2},
\end{equation*}
as well as the "logarithmic" parts of $z_{\mu,1}^{L}$ and $z_{\mu,2}^L$, which must be additionally calculated. To determine the first coefficient, it suffices to note that the contribution to the first correction, containing $\mu^2$, is proportional to $\mathbb{H}_0^{\mathrm{sc}}\big(S_2[\,\cdot\,]\big)$. Taking into account the decomposition for the Green's function near the diagonal from Section \ref{ym:sec:two1-1}, the equality $F_{\nu\nu}^a=0$, and the relation for $J_{\odot}$ from \eqref{ya-p-6}, we arrive at the conclusion that there are no logarithmic singularities proportional to $\mu^2$ in the first approximation. Thus, we have\footnote{Generally speaking, we can add an arbitrary constant. However, in this case, the first coefficient in $Z_\mu$ will not be equal to one.} $z_{\mu,1}^{L}=0$. Moving on to the second "loop", we note that all diagram blocks are covariant, so for general reasons it can be argued that $z_{\mu,2}^L$ will not contain singularities. It can be chosen to be equal to any finite constant, in particular, $z_{\mu,2}^L=0$. If we turn to the mechanics of computing, then for the search it is necessary to solve the equality\footnote{It follows from \eqref{ya-a-163} and \eqref{ya-a-184} after adding the "mass" vertex.}
\begin{equation}\label{ya-p-7}
-\frac{1}{2}\mathbb{H}_0^{\mathrm{sc}}\big(\Gamma_3^2S_2^{\phantom{1}}\big)
+\frac{1}{4}\mathbb{H}_0^{\mathrm{sc}}\big(\Gamma_4^{\phantom{1}}S_2^{\phantom{1}}\big)-\frac{1}{2}\mathbb{H}_0^{\mathrm{sc}}\big(\Omega_3^2S_2^{\phantom{1}}\big)+\mathbb{H}_0^{\mathrm{sc}}\big(\hat{\Gamma}_2^0S_2^{\phantom{1}}\big)-z_{\mu,2}^L
\mathbb{H}_0^{\mathrm{sc}}\big(S_2^{\phantom{1}}\big)
\stackrel{\mathrm{s.p.}}{=}0.
\end{equation}
To do this, perform calculations related to the substitution of the decomposition for the Green's function from Section \ref{ym:sec:two1-1}, and use Lemma 4 from \cite{ya-30}. Then we can see that the nonlocal terms are reduced, since in this case the combinations completely repeat those already studied in Section \ref{ym:sec:two}, and the local singular part is proportional to the trace of the field strength tensor and therefore equal to zero. Based on all that has been said, the renormalization constant can be rewritten as
\begin{equation*}
	Z_\mu^{\phantom{1}}=1+g^2_{\mathrm{ren}}\frac{\Lambda^2c_2(2\rho_3-3\rho_0)}{4\mu^2\pi^2}+\mathcal{O}(g^4_{\mathrm{ren}}).
\end{equation*}
It can be assumed that in higher loops, the logarithmic parts will also be missing due to the covariance of regularization. However, this fact needs to be further investigated, as it may depend on renormalizability\footnote{At the moment, there is no proof that in higher corrections, the standard $\mathcal{R}$- operation of Bogoliubov--Parasyuk subtracts all necessary subsingularities in diagrams.} as a whole. Nevertheless, at the level of two corrections, it can be argued that in the covariant case there was not much point in introducing a mass term, since the power singularities remained the same, and the logarithmic ones did not appear.\\

\noindent\textbf{Strong deformation.} In this case, it is necessary to rely on the results of Theorem \ref{ya-t1}. It can be seen that within the framework of two quantum corrections, new counter-vertices\footnote{The type of the first functional is determined by the need to keep the relation of the action and the equation density. For this reason, $S_2[B+ga]$ is selected rather than $S_2[B]$.} are $S_2[B+ga]$ and $S_2[a]$, so the extended action for the weak deformation from \eqref{ya-p-5} is no longer sufficient. It is necessary to consider the following extended functional of the classical action
\begin{equation*}
	S_{\mathrm{tot}}[B,e]+\frac{\mu^2}{2}S_2[a]+\frac{m^2}{2g^2}S_2[B+ga].
\end{equation*}
The last term\footnote{Note that the standard notation of the mass parameter is used here.}, at the same time, is completely consistent with the connection between the action and the equation, see \eqref{ya-a-23}. In this case, invariance with respect to gauge transformations is violated already at the level of the classical action, which is due to the type of regularization. Let us move on to the renormalization process. Here we have already two renormalization constants\footnote{In the case of $Z_m$, the ansatz is chosen for convenience reasons.}
\begin{equation*}
Z_\mu^{\phantom{1}}=1+z_{\mu,1}^{\phantom{1}}+g^2_{\mathrm{ren}}z_{\mu,2}^{\phantom{1}}+\mathcal{O}(g^4_{\mathrm{ren}}),\,\,\,
Z_m^{\phantom{1}}=\frac{g^2}{g^2_{\mathrm{ren}}}\Big(1+g^2_{\mathrm{ren}}z_{m,1}^{\phantom{1}}+g^4_{\mathrm{ren}}z_{m,2}^{\phantom{1}}+\mathcal{O}(g^6_{\mathrm{ren}})\Big).
\end{equation*}
To restore the form of the coefficients, let us turn to the results of Theorem \ref{ya-t1}. It follows from the first point that in order to compensate for the singularities, the classical action must be shifted by the value
\begin{equation*}
\frac{\Lambda^2c_2(\rho_3-2\rho_0)}{8\pi^2}S_2[B]+g^2_{\mathrm{ren}}\frac{\Lambda^2c_2\alpha}{8\pi^2}S_2[B]+g^2_{\mathrm{ren}}\frac{\Lambda^2c_2(2\rho_3-3\rho_0)}{8\pi^2}S_2[a]+\ldots,
\end{equation*}
where the ellipsis denotes the part containing the cross term with $B_\mu^aa_\mu^a$, which is uniquely reconstructed taking into account \eqref{ya-z-16}. Therefore, the "power" parts of the first two coefficients take the form
\begin{equation*}
z_{m,1}^{\Lambda}=\frac{\Lambda^2c_2(\rho_3-2\rho_0)}{4m^2\pi^2},
\,\,\,
z_{m,2}^{\Lambda}=\frac{\Lambda^2c_2\alpha}{4m^2\pi^2}
,
\end{equation*}
\begin{equation*}
z_{\mu,1}^{\Lambda}=0,\,\,\,
z_{\mu,2}^{\Lambda}=\frac{\Lambda^2c_2(\rho_3-\rho_0)}{4\mu^2\pi^2}.
\end{equation*}
Returning to logarithmic singularities, we calculate only the first ones. They have the form $z_{m,1}^{\Lambda}=0$ and $z_{m,1}^{\Lambda}=0$ in the scheme with minimal subtractions and are a consequence of decompositions of the Green's functions near the diagonal, see formula \eqref{ya-a-127} in Section \ref{ym:sec:two-1}, and relation \eqref{ya-a-160}. Indeed, in this case, the diagram $\mathbb{H}_0^{\mathrm{sc}}\big(S_2[\,\cdot\,]\big)$ has no singularities. The second coefficients, which are generally proportional to the first power of the logarithm $L$, are the result of a study of equality \eqref{ya-p-7} and are not included in this work. They should be perceived as a "deviation" from the covariant case.

\section{On quasi-local vertices}
\label{ym:sec:quas}

In Theorems \ref{ya-t1} and \ref{ya-t2}, the two-loop coefficient of the beta-function was calculated. It turned out that, in addition to the numerical term, the answer also contains the functionals $\rho_2$ and $\rho_4$, depending on the deforming function $\mathbf{f}(\cdot)$, see \eqref{ya-a-37}. These functions differ significantly and, according to the authors, have a different nature of origin. For example, the functional $\rho_2$, see Section \ref{ym:sec:two-vs}, constructed as the difference of convolutions of functions $R_0^1(x)$ and $A_0^{\phantom{1}}(x)R_0^\Lambda(x)$ and, thus, reflects the presence of a deformed $\delta$-functional, that is, the transition $\delta(x)\longrightarrow A_0^{\phantom{1}}(x)R_0^1(x)$. In the absence\footnote{They are also zero when the "sharp cutoff" in the momentum representation is used.} of the regularization, there are also no such functionals.

In turn, the functional $\rho_4$, see Section \ref{ym:sec:two-vs}, contains no difference and is nonzero, including the case of absence of any regularization. Indeed, in this case we get
\begin{equation*}
\frac{\pi^2}{2}\int_{\mathrm{B}_{1}}\mathrm{d}^4x\,R_0(x)|x|^2A_0(x)R_0(x)=
\frac{1}{8}\int_{\mathrm{B}_{1}}\mathrm{d}^4x\,\delta(x)=\frac{1}{8}.
\end{equation*}
Therefore, in a situation without regularization, when all the symmetries of the original model are present, such an integral also does not vanish. At the same time, its contribution is quite natural and does not violate "physics" or "combinatorics". Additional interest in $\rho_4$ is due to its appearance in two-loop diagrams, see Table \ref{ya:table:1}. Note that the coefficient with which this functional appears is proportional to the coefficient for the functional $\rho_3$, which is completely subtracted by the counter-vertex $\mathbb{H}_0^{\mathrm{sc}}(S_2)$. The question arises: "Is it possible to generalize the counter-vertex $S_2[a]$ by passing to a quasi-local density in such a way that it subtracts all the functionals of $\rho_4$ and does not affect the rest?" We give an explicit construction as a positive answer.

So, the vertex $S_2[a]$ is determined by integral \eqref{ya-s-1}. In this case, the decomposition by ultraviolet singularities for the diagram $\Lambda^2\mathbb{H}_0^{\mathrm{sc}}(S_2)$, taking into account \eqref{ya-a-158}, in the case of weak deformation, see Section \eqref{ym:sec:two1-1}, can be written out as
\begin{equation}\label{ya-s-2}
\Lambda^2\int_{\mathbb{R}^4}\mathrm{d}^4x\,\Big(G_{1\mu\mu}^{\Lambda\, aa}(x,x)-4R_0^\Lambda(0)\dim(\mathfrak{g})\Big)\stackrel{\mathrm{s.p.}}{=}
\frac{3c_2L\rho_5}{2\pi^2}W_{-1}+\Lambda^2\times\mbox{$PS$-part}.
\end{equation}
Let us take as an example\footnote{The introduction of quasi-local kernel is not unique. We can build a whole family of them. A similar situation occurs in the two-dimensional model of the principal chiral field, see \cite{AIK-25}. The additional conditions that uniquely fix the kernel are not clear at the moment.} the transition from the functional $S_2[\,\cdot\,]$ with the local density to the functional $Q[\,\cdot\,]$ with a quasi-local density, that is, allowing the presence of an integral kernel depending on the variables $x$ and $y$, which is zero at $|x-y|>1/\Lambda$. An explicit transition can be represented as
\begin{equation*}
S_2[a]\longrightarrow Q[a]=\frac{4\pi^2}{\rho_3\Lambda^2}
\int_{\mathbb{R}^4}\mathrm{d}^4x
\int_{\mathbb{R}^4}\mathrm{d}^4y\,
a_\mu^b(x)\Big(R_0^\Lambda(x-y)A_0^{\phantom{1}}(x)R_0^\Lambda(x-y)\Big)a_\mu^b(y).
\end{equation*}
Let us study the singular part of $\Lambda^2\mathbb{H}_0^{\mathrm{sc}}(Q)$. Immediately note that the calculation of the nonlocal component ($PS$-part) reduces to considering the function on the diagonal. Given the definition for $\rho_3$ from Section \ref{ym:sec:two-vs}, we get the same answer as for $\Lambda^2\mathbb{H}_0^{\mathrm{sc}}(S_2)$. Next, we note that part \eqref{ya-z-5} does not give a singular contribution depending on the background field, since there are no logarithmic components in the decomposition. Therefore, as in the case of the vertex $S_2$, the contribution comes only from the $\mathcal{L}$-part. Further, using the definition of $\mathrm{\hat{I}}_3$ from Section \eqref{ym:sec:two-vs}, we can show that
\begin{equation*}
\Lambda^2\frac{c_2\rho_3}{(4\pi)^2}\Big(\mathbb{H}_0^{\mathrm{sc}}(Q)-\mathbb{H}_0^{\mathrm{sc}}(S_2)\Big)
\stackrel{\mathrm{s.p.}}{=}
\frac{c_2}{(4\pi)^2}\Big(4c_2\mathrm{\hat{I}}_3-\frac{3c_2\rho_5\rho_3}{2\pi^2}W_{-1}L\Big)
\stackrel{\mathrm{s.p.}}{=}
-\frac{8c_2^2\rho_4L}{(4\pi)^4}W_{-1}.
\end{equation*}
Finally, comparing the result with Table \ref{ya:table:1}, we see that the new vertex $Q[\,\cdot\,]$ reduces the functionals of $\rho_4$ in all loops, while leaving the other parts the same. This idea is very attractive from the point of view of adjusting singular contributions. The fact is that at the moment there is no proof that the subtraction of subsingularities works for this regularization in all loops. Therefore, in case of violations of such a procedure, the introduction of quasi-local terms may give additional freedom in the selection of suitable counter-vertices. In other words, "superfluous" singularities can be neutralized by introducing a quasi-locality.

\section{Conclusion}
\label{ym:sec:zakl}

\subsection{Summary of results}

Let us briefly describe the main results. First of all, we recall that the work consists of two main parts devoted to the strong deformation, see Section \ref{ym:sec:pr:gen3}, as well as the weak deformation, see Section \ref{ym:sec:re}. Moreover, in the first case, the regularization is introduced by quasi-local probabilistic averaging of fluctuation fields, see formula \eqref{ya-a-41}, and in the second case, a covariant method of "averaging" is proposed, see formula \eqref{ya-re-1}, which coincides with the first in the main order. A detailed description of the relevant results is given in two theorems and auxiliary comments after them:
\begin{align*}
	\mbox{\textbf{strong deformation}}&\,\,\longrightarrow\,\,\mbox{Theorem \ref{ya-t1}}
	\,\,+\,\,\mbox{comments};
	\\
	\mbox{\textbf{weak deformation}}&\,\,\longrightarrow\,\,\mbox{Theorem \ref{ya-t2}}\,\,+\,\,\mbox{comments}.
\end{align*}
Schematically, the results can be presented in the form of the following auxiliary list.\\

\noindent\textbf{The first "loop".} In both cases, the logarithmic singular part for the first quantum correction has the standard form $11/3\times c_2LW_{-1}/(4\pi)^2$, where $L=\ln(\Lambda/\sigma)$, and coincides with the case of dimensional regularization, see for example \cite{16,17}. At the same time, the use of "strong" regularization additionally entails the appearance of power singularities $\sim\Lambda^2$, which can either be compensated by introducing auxiliary counter-vertices, as was done in Theorem \ref{ya-t1}, or by extending the classical action by adding a "mass" term, and by further renormalizing the mass, see Section \ref{ym:sec:mas}.\\

\noindent\textbf{The local part of the second "loop".} In the "weak" case, the local singular part is $\sim LW_{-1}$, see point 3 of Theorem \ref{ya-t2}. In this case, the proportionality coefficient does not coincide with the case of dimensional regularization, see for example \cite{12}, and depends on the regularizing function $\mathbf{f}(\cdot)$, see \eqref{ya-a-37}. In the "strong" case, the local singularity is complemented by the "massive" terms $\sim\Lambda^2$ and $\sim L\Lambda^2$, as well as a set of parts of the classical action $W_{-1}$ proportional to $L$ and reflecting a violation of invariance with respect to the gauge transformations of the background fields \eqref{ya-p-1}.\\

\noindent\textbf{The nonlocal part of the second "loop".} In both situations, see point 3 in Theorems \ref{ya-t1} and \ref{ya-t2}, the contribution has the same form and consists of two parts: the logarithmic $\sim L$, which leads to the renormalization of the functional $S_{\mathrm{f}}$, which fixes the gauge condition, and the power-law one $\sim\Lambda^2$, which leads to an additional "mass" term. At the same time, the second part does not violate invariance, since the fluctuation fields are transformed according to the law from \eqref{ya-p-4}. Also note that the logarithmic part is consistent with the case of other regularizations.\\

\noindent\textbf{The first "loop" for the first variation.} The singular part of the variation of quantum action \eqref{ya-a-17} with respect to the field responsible for the gauge condition was considered. It was shown that in both situations it has the same form, see point 2 in Theorems \ref{ya-t1} and \ref{ya-t2}, and is consistent with other regularizations. At the same time, in the "weak" case, it was also proved that the variation with respect to the field responsible for covariant deformation does not add new singularities within the framework of the first approximation.\\

\noindent\textbf{Quasi-local vertices.} In Section \ref{ym:sec:quas}, it was shown that part of the local singular component of the second approximation can be removed by moving from a "mass" counter-vertex with a local density to a generalized case with a quasi-local counter-vertex. At the same time, this change does not affect other singular contributions.\\

\noindent\textbf{Comparison.} In Section \ref{ym:sec:two1-4}, a detailed comparison of the singular parts in the case of dimensional regularization and in the case of weak deformation was performed. It has been shown that singular contributions can be represented as a linear combination of a finite number of master-integrals. In the case of dimensional regularization, the coefficients with which these integrals enter are deformed, and in the case of cutoff, the integrals themselves are deformed. For comparison, see Tables \ref{ya:table:2} and \ref{ya:table:3}.\\

\noindent\textbf{Faddeev's approach.} As part of the approach to working with the quantum action, see Sections \ref{ym:sec:pr:gen3-3}--\ref{ym:sec:pr:gen3-5} and \ref{ym:sec:re-5}, which consists in using a special type of background field, the renormalization procedure was considered. The consistency of the functionals before and after regularization was discussed, as well as options for adjusting the relations during the renormalization process. As an explanation, the procedure was performed for the first two quantum corrections.

\subsection{Comments}
\label{ym:sec:zakl-1}

\noindent\textbf{About the second correction.} An interesting open problem is to study the second correction for the variation of the quantum action both in the background field (that is, the standard second correction for the quantum equation of motion) and in the field responsible for the choice of the gauge condition. Moreover, these calculations are important not only for searching for the auxiliary current from \eqref{ya-a-181}, but also for searching for the singular part of the variation of the effective action with respect to the field responsible for regularization in the case of weak deformation. In this paper, it was shown that a similar value for the first approximation is zero. At the same time, it is expected that for the next correction term it will not only be nonzero, but will additionally lose gauge invariance with respect to background field changes. This is due to the fact that the gauge invariant regularizing operator restores invariance during the transition from the strong deformation to the weak one.\\

\noindent\textbf{About the third "loop".} The problem of calculating singular contributions for the third quantum correction in the case of cutoff regularization is open. At the same time, as part of the use of weak deformation, see Section \ref{ym:sec:two1}, it is quite feasible. Note that in this case, one can immediately use a decomposition of the form \eqref{ya-z-5}, since all constructions are invariant with respect to the gauge transformations of the background field, see \eqref{ya-p-1}. Note that, in addition, it is convenient to use an expansion with respect to the field strength tensor for ordered exponentials from \cite{ya-30}, as well as a number of already studied diagrams \cite{Ivanov-Kharuk-2023}.\\

\noindent\textbf{About classical solutions.} Note that the functional \eqref{ya-a-17} can be considered in the case when the background field $B_\mu^a=B_{\mathrm{cl},\mu}^a$ solves the classical equation of motion. In this case, the decomposition will contain connected diagrams, not just strongly connected ones. For example, a diagram like "glasses" is not present in the quantum action of $W_{\mathrm{reg}}^{\mathrm{sc}}[\,\cdot\,]$. From the point of view of the renormalization process, the connection between the action and the quantum equation is important in this case, since it guarantees the absence of additional singularities. Note that in the case of the weak deformation in the first two corrections, it is quite easy to verify renormalizability. Indeed, due to the covariance of the deformation in the quantum equation of motion, the singularity is proportional to the density of the classical equation of motion, which is zero due to the choice of the background field.\\

\noindent\textbf{About tree diagrams.} In Section \ref{ym:sec:re-6}, a special approach was discussed, which consisted in studying a strongly connected quantum action on a diagonal without controlling the renormalizability of partial variations in the background field. In this case, a renormalized connected action is, by definition, constructed using the renormalized strongly connected action and its derivatives. Explicit formulas for tree diagrams are important. Unfortunately, the authors were unable to find suitable formulas. Moreover, it is not clear whether the new action can be represented as a functional integral of some "classical" action.\\

\noindent\textbf{On quasi-local vertices.} In Section \ref{ym:sec:quas}, the possibility of using a vertex with a non-local density was considered. In general, the presence of such an object in a non-local theory is quite natural, although it violates the standard requirement of locality of counter-vertices. Nevertheless, this new approach makes it possible to eliminate the deformed coefficient by moving it from the renormalization constant to a new vertex. The technique used is very attractive from a computational point of view, since it reduces the number of new functionals in the renormalization constants, but the physical justification for its application is not entirely clear. Moreover, another open question arises related to the connection between hypothetical quasi-local vertices and local functionals. Is it possible to perform renormalization within local objects and then obtain quasi-local ones by additional summation, and vice versa? Similar observations have been proposed in the two-dimensional sigma model \cite{AIK-25}. \\

\noindent\textbf{On "mass" terms.} In Section \ref{ym:sec:mas}, the standard renormalization procedure for an extended classical action with a "mass" term was considered. Moreover, two parameters $m^2$ and $\mu^2$ were introduced, since a quadratic field functional can be constructed using both $B_\mu^a+ga_\mu^a$ and $a_\mu^a$. Intuitively, the first option is related to the standard mass functional, while the second one is related to the measure. Accordingly, it is not completely clear what exactly is being renormalized: mass or measure, or both.\\

\noindent\textbf{About the gauge condition.} Gauge conditions that do not depend on the background field are of particular interest. The most trivial can be obtained from \eqref{ya-a-31} by choosing $e_\mu^a=0$. Such a condition is standard and is called the Lorenz gauge. In this case, by choosing a regularization that does not depend on the background field, for example, using the strong deformation from Section \ref{ym:sec:pr:gen3}, we obtain equality of the full and partial functional derivatives with respect to the background field. Thus, we have
\begin{equation*}
	\frac{\delta W^{\mathrm{sc}}_{\mathrm{reg}}}{\delta B_{\mu}^a(x)}=Q_\mu^a|_{\mathrm{reg.}}(x),
\end{equation*}
see for comparison \eqref{ya-z-16} and \eqref{ya-zzz-1}. This choice automatically leads to the consistency of the renormalization process of the quantum action and its derivatives. Unfortunately, this construction is not invariant under the gauge transformations of the background field \eqref{ya-p-1}, and the weak deformation from Section \ref{ym:sec:re} does not allow such symmetry to be restored. The authors expect that at the level of the second approximation, the classical action will split, see for reference Comment 1 for Theorem \ref{ya-t1}, which will significantly complicate further calculations.

\subsection{Acknowledgements}
\vspace{2mm}
\begin{center}
We are grateful to our\,\, \large\mbox{\calligra{K}}\normalsize \,\,\,\,\,\,\,\,and\,\, \small\textbf{\textsl{L}}\normalsize\,\, for inspiration.
\,\,\,\,\,\,\,
\end{center}

\subsection{Statements}

\noindent\textbf{Data Availability Statement.} Data exchange is not applicable to this article because no data sets have been generated or analyzed during the current study.

\vspace{2mm}
\noindent\textbf{Code Availability Statement.} The corresponding code/software is not attached to the article.

\vspace{2mm}
\noindent\textbf{Conflict of Interest Statement.} The authors claim that there is no conflict of interest.

\end{document}